\RequirePackage[l2tabu,orthodox]{nag}
\documentclass
[11pt,letterpaper]
{article} 

\usepackage{subfiles}
\usepackage[notes=false,later=false,camera=false, draft=false]{dtrt}
\usepackage[utf8]{inputenc}
\usepackage{etex}
\usepackage{ stmaryrd }
\usepackage{xspace,enumerate}
\usepackage[T1]{fontenc}
\usepackage[full]{textcomp}
\usepackage[american]{babel}
\usepackage{mathtools}

\usepackage{amsthm}
\usepackage{empheq}

\usepackage{hyperref}
\hypersetup{hyperindex=true,pdfpagemode=UseOutlines,bookmarksnumbered=true,bookmarksopen=true,bookmarksopenlevel=2,pdfstartview=FitH,pdfborder={0 0 1},colorlinks, linkcolor={blue!55!white},
    citecolor={red!65!black},    urlcolor={blue}}
\usepackage[capitalise,nameinlink]{cleveref}
\crefname{lemma}{Lemma}{Lemmas}
\crefname{fact}{Fact}{Facts}
\newcommand{\colorconstraints}{\text{Color Constraints}}
\crefname{colorconstraints}{(color constraints)}{Color Constraints}
\crefformat{colorconstraints}{#2\colorconstraints#3}
\crefname{indsetconstraints}{(indset constraints)}{IndSet Constraints}
\crefformat{indsetconstraints}{#2$\mathsf{IndSet\ Axioms}$#3}
\crefname{theorem}{Theorem}{Theorems}
\crefname{mtheorem}{Theorem}{Theorems}
\crefname{corollary}{Corollary}{Corollaries}
\crefname{claim}{Claim}{Claims}
\crefname{example}{Example}{Examples}
\crefname{algorithm}{Algorithm}{Algorithms}
\crefname{problem}{Problem}{Problems}
\crefname{definition}{Definition}{Definitions}
\usepackage{paralist}
\usepackage{turnstile}
\usepackage{mdframed}
\usepackage{tikz}
\usepackage{caption}
\DeclareCaptionType{Algorithm}
\usepackage{newfloat}
\newtheorem{theorem}{Theorem}[section]
\newtheorem*{theorem*}{Theorem}

\newtheorem*{proposition*}{Proposition}
\newtheorem{lemma}[theorem]{Lemma}
\newtheorem*{lemma*}{Lemma}

\newtheorem*{conjecture*}{Conjecture}
\newtheorem{fact}[theorem]{Fact}
\newtheorem*{fact*}{Fact}

\newtheorem*{hypothesis*}{Hypothesis}

\theoremstyle{definition}
\newtheorem{definition}[theorem]{Definition}
\newtheorem*{definition*}{Definition}

\newtheorem{example}[theorem]{Example}

\newtheorem{claim}[theorem]{Claim}
\newtheorem*{claim*}{Claim}
\newtheorem{remark}[theorem]{Remark}
\newtheorem{groupcase}[theorem]{Remark: Abelian Group Case}
\newtheorem{observation}[theorem]{Observation}
\newtheorem*{observation*}{Observation}
\newtheorem{notation}[theorem]{Notation}

\usepackage[
letterpaper,
top=1in,
bottom=1.1in,
left=1in,
right=1in]{geometry}
\usepackage{newpxtext} 
\usepackage{textcomp} 
\usepackage[varg,bigdelims]{newpxmath}
\usepackage[scr=rsfso]{mathalfa}
\usepackage{bm} 
\let\mathbb\varmathbb
\usepackage{microtype}
\usepackage{footnotebackref}

\newcommand{\wt}{\text{wt}}
\newcommand{\sbra}[1]{{\left[ #1 \right]}}
\newcommand{\cbra}[1]{{\left\{ #1 \right\}}}
\usepackage[most]{tcolorbox}

\definecolor{petergreen}{rgb}{0, 0.75, 0}

\allowdisplaybreaks
\newcommand{\FormatAuthor}[3]{
\begin{tabular}{c}
#1 \\ {\small\texttt{#2}} \\ {\small #3}
\end{tabular}
}

\newcommand{\R}{{\mathbb R}}
\newcommand{\N}{{\mathbb N}}
\newcommand{\norm}[1]{\lVert #1 \rVert}

\newcommand{\abs}[1]{\lvert #1 \rvert}

\newcommand{\eps}{\varepsilon}
\newcommand{\defeq}{\coloneqq}
\newcommand{\F}{{\mathbb F}}

\newcommand{\E}{{\mathbb E}}
\newcommand{\1}{\mathbf{1}}
\newcommand{\ip}[1]{\langle #1 \rangle}

\newcommand{\tr}{\mathrm{tr}}
\newcommand{\Z}{\mathbb Z}

\newcommand{\C}{\mathbb C}
\newcommand{\Bits}{\{0,1\}}

\newcommand{\pE}{\tilde{\E}}

\newcommand{\poly}{\mathrm{poly}}
\newcommand{\val}{\mathrm{val}}

\newcommand{\mper}{\,.}
\newcommand{\mcom}{\,,}

\newcommand{\Id}{\mathbb{I}}

\newcommand{\algval}{\mathrm{alg}\text{-}\mathrm{val}}

\newcommand{\polylog}{\mathrm{polylog}}

\newcommand{\mcH}{\mathcal H}
\newcommand{\mcI}{\mathcal I}
\newcommand{\mcU}{\mathcal U}
\newcommand{\mcV}{\mathcal V}
\newcommand{\chara}{\mathrm{char}}

\let\svthefootnote\thefootnote
\newcommand\blfootnote[1]{%
  \let\thefootnote\relax%
  \footnotetext{#1}%
  \let\thefootnote\svthefootnote%
}

\newcommand{\supp}{\mathrm{supp}}

\newcommand{\NP}{\mathsf{NP}}
\newcommand{\onevec}{\mathbf{1}}

\let\LIN\relax
\DeclareMathOperator{\LIN}{\mathsf{LIN}}

\newcommand{\Tr}{\mathrm{Tr}}
\newcommand{\vardeg}{\textrm{vardeg}}

\begin{document}

\title{Spectral Refutations of Semirandom \texorpdfstring{$k$-$\LIN$}{k-LIN} over Larger Fields}

\author{
\begin{tabular}[h!]{cc}
               \FormatAuthor{Nicholas Kocurek}{nichok6@cs.washington.edu}{University of Washington} &
               \FormatAuthor{Peter Manohar\thanks{This material is based upon work supported by the National Science Foundation
under Grant No.\ DMS-1926686.}}{pmanohar@ias.edu}{The Institute for Advanced Study}
\end{tabular}
} %
\date{}

\maketitle
\thispagestyle{empty}

\begin{abstract}
We study the problem of strongly refuting semirandom $k$-$\LIN(\F)$ instances: systems of $k$-sparse inhomogeneous linear equations over a finite field $\F$. For the case of $\F = \F_2$, this is the well-studied problem of refuting semirandom instances of $k$-XOR, where the works of~\cite{GuruswamiKM22,HsiehKM23} establish a tight trade-off between runtime and clause density for refutation: for any choice of a parameter $\ell$, they give an $n^{O(\ell)}$-time algorithm to certify that there is no assignment that can satisfy more than $\frac{1}{2} + \eps$-fraction of constraints in a semirandom $k$-XOR instance, provided that the instance has $O(n) \cdot \left(\frac{n}{\ell}\right)^{k/2 - 1} \log n /\eps^4$ constraints, and the work of~\cite{KothariMOW17} provides good evidence that this tight up to a $\polylog(n)$ factor via lower bounds for the Sum-of-Squares hierarchy. 
However for larger fields, the only known results for this problem are established via black-box reductions to the case of $\F_2$, resulting in an $\abs{\F}^{3k}$ gap between the current best upper and lower bounds.

In this paper, we give an algorithm for refuting semirandom $k$-$\LIN(\F)$ instances with the ``correct'' dependence on the field size $\abs{\F}$. For any choice of a parameter $\ell$, our algorithm runs in $(\abs{\F}n)^{O(\ell)}$-time and strongly refutes semirandom $k$-$\LIN(\F)$ instances with at least $O(n) \cdot \left(\frac{\abs{\F^*} n}{\ell}\right)^{k/2 - 1} \log(n \abs{\F^*}) /\eps^4$ constraints. We give good evidence that this dependence on the field size $\abs{\F}$ is optimal by proving a lower bound for the Sum-of-Squares hierarchy that matches this threshold up to a $\polylog(n \abs{\F^*})$ factor. Our results also extend beyond finite fields to the more general case of $\Z_{m}$ and arbitrary finite Abelian groups. Our key technical innovation is a generalization of the ``$\F_2$ Kikuchi matrices'' of~\cite{WeinAM19,GuruswamiKM22} to larger fields, and finite Abelian groups more generally.
\end{abstract}

\clearpage
 \microtypesetup{protrusion=false}
 {
  \hypersetup{hidelinks}
  \tableofcontents{}
}
  \microtypesetup{protrusion=true}

\thispagestyle{empty}
\clearpage

\pagestyle{plain}
\setcounter{page}{1}

\section{Introduction}
\label{sec:intro}
A $k$-$\LIN(\F)$ instance over a finite field $\F$ is a collection of $k$-sparse $\F$-linear inhomogeneous equations in $n$ variables $x_1, \dots, x_n$. That is, each equation has the form $\sum_{i \in I} \alpha_i x_i = b_I$, where $\abs{I} = k$, $\alpha_i \in \F \setminus \{0\}$, and $b_I \in \F$. 
For worst-case instances, there has been a long line of work on developing algorithms for and proving hardness of determining the value, i.e., the maximum fraction of satisfiable constraints, of a $k$-$\LIN(\F)$ instance, with a special focus on the case of $\F = \F_2$, where the $k$-$\LIN(\F_2)$ problem is commonly referred to as ``$k$-XOR''.
For $k$-$\LIN(\F)$ instances with $O(n)$ constraints, H{\aa}stad's PCP~\cite{Hastad01} shows that for $k \geq 3$, it is $\NP$-hard to decide if the instance has value $\leq \frac{1}{\abs{\F}} + \eps$ (a random assignment is near-optimal) or $\geq 1 - \eps$ (nearly satisfiable). On the algorithmic side, the work of~\cite{AroraKK95} gives a PTAS for $k$-XOR instances with $n^{O(k)}$ constraints (``maximally dense'' instances). And, assuming the exponential time hypothesis~\cite{ImpagliazzoP01}, the work of~\cite{FotakisLP16} gives an essentially tight ``runtime vs.\ number of constraints'' trade-off for worst-case instances. For the case of $k = 2$, the $2$-$\LIN(\F)$ problem is closely related to the Unique Games Conjecture~\cite{Khot02}, which conjectures that deciding if a $2$-$\LIN(\F)$ instance has value $\leq \frac{1}{\abs{\F}} + \eps$ or $\geq 1 - \eps$ is $\NP$-hard when $\abs{\F}$ is sufficiently large as a function of $\eps$.

Despite its worst-case hardness, there have been many works on designing algorithms for \emph{random} $k$-XOR instances. As random $k$-XOR instances with $\gg n$ constraints have value $\leq \frac{1}{2} + \eps$ with high probability, the natural problem to consider is the task of (strong) \emph{refutation}, where the algorithmic goal is to output a certificate that the value of the random instance is at most $\frac{1}{2} + \eps$. This problem has been the focus of several works~\cite{GoerdtL03,Coja-OghlanGL07,AllenOW15,RaghavendraRS17}, with the ultimate goal being to understand the trade-off between ``runtime'' and ``number of constraints''. That is, given $n$ and a ``time budget'' of $n^{O(\ell)}$ for a parameter $\ell$ (which may be super-constant, e.g., $n^{\delta}$ for constant $\delta > 0$), how many constraints, as a function of $n$ and $\ell$, are required to refute a random $k$-XOR in $n^{O(\ell)}$ time? This question was essentially answered by~\cite{RaghavendraRS17}, which gives for any $\ell$, an $n^{O(\ell)}$-time algorithm that, with high probability over a random instance, certifies that a random $k$-XOR instance has maximum value at most $\frac{1}{2} + \eps$ if it has at least $n \cdot (n/\ell)^{k/2 - 1} \poly(\log(n), \eps^{-1})$ constraints. This trade-off between runtime and number of constraints is conjectured to be essentially optimal, with evidence coming in the form of lower bounds in various restricted computational models~\cite{Grigoriev01,Feige02,Schoenebeck08,BenabbasGMT12,ODonnellW14,MoriW16,BarakCK15,KothariMOW17}. More recently, there has been a flurry of work~\cite{AbascalGK21,GuruswamiKM22,HsiehKM23} on designing algorithms for $k$-XOR in the harder \emph{semirandom} model, where the ``left-hand sides'' of the equations are worst case, and only the ``right-hand sides'' $b_I$ are chosen at random. These works show that one can refute semirandom instances at the same runtime vs.\ number of constraints trade-off as shown in~\cite{RaghavendraRS17}. That is, semirandom instances are just as easy to refute as fully random ones.

Thus, for the Boolean case of $\F = \F_2$, we have a near-complete understanding: for any choice of the parameter $\ell$, if the number of constraints in the semirandom $k$-XOR instance is at least $n \cdot (n/\ell)^{k/2 - 1} \poly(\log(n), \eps^{-1})$, then the algorithm of~\cite{AbascalGK21,GuruswamiKM22,HsiehKM23} can refute the instance in $n^{O(\ell)}$ time, and if the number of constraints is smaller than $n \cdot (n/\ell)^{k/2 - 1} \poly(\log(n), \eps^{-1})$, the lower bound of, e.g.,~\cite{KothariMOW17} provides good evidence that there is no algorithm to refute in $n^{O(\ell)}$ time, even for random instances.

\parhead{The case of larger fields.} What can we say about semirandom $k$-$\LIN$ over finite fields $\F \ne \F_2$? There is a simple reduction to the Boolean case (see Appendix B in~\cite{AllenOW15}) that loses an $\abs{\F}^{3k}$ factor in the number of constraints. That is, the reduction gives an algorithm to refute such instances with at least $\abs{\F}^{3k} \cdot n \cdot (n/\ell)^{k/2 - 1} \poly(\log(n), \eps^{-1})$ constraints.
On the side of lower bounds, the work~\cite{KothariMOW17} establishes the same lower bound as in the case of $\F_2$, i.e., with no field-dependent factor. Thus, for e.g., $\F$ with $\abs{\F} = n^{\delta}$, there a large gap between the upper and lower bounds.

Understanding the optimal dependence on the field size for refuting semirandom $k$-$\LIN$ instances has many potential applications. One immediate application is obtaining better attacks on the (sparse) learning parity with noise (LPN) assumption commonly used in cryptography. The $k$-sparse LPN assumption\footnote{The phrase ``sparse LPN'' sometimes refers to the case when the secret is sparse. Here, we mean that the \emph{equations} are sparse.} is the distinguishing variant of the refutation problem for $k$-$\LIN(\F)$ --- the ``dense'' LPN assumption removes the sparsity requirement on the equations --- and is considered a foundational assumption in cryptography. While many cryptographic applications only use this assumption over the field $\F_2$~\cite{ApplebaumBW10}, many applications such as constructing indistinguishability obfuscation~\cite{JainLS21,JainLS22,RagavanVV24} and others~\cite{DaoIJL23,CorriganGibbsHKV24} require fields of much larger (even superpolynomial) size.

As another application, one could hope to prove stronger lower bounds for information-theoretic private information retrieval (PIR) schemes. Recent work of~\cite{AlrabiahGKM23} has led to a flurry of improvements in lower bounds for \emph{binary} locally decodable~\cite{AlrabiahGKM23,BasuHKL24,JanzerM24} and locally correctable codes~\cite{KothariM23,Yankovitz24,AlrabiahG24,KothariM24} by establishing a connection between these lower bounds and refuting ``semirandom-like'' instances of $k$-$\LIN$ over $\F_2$. While these results can be extended to larger, constant-sized alphabets (see, e.g., Appendix A in~\cite{AlrabiahGKM23}), information-theoretic PIR schemes are essentially equivalent to locally decodable codes over alphabets of $\poly(n)$ size. Thus, the large loss in $\abs{\F}$ above prevents the approach of~\cite{AlrabiahGKM23} from being able to prove stronger PIR lower bounds.

\subsection{Our results}
 In this paper, we investigate the dependence on the field size in the number of constraints required to refute semirandom $k$-$\LIN(\F)$ instances. As our main results, we give both an algorithm and a matching Sum-of-Squares lower bound with the ``correct'' polynomial dependence on the field size $\abs{\F}$.

Before stating our main results, we formally define semirandom $k$-$\LIN$ instances.
\begin{definition}[(Semirandom) $k$-$\LIN$ over $\F$]
\label{def:semirandomklin}
    An instance of $k$-$\LIN(\F)$ is $\mcI = (\mcH, \{b_v\}_{v \in \mcH})$, where $\mcH$ is a set of $k$-sparse vectors\footnote{A vector $v \in \F^n$ is $k$-sparse if $\abs{\{i : v_i \ne 0\}} = k$.} in $\F^n$ and $b_v \in \F$ for all $v \in \mcH$. We view $\mcI$ as representing the system of linear equations with variables $x_1, \dots, x_n$ specified by $\sum_{i=1}^n v_i x_i = b_v$ for each $v \in \mcH$. The value of the instance, which we denote by $\val(\mcI)$, is the maximum over $x \in \F^n$ of the fraction of constraints satisfied by $x$. That is, $\val(\mcI) = \max_{x \in \F^n} \frac{1}{\abs{\mcH}} \sum_{v \in \mcH} \1(\sum_{i=1}^n v_i x_i = b_v)$.
    
 An instance of $k$-$\LIN$ is \emph{random} if $\mcH$ is a random subset of $k$-sparse vectors and each $b_v$ is drawn independently and uniformly from $\F$.    
    
 An instance of $k$-$\LIN$ is \emph{semirandom} if each $b_v$ is drawn independently and uniformly from $\F$ (but $\mcH$ may be arbitrary).
\end{definition}

The first main result of this paper gives a refutation algorithm for semirandom $k$-$\LIN(\F)$ for any field $\F$.
\begin{theorem}[Tight refutation of semirandom $k$-$\LIN(\F)$]
    \label{thm:refutation}
    Fix $\ell \geq k/2$. There is an algorithm that takes as input a $k$-$\LIN(\F)$ instance $\mcI = (\mcH, \{b_v\}_{v \in \mcH})$ in $n$ variables and outputs a number $\algval(\mcI) \in [0,1]$ in time $(\abs{\F} n)^{O(\ell)}$ with the following two guarantees:
    \begin{enumerate}
        \item \label{item:refutation1} $\algval(\mcI) \geq \val(\mcI)$ for every instance $\mcI$;
        \item \label{item:refutation2} If $\abs{\mcH} \geq O(n) \cdot \left(\frac{n\abs{\F^*}}{\ell}\right)^{k/2-1} \cdot \log(\abs{\F^*}n) \cdot \varepsilon^{-4}$ and $\mcI$ is drawn from the semirandom distribution described in \cref{def:semirandomklin}, then with probability $\geq 1-\frac{1}{\poly(n)}$ over the draw of the semirandom instance, i.e., the randomness of $\{b_v\}_{v \in \mcH}$, it holds that $\algval(\mcI) \leq \frac{1}{\abs{\F}} + \varepsilon$.
    \end{enumerate}
\end{theorem}
As a byproduct of the analysis of \cref{thm:refutation}, we also establish an extremal combinatorics statement on the existence of short linear dependencies in any sufficiently dense collection of $k$-sparse vectors $\mcH$ over a finite field $\F$.

\begin{theorem}[Short linear dependencies in $k$-sparse vectors over $\F$]
    \label{thm:feige}
     Let $\mcH$ be a set of $k$-sparse vectors in $\F^n$ with $\abs{\mcH} \geq O(n) \cdot  \left(\frac{n\abs{\F^*}}{\ell}\right)^{k/2-1} \cdot \log(\abs{\F^*}n)$. Then, there exists a set $\mcV \subseteq \mathcal{H}$ with $\abs{\mcV} \leq \ell \log (\abs{\F^*} n)$ and non-zero coefficients $\cbra{\alpha_v}_{v \in \mcV}$ in $\F^*$ such that:
     \begin{equation*}
         \sum_{v \in \mcV} \alpha_v \cdot v = 0\mper
     \end{equation*}
     That is, $\mcV$ is a linearly dependent subset of $\mathcal{H}$.
\end{theorem}
\cref{thm:feige} is a generalization of the hypergraph Moore bound, or Feige's conjecture on the existence of short even covers in hypergraphs~\cite{Feige08} (first proven in~\cite{GuruswamiKM22} for the case of $\F_2$) to arbitrary finite fields. The hypergraph Moore bound establishes  a rate vs.\ distance trade-off for binary LDPC codes (see~\cite{NaorV08}). One can similarly view \cref{thm:feige} as establishing such a trade-off for LDPC codes over general finite fields.

The key technical innovation in our proofs of \cref{thm:refutation,thm:feige} is the introduction of a new Kikuchi matrix for any finite field $\F$ (\cref{def:kikuchimatrix}). Our Kikuchi matrices are a generalization of the ``$\F_2$ Kikuchi matrices'' of~\cite{WeinAM19,GuruswamiKM22} to other fields,  and finite Abelian groups more generally. As we point out after \cref{def:kikuchimatrix}, the natural generalization of the $\F_2$ Kikuchi matrix is not sufficient, and we have to add an additional condition to the matrix to make the analysis work (see \cref{rem:nontrivialkikuchi}).

In our second main result, we prove a Sum-of-Squares lower bound for refuting $k$-$\LIN(\F)$ instances that nearly matches the threshold in \cref{thm:refutation}.

\begin{theorem}[Sum-of-Squares lower bounds for refuting random $k$-$\LIN$, informal]
    \label{thm:infsoslowerbound}
    Fix $k \geq 3$ and $\frac{n}{\max(\abs{\F^*}, k)} \geq \ell \geq k$. Let $\mcI$ be a random $k$-$\LIN(\F)$ instance with $\abs{\mcH} \leq \Omega(n) \cdot \left(\frac{n\abs{\F^*}}{\ell}\right)^{k/2-1} \cdot \varepsilon^{-2}$. Then, with large probability over the draw of $\mcI$, it holds that
    \begin{enumerate}
        \item $\val(\mcI) \leq \frac{1}{\abs{\F}} + \varepsilon$;
        \item The degree-$\tilde{\Omega}(\ell)$ Sum-of-Squares relaxation for $k$-$\LIN(\F)$ fails to refute $\mcI$.
        \end{enumerate}
\end{theorem}
We note that \cref{thm:infsoslowerbound} requires $k \geq 3$ as \cref{thm:refutation} gives a polynomial-time algorithm when $\abs{\mcH} = O(n)$ and $k = 2$.

The threshold in \cref{thm:infsoslowerbound} matches the threshold in \cref{thm:refutation} up to the ``lower order'' $\poly(\log n, \eps^{-1})$ factor. However, it has a one limitation: the degree $\ell$ of the Sum-of-Squares relaxation can only be at most $n/\abs{\F^*}$, rather than the whole of range of $O(1)$ to $\Omega(n)$. That is, \cref{thm:infsoslowerbound} can only give a subexponential-time lower bound of $2^{n^{1 - \delta}}$ for the Sum-of-Squares algorithm when $\abs{\F} = n^{\delta}$. It turns out that this is nearly necessary, as there is a very simple refutation algorithm, implementable within degree $\tilde{O}(\ell)$ Sum-of-Squares, that out-performs the algorithm in \cref{thm:refutation} when $\abs{\F}$ is very large.

\begin{theorem}[Simple refutation algorithm]
\label{thm:trivialalgorithm}
    Fix $n/2 \geq \ell \geq k$. There is an algorithm that takes as input a $k$-$\LIN(\F)$ instance $\mcI = (\mcH, \{b_v\}_{v \in \mcH})$ in $n$ variables and outputs a number $\algval(\mcI) \in [0,1]$ in time $(\abs{\F} n)^{O(\ell)}$ with the following two guarantees:
    \begin{enumerate}
        \item \label{item:simplerefutation1} $\algval(\mcI) \geq \val(\mcI)$ for every instance $\mcI$;
        \item \label{item:simplerefutation2} If $\abs{\mcH} \geq O(n) \cdot \left(\frac{n}{\ell}\right)^{k - 1}  \cdot \log (\abs{\F}n) \cdot \varepsilon^{-2}$ and $\mcI$ is drawn from the fully random distribution described in \cref{def:semirandomklin}, then with probability $\geq 1-\frac{1}{\poly(n)}$ over the draw of the random instance, it holds that $\algval(\mcI) \leq \frac{1}{\abs{\F}} + \varepsilon$;
        \item \label{item:simplerefutation3} If $\abs{\mcH} \geq O(n) \cdot \left(\frac{n}{\ell}\right)^{k - 1}  \cdot \log (\abs{\F}n) \cdot \varepsilon^{-3}$ and $\mcI$ is drawn from the semirandom distribution described in \cref{def:semirandomklin}, then with probability $\geq 1-\frac{1}{\poly(n)}$ over the draw of the semirandom instance, i.e., the randomness of $\{b_v\}_{v \in \mcH}$, it holds that $\algval(\mcI) \leq \frac{1}{\abs{\F}} + \varepsilon$.
    \end{enumerate}
\end{theorem}

\cref{thm:trivialalgorithm} is nearly identical to \cref{thm:refutation}: the main difference is that, ignoring the lower order $\poly(\log n, \eps^{-1})$ term, the ``constraint threshold'' is now $n \cdot (n/\ell)^{k-1}$ instead of $n \cdot (n \abs{\F^*}/\ell)^{k/2 - 1}$, and this is smaller when $\ell \geq n/\abs{\F^*}^{1 - 2/k}$. As this algorithm is ``captured'' by the Sum-of-Squares hierarchy, the Sum-of-Squares lower bound in \cref{thm:infsoslowerbound} is false when $\ell \geq n/\abs{\F^*}^{1 - 2/k}$. Thus, the range of $\ell$ where the lower bound in \cref{thm:infsoslowerbound} holds is nearly tight: we show it holds for $\ell \leq n/\abs{\F^*}$, and it is false for $\ell \geq n/\abs{\F^*}^{1 - 2/k}$.

\parhead{Extensions to finite Abelian groups.}
We also extend our refutation result for $k$-$\LIN(\F)$ to $k$-$\LIN(\Z_m)$ for composite $m$, and more generally to any finite Abelian group $G$. Below, we define the natural extension of \cref{def:semirandomklin} to Abelian groups, and then state our generalization of \cref{thm:refutation} to this setting.  Recall that by the fundamental theorem of finite Abelian groups, any finite Abelian group $G$ is isomorphic to $\bigotimes_{i=1}^r \Z_{m_i}$ for some $m_1,..., m_r \in \N$. 

\begin{definition}[(Semirandom) $k$-$\LIN$ over an Abelian group $G$]
\label{def:groupsemirandomklin}
    Let $G = \bigotimes_{i=1}^r \Z_{m_i}$ for some $m_1,..., m_r \in \N$. An instance of $k$-$\LIN(G)$ is $\mcI = (\mcH, \{b_v\}_{v \in \mcH})$, where $\mcH$ is a set of $k$-sparse vectors in $G^n$ and $b_v \in G$ for all $v \in \mcH$. We view $\mcI$ as representing the system of linear equations with variables $x_1, \dots, x_n$ specified by $\ip{v,x} = b_v$ for each $v \in \mcH$. Note that the inner product notation here represents
    \begin{equation*}
        \ip{v, x} = \sum_{i=1}^n v_i \cdot x_i\mcom
    \end{equation*}
    where the multiplication is direct product multiplication over each $\Z_{m_i}$. The value of the instance, which we denote by $\val(\mcI)$, is the maximum over $x \in \F^n$ of the fraction of constraints satisfied by $x$. That is, $\val(\mcI) = \max_{x \in G^n} \frac{1}{\abs{\mcH}} \sum_{v \in \mcH} \1(\ip{x,v} = b_v)$.

     An instance of $k$-$\LIN$ is \emph{random} if $\mcH$ is a random subset of $k$-sparse vectors and each $b_v$ is drawn independently and uniformly from $G$.    
    
    An instance of $k$-$\LIN$ is \emph{semirandom} if each $b_v$ is drawn independently and uniformly from $G$ (but $\mcH$ may be arbitrary).
\end{definition}

\cref{def:groupsemirandomklin} allows each equation to have coefficients on the variables, which was natural in the case of finite fields (\cref{def:semirandomklin}), but may appear strange in the case of a general finite Abelian group, where it is perhaps more natural to only have coefficients that are $1$. However, because we are working with semirandom instances, the ``left-hand sides'' of the equations are arbitrary, and so semirandom instances ``capture'' the special case where the coefficients are all $1$. We choose this perhaps nonstandard definition because it is more general; it seamlessly captures both the case of a  finite field $\F$ or a ring $\Z_m$, where coefficients are natural, and also a finite Abelian group.

Our final result generalizes \cref{thm:refutation} to the case of $k$-$\LIN(G)$.

\begin{theorem}[Tight refutation of semirandom $k$-$\LIN(G)$]
    \label{thm:grouprefutation}
    Fix $\ell \geq k/2$. There is an algorithm that takes as input a $k$-$\LIN(G)$ instance $\mcI = (\mcH, \{b_v\}_{v \in \mcH}$ in $n$ variables and outputs a number $\algval(\mcI) \in [0,1]$ in time $(\abs{G} n)^{O(\ell)}$ with the following two guarantees:
    \begin{enumerate}
        \item \label{item:grouprefutation1} $\algval(\mcI) \geq \val(\mcI)$ for every instance $\mcI$;
        \item \label{item:grouprefutation2} If $\abs{\mcH} \geq O(n) \cdot  \left(\frac{n\abs{G}}{\ell}\right)^{k/2-1} \cdot \log(\abs{G}n) \cdot \varepsilon^{-5}$ and $\mcI$ is drawn from the semirandom distribution described in \cref{def:groupsemirandomklin}, then with probability $\geq 1-\frac{1}{\poly(n)}$ over the draw of the semirandom instance, i.e., the randomness of $\{b_v\}_{v \in \mcH}$, it holds that $\algval(\mcI) \leq \frac{1}{\abs{G}} + \varepsilon$.
    \end{enumerate}
\end{theorem}
The proof of \cref{thm:grouprefutation} encounters additional technical difficulties compared to \cref{thm:refutation}, arising from zero divisors in $\Z_m$ for composite $m$. Roughly, this allows a semirandom instance to embed equations within a subgroup of $\Z_m$, by, e.g., choosing only equations where the coefficients are divisible by some integer $d \geq 2$, and handling this issue requires an additional step and an extra factor $\varepsilon^{-1}$.

\parhead{Organization.}
The rest of the paper is organized as follows. In \cref{sec:prelims}, we introduce notation and standard facts that we use. In \cref{sec:refutation}, we prove \cref{thm:refutation} for even $k$. This section also serves as a proof overview and warmup for the odd $k$ and general group case. In \cref{sec:grouprefutation}, we prove \cref{thm:grouprefutation} in the even case. In 
\cref{sec:mainrefutation}, we prove \cref{thm:refutation} and \cref{thm:grouprefutation} for odd $k$. In \cref{sec:sos}, we prove \cref{thm:infsoslowerbound}, and finally in \cref{sec:trivialalgorithm}, we prove \cref{thm:trivialalgorithm}.

\section{Preliminaries}
\label{sec:prelims}

\subsection{Notation}
We let $[n]$ denote the set $\{1, \dots, n\}$. For two subsets $S, T \subseteq [n]$, we let $S \oplus T$ denote the symmetric difference of $S$ and $T$, i.e., $S \oplus T \coloneqq \{i : (i \in S \wedge i \notin T) \vee (i \notin S \wedge i \in T)\}$. For a natural number $t \in \N$, we let ${[n] \choose t}$ be the collection of subsets of $[n]$ of size exactly $t$. 

For a rectangular matrix $A \in \C^{m \times n}$, we let $\norm{A}_2 \coloneqq \max_{x \in \C^m, y \in \C^n: \norm{x}_2 = \norm{y}_2 = 1} x^{\dagger} A y$ denote the spectral norm of $A$.

For a vector $v\in \F^n$, we let $\supp(v) \defeq \{i : v_i \ne 0\}$ and $\wt(v) \defeq \abs{\supp(v)}$. For a field $\F$ with $\chara(\F) = p$, we let $\Tr(\cdot)$ denote the trace map of $\F$ over $\F_p$. We denote for two vectors $u, v \in \F^n$ the set of indices in $\supp(u) \cap \supp(v)$ such that $u_i = v_i$ as the agreement of $u$ and $v$, or simply $u \sqcap v$. Analogously we say $u \sqsubseteq v$ if $u \sqcap v = \supp(u)$.

For a matrix $A \in \C^{n \times n}$, we let $\tr(A)$ be the trace of $A$, i.e., $\sum_{i = 1}^n A_{i,i}$. This should not be confused with the trace map for field elements, which we denote by $\Tr(\cdot)$.
For two vectors $x, y \in \C^n$ we define the following inner product:
\begin{equation*}
    \langle x, y \rangle = x^\dagger y = \sum_{i = 1}^n \overline{x_i} \cdot y_i\mper
\end{equation*}

\subsection{Fourier analysis}

Let $G$ be an Abelian group isomorphic to $\Z_{m_1} \times ... \times \Z_{m_r}$ via the isomorphism $\psi$. For $m \in \N$, we let $\omega_m := e^{\frac{2\pi i}{m}}$. For $\alpha, x \in G$, we define
\begin{equation*}
    \chi_\alpha(x) = \prod_{i=1}^r \omega_{m_i}^{\psi(\alpha)_i \psi(x)_i}\mper
\end{equation*}
These functions form a Fourier basis for $G$, as shown in \cite{ODonnell14}. This extends to a Fourier basis for $G^n$ as follows. For $v, x \in G^n$, we define
\begin{equation*}
    \chi_v(x) = \prod_{i=1}^n \chi_{v_i}(x_i)\mper
\end{equation*}
For a function $f \colon G^n \to \C$, we have that for each $x \in G^n$,
\begin{equation*}
    f(x) = \sum_{v \in G^n} \hat{f}(v) \cdot \chi_v(x)\mcom
\end{equation*}
where $\hat{f}(v) = \E_{x \in G^n}\sbra{f(x) \cdot \overline{\chi_v(x)}}$. 

For the special case of functions $f \colon \F^n \to \C$ with $\textrm{char}(\F) = p$, we note that the standard Fourier basis is simply
\begin{equation*}
    \chi_v(x) = \omega_p^{\Tr(\langle v, x \rangle)}\mper
\end{equation*}

\subsection{Binomial coefficient inequalities}

\begin{fact}
\label{fact:binomest}
Let $n, \ell, q$ be positive integers with $\ell \leq n$. Let $q$ be constant and $\ell, n$ be asymptotically large with $\ell \leq n/2$. Then, 
\begin{flalign*}
&\frac{ {n \choose \ell - q}}{ {n \choose \ell }} = \Theta\left(\left(\frac{\ell}{n}\right)^q\right) \mcom \\
&\frac{ {n - q \choose \ell} }{{n \choose \ell}} = \Theta(1) \mper
\end{flalign*}
\end{fact}
\begin{proof}
We have that
\begin{flalign*}
&\frac{ {n \choose \ell - q}}{ {n \choose \ell }} = \frac{ { \ell \choose q}}{ {n - \ell + q \choose q}} \mper
\end{flalign*}
Using that $\left(\frac{a}{b} \right)^{b} \leq {a \choose b} \leq \left(\frac{e a}{b} \right)^{b}$ finishes the proof of the first equation.

We also have that
\begin{flalign*}
&\frac{ {n - q \choose \ell} }{{n \choose \ell}} = \frac{(n - q)! (n - \ell)!}{n! (n - \ell - q)!} = \prod_{i = 0}^{q -1 } \frac{n - \ell - i}{n - i} = \prod_{i = 0}^{q-1} \left(1 - \frac{\ell}{n - i}\right) \mcom
\end{flalign*}
and this is $\Theta(1)$ since $\ell \leq n/2$ and $q$ is constant.
\end{proof}

\subsection{The Sum-of-Squares algorithm}
\label{sec:sum-of-squares}
We briefly define the key Sum-of-Squares facts that we use. These facts are all taken from~\cite{BarakS16,FlemingKP19}.
\begin{definition}[Pseudo-expectations over the hypercube] \label{def:pseudo-expectation}
A degree $d$ pseudo-expectation $\pE$ over $\Bits^n$ is a linear operator that maps degree $\leq d$ polynomials on $\Bits^n$ into real numbers with the following three properties:
\begin{enumerate}
    \item (Normalization) $\pE_\mu[1] = 1$.
	\item (Booleanity) For any $x_i$ and any polynomial $p$ of degree $\leq d-2$, $\pE_\mu[p(x_i^2 - x_i)] = 0$. 
	\item (Positivity) For any polynomial $p$ of degree at most $d/2$, $\pE_\mu[p^2] \geq 0$. 
\end{enumerate} 
\end{definition}
We note that if $\E$ is the expectation operator of a distribution over $\Bits^n$, then $\E$ is a degree $d$ pseudo-expectation (for any $d$), and thus $\max_{x \in \Bits^n} f(x) \leq \max_{\pE} \pE[f]$, where the second max is taken over all degree $d$ pseudo-expectations $\pE$.

The Sum-of-Squares algorithm shows that we can efficiently maximize $\pE[f]$ over degree $d$ pseudo-expectations $\pE$ for a polynomial $f$. 
\begin{fact}[Sum-of-Squares algorithm, Corollary 3.40 in \cite{FlemingKP19}]
\label{fact:sosalg}
Let $f(x_1, \dots, x_n)$ be a polynomial of degree $k$, where the coefficients of $f$ are rational numbers with $\poly(n)$ bit complexity. Let $d \geq k$. There is an algorithm that, on input $f,d$, runs in time $n^{O(d)}$ and outputs a value $\alpha$ such that $\beta + 2^{-n} \geq \alpha \geq \beta$, where $\beta$ is the maximum, over all degree $d$ pseudo-expectations $\pE$ over $\Bits^n$, of $\pE[f]$. 
\end{fact}


\section{Refuting Semirandom \texorpdfstring{$k$-$\LIN$}{k-LIN} over Fields for Even $k$}
\label{sec:refutation}

In this section, we give a complete proof of \cref{thm:refutation} in the case when $k$ is even. As in~\cite{GuruswamiKM22, HsiehKM23}, the proof is substantially simpler in the case of even $k$, so this section serves as a warmup to the proof for odd $k$, which we do in \cref{sec:mainrefutation}.

Our refutation algorithm for semirandom $k$-$\LIN$ roughly follows the framework established in \cite{GuruswamiKM22, HsiehKM23}. The main technical tool we use is a generalization of the Kikuchi matrix of~\cite{WeinAM19} for $\F_2$ to arbitrary finite fields $\F$. Analyzing the spectral norm of this matrix requires a more complicated trace moment calculation as compared to the case of $\F_2$, and requires a careful choice of the Kikuchi matrix (see \cref{rem:nontrivialkikuchi}).

We let $\textrm{char}(\F) = p$ and $\omega_p = e^{2 \pi i/p}$ denote a primitive $p$-th root of unity in $\C$.

\subsection{Step 1: Expressing a \texorpdfstring{$k$-$\LIN(\F)$}{k-LIN(F)} instance as a polynomial in $\C$} 
As the first step in the proof, we make the following observation, which shows that we can express the fraction of constraints satisfied by an assignment $x \in \F^n$ as a polynomial in $n$ variables in $\C$.
\begin{observation}
    \label{obs:advantageoverpolynomial}
    For a $k$-$\LIN(\F)$ instance $\mcI = (\mcH, \{b_v\}_{v \in \mcH})$, let $\val(\mcI, x)$ denote the fraction of constraints satisfied by an assignment $x \in \F^n$. Then, we can express $\val(\mcI, x)$ as a polynomial in $\C$. That is,
    \begin{equation*}
        \val(\mcI, x) = \frac{1}{\abs{\F}} + \frac{1}{\abs{\mcH} \abs{\F}}\sum_{v \in \mcH} \sum_{\beta \in \F^*} \omega_p^{\Tr(\beta b_v)} \cdot \overline{\chi_{\beta v}(x)} := \frac{1}{\abs{\F}} + \Phi(x) \mper
    \end{equation*}
\end{observation}
\begin{proof}
    Recall that a constraint in $\mcI$ takes the form $\langle v, x \rangle = b_v$ for $v \in \mcH$, where $x \in \F^n$ are the variables. The indicator variable for this event is simply:
    \begin{equation*}
        \1(\langle v, x \rangle = b_v) = \E_{\beta \sim \F}\sbra{\omega_{p}^{\Tr\left(\beta b_v - \beta \langle v, x \rangle\right)}} = \frac{1}{\abs{\F}}\sum_{\beta \in \F} \omega_p^{\Tr(\beta b_v)} \cdot \overline{\chi_{\beta v}(x)} \mper
    \end{equation*}
    where $p = \chara(\F)$.
    Indeed, if $\langle v, x \rangle = b_v$, then $\Tr(\beta b_v - \beta\ip{v,x}) = 0$ for all $\beta \in \F$. If $b_v - \ip{v,x} \ne 0$, i.e., it is some $\alpha \in \F^*$, then $\E_{\beta \sim \F}\sbra{\omega_p^{\Tr(\beta \alpha)}} = \E_{\beta \sim \F}\sbra{\omega_p^{\Tr(\beta)}} = 0$. Hence, it follows that
    \begin{flalign*}
    &\val(\mcI, x) = \frac{1}{\abs{\mcH}} \sum_{v \in \mcH}  \1(\langle v, x \rangle = b_v) = \frac{1}{\abs{\mcH}} \sum_{v \in \mcH}  \frac{1}{\abs{\F}} \sum_{\beta \in \F} \omega_p^{\Tr(\beta b_v)} \cdot \overline{\chi_{\beta v}(x)} \\
    &= \frac{1}{\abs{\F}} + \frac{1}{\abs{\mcH} \abs{\F}} \sum_{v \in \mcH} \sum_{\beta \in \F^*} \omega_p^{\Tr(\beta b_v)} \cdot \overline{\chi_{\beta v}(x)}  \mcom   \end{flalign*}
    which finishes the proof.
\end{proof}

\subsection{Step 2: Expressing $\Phi(x)$ as a quadratic form on a Kikuchi matrix}
In light of \cref{obs:advantageoverpolynomial}, it thus remains to find a certificate that bounds $\max_{x \in \F^n} \Phi(x)$. We do this by generalizing the analysis of~\cite{GuruswamiKM22} and constructing a Kikuchi matrix whose spectral norm provides a certificate bounding the maximum value of $\Phi$.

\begin{definition}{(Even-arity Kikuchi matrix over $\F$).}
\label{def:kikuchimatrix}
Let $k/2 \leq \ell \leq n/2$ be a parameter,\footnote{Note that it suffices to prove \cref{thm:refutation} for $\ell$ in this range.} and let $N = \abs{\F^*}^\ell {n \choose \ell}$. For each $k$-sparse vector $v \in \F^n$ and $\beta \in \F^*$, we define a matrix $A_{v, \beta} \in \C^{N \times N}$ as follows. First, we identify $N$ with the set of $\ell$-sparse vectors in $\F^n$. Then, for $\ell$-sparse vectors $U, V \in \F^n$, we let
    \begin{equation*}
        A_{v, \beta}(U, V) = \begin{cases}
                      1  & U\xrightarrow{\text{$v, \beta$}} V\\
                      0 & \text{otherwise}
                    \end{cases}
    \end{equation*}
    where we say $U\xrightarrow{\text{$v, \beta$}} V$ if $U-V = \beta v$ and $\supp(U) \oplus \supp(V) = \supp(v)$.
    
Let $\Phi(x) = \frac{1}{\abs{\F} \abs{\mcH}}\sum_{v \in \mcH} \sum_{\beta \in \F^*} c_{v, \beta} \cdot \chi_{\beta v}$ be a polynomial defined by a set $\mcH$ of $k$-sparse vectors from $\F^n$ and complex coefficients $\cbra{c_{v, \beta}}_{\substack{v \in \mcH \\ \beta \in \F^*}}$. We define the level-$\ell$ Kikuchi matrix for this polynomial to be $A = \sum_{v \in \mcH} \sum_{\beta \in \F^*} c_{v, \beta} \cdot A_{v, \beta}$. We refer to the graph (with complex weights) defined by the underlying adjacency matrix as the Kikuchi graph.
\end{definition}

\begin{remark}
\label{rem:nontrivialkikuchi}
Our Kikuchi matrix in~\cref{def:kikuchimatrix} has an additional condition that $\supp(U) \oplus \supp(V) = \supp(v)$. A perhaps more natural generalization of the $\F_2$ Kikuchi matrix of \cite{WeinAM19,GuruswamiKM22} would be the matrix where this condition is removed, i.e., we only require that $U - V = \beta v$. As we shall see, when we do the trace moment calculation at the end of \cref{sec:tracemethod}, it is crucial that for any $U$ and $v$, there is at most one $\beta \in \F^*$ such that $U\xrightarrow{\text{$v, \beta$}} V$ for some $V$. That is, the number of edges adjacent to $U$ ``coming from'' a constraint $v$ is at most $1$. If this uniqueness of $\beta$ did not hold, we would lose an additional factor of $\F^*$ in the number of constraints that we require in \cref{thm:refutation}. This is a substantial increase, as e.g., this would increase the number of constraints when $k = 2$ to $\sim n \abs{\F^*}$ when the correct dependence is $\sim n$. The condition that $\supp(U) \oplus \supp(V) = \supp(v)$ implies uniqueness of $\beta$ above, and without this condition we could have a $U$ with $\supp(v) \subseteq U$, in which case $U - \alpha v = V$ for an $\ell$-sparse vector $V$ in $\F^n$ for all $\alpha \in \F$.
\end{remark}

\begin{remark}
We note that in \cref{def:kikuchimatrix}, we have $A_{v, \beta} = A_{\beta v, 1}$. The reason we use the above definition with two parameters $v$ and $\beta$ is that it is more convenient when counting walks in the matrix $A$, as it makes explicit the choice of $v$ and $\beta$. Note that in $\mcH$, there could exist $v$ and $v'$ with $\beta v = v'$ for some $\beta \in \F^*$, and we need to count these terms separately.
\end{remark}

\begin{observation}
    \label{fact:hermitian}
    The Kikuchi matrix $A$ is always Hermitian.
\end{observation}
\begin{proof}
    To see this note that $U - V = \beta v \iff V - U = - \beta v$, $\overline{\chi_\beta} = \chi_{-\beta}$, and $\oplus$ is commutative.
\end{proof}

The following observation shows that we can express $\Phi(x)$ as a quadratic form on the matrix $A$ defined in \cref{def:kikuchimatrix}. Thus, $\norm{A}_2$ bounds $\max_{x \in \F^n} \Phi(x)$.
\begin{observation}
    \label{fact:bilinearforms}
    For $x \in \F^n$ define $y \in \C^N$ as follows. For each $\ell$-sparse $U \in \F^n$, we set $y_U = \chi_U(x)$. Then
    \begin{equation*}
        \Phi(x) = \frac{1}{\abs{\mcH}\abs{\F} \Delta} y^\dagger A y\mcom
    \end{equation*}
    where $\Delta := {k \choose k/2} {n-k \choose \ell - k/2}\abs{\F^*}^{\ell-k/2}$.
\end{observation}

\begin{proof}
    \begin{align*}
        y^\dagger A y &= \sum_{\substack{U, V \in \F^n\\ \wt(U) = \wt(V) = \ell}} A(U, V) \cdot \overline{\chi_U(x)} \cdot \chi_V(x)\\
        &= \sum_{\substack{U, V \in \F^n\\ \wt(U) = \wt(U) = \ell}} \sum_{v \in \mcH, \beta \in \F^*}
        \1\left(U\xrightarrow{\text{$v, \beta$}} V\right) \cdot c_{v, \beta} \cdot \overline{\chi_U(x)} \cdot \chi_V(x)\\
                &= \sum_{\substack{U, V \in \F^n\\ \wt(U) = \wt(U) = \ell}} \sum_{v \in \mcH, \beta \in \F^*} \1\left(U\xrightarrow{\text{$v, \beta$}} V\right) \cdot c_{v, \beta} \cdot \overline{\chi_{U - V}(x)} \\
        &= \sum_{\substack{U, V \in \F^n\\ \wt(U) = \wt(U) = \ell}} \sum_{v \in \mcH, \beta \in \F^*} \1\left(U\xrightarrow{\text{$v, \beta$}} V\right) \cdot c_{v, \beta} \cdot \overline{\chi_{\beta v}(x)} \mper
    \end{align*}
    For each $v \in \mcH$ and $\beta \in \F^*$, the term $c_{v, \beta} \cdot \overline{\chi_{\beta v}(x)}$ appears once for each pair of vertices $(U, V)$ with $U\xrightarrow{\text{$v, \beta$}} V$. Let us now argue that the number of such pairs $(U,V)$ is exactly $\Delta = {k \choose k/2} {n-k \choose \ell - k/2}\abs{\F^*}^{\ell-k/2}$. We count the number of pairs $(U,V)$ by first specifying $\supp(U)$ and $\supp(V)$, and then by specifying $U_i$ for each $i \in \supp(U)$ (and same for $V$). We first require that $\supp(U) \oplus \supp(V) = \supp(v)$, which in turn means that $\supp(U)$ has intersection exactly $k/2$ with $\supp(v)$ and likewise for $\supp(V)$. Thus, we can pay ${k \choose k/2}$ to count the number of ways to split $\supp(v)$ into two equal parts. Second, we need to specify $\supp(U) \setminus \supp(v)$, which is equal to $\supp(V) \setminus \supp(v)$, which is ${n - k\choose \ell - k/2}$ choices. Finally, we need to specify $U_i$ for each $i \in \supp(U)$ and $V_i$ for each $i \in \supp(V)$. For each $i \in \supp(U) \cap \supp(v)$, we set $U_i = (\beta v)_i$, and for each $i \in \supp(U) \setminus \supp(v)$, we can set $U_i$ to be any element in $\F^*$. Note that specifying $U$ then determines $V$, so we have $\abs{\F^*}^{\ell - k/2}$ choices. This finishes the proof.
 \end{proof}

Next, we compute the average degree (or number of non-zero entries) in a row/column in $A$.
\begin{observation}
    \label{fact:avgdegree}
    For $U \in \F^n$ with $\wt(U) = \ell$, we define the graph degree as normal:
    \begin{equation*}
    \deg(U) := \abs{\{\beta v \mid \beta \in \F^*, v \in \mcH \text{ s.t. } \exists V \in \F^n, \wt(V) = \ell, U \xrightarrow{v, \beta} V\}}\mper
    \end{equation*}
    Then $\E[\deg(U)] \geq \frac{\abs{\F^*}}{2} \left({\frac{\ell}{\abs{\F^*}n}}\right)^{k/2} \cdot \abs{\mcH}$.
\end{observation}

\begin{proof}
    Each $v \in \mcH$ contributes $\abs{\F^*}\Delta$ to the total degree, so the average degree is $\E[\deg(U)] = \frac{\abs{\mcH}\abs{\F^*}\Delta}{N}$. We then have:
    \begin{equation*}
    \E[\deg(U)] = \frac{\abs{\F^*}\Delta}{N} \cdot \abs{\mcH} = \frac{\abs{\F^*}^{\ell-k/2+1} {k \choose k/2}{n-k \choose \ell-k/2}}{\abs{\F^*}^\ell {n \choose \ell}} \cdot \abs{\mcH} \geq \frac{\abs{\F^*}}{2} \left({\frac{\ell}{\abs{\F^*}n}}\right)^{k/2} \cdot \abs{\mcH} \mcom
    \end{equation*}
    where the last inequality follows from \cref{fact:binomest}.
\end{proof}

\subsection{Step 3: Bounding the spectral norm of $A$ via the trace moment method}
\label{sec:tracemethod}
The following spectral norm bound now implies \cref{thm:refutation}.
\begin{lemma}
    \label{lem:countingbacktrackingwalks}
    Let $A$ be the level-$\ell$ Kikuchi matrix over $\F^n$ defined in \cref{def:kikuchimatrix} for the $k$-$\LIN$ instance $\mcI = (\mcH, \{b_v\}_{v \in \mcH})$. Let $\Gamma \in \C^{N \times N}$ be the diagonal matrix $\Gamma = D + d \Id$ where $D_{U, U} := \deg(U)$ and $d = \E[\deg(U)]$. Suppose that the $b_v$'s are drawn independently and uniformly from $\F$, i.e., the instance $\mcI$ is \emph{semirandom} (\cref{def:semirandomklin}). Then, with probability $\geq 1 - \frac{1}{\poly(n)}$, it holds that
    \begin{equation*}
    \norm{\Gamma^{-1/2} A \Gamma^{-1/2}}_2 \leq O\left(\sqrt{\frac{\ell \log(\abs{\F^*}n)}{d}}\right) \mper
    \end{equation*}
\end{lemma}
We postpone the proof of \cref{lem:countingbacktrackingwalks} to the end of this section, and now finish the proof of \cref{thm:refutation}.
\begin{proof}[Proof of \cref{thm:refutation} from \cref{lem:countingbacktrackingwalks}]
Let $\mcI = (\mcH, \{b_v\}_{v \in \mcH})$ be the input to the algorithm. Given $\ell$, the algorithm constructs the matrix $A$ and computes $\algval(\mcI) = \frac{1}{\abs{\F}} +  \frac{2\abs{\F^*}}{\abs{\F}}\norm{\tilde{A}}_2 $, where $\tilde{A} = \Gamma^{-1/2} A \Gamma^{-1/2}$. It remains to argue that this quantity has the desired properties. 

Let $\Phi(x)$ be the polynomial defined in \cref{obs:advantageoverpolynomial}. For each $x \in \F^n$, letting $y \in \C^n$ be the vector defined in \cref{fact:bilinearforms}, we have
    \begin{flalign*}
 &\Phi(x) = \frac{1}{\abs{\F}  \abs{\mcH} \Delta} \cdot y^\dagger A y =  \frac{1}{\abs{\F}  \abs{\mcH} \Delta} \cdot (\Gamma^{1/2} y)^\dagger \tilde{A} (\Gamma^{1/2} y) \leq \frac{1}{\abs{\F}  \abs{\mcH} \Delta} \cdot \norm{\tilde{A}}_2 \norm{\Gamma^{1/2} y}_2^2 \\
 &=  \frac{1}{\abs{\F}  \abs{\mcH} \Delta} \cdot \norm{\tilde{A}}_2 \cdot \tr(\Gamma) = \frac{2\abs{\F^*}}{\abs{\F}}\norm{\tilde{A}}_2  \mcom
     \end{flalign*}
     where we use that $\norm{\Gamma^{1/2} y}_2^2 = y^{\dagger} \Gamma y = \sum_{U} \Gamma_U \abs{y_U}^2 = \sum_U \Gamma_U = \tr(\Gamma)$ since $\abs{y_U} = 1$ for all $U$, and that $\tr(\Gamma) = 2 \abs{\mcH} \abs{\F^*} \Delta$. Hence, 
     \begin{equation*}
     \val(\mcI) = \frac{1}{\abs{\F}} +  \max_{x \in \F^n} \Phi(x) \leq \frac{1}{\abs{\F}} +  \frac{2\abs{\F^*}}{\abs{\F}}\norm{\tilde{A}}_2  \mcom
     \end{equation*}
     which proves \cref{item:refutation1} in \cref{thm:refutation}.
     
     To prove \cref{item:refutation2}, we observe that by \cref{lem:countingbacktrackingwalks}, if $\mcI$ is semirandom, then with high probability over the draw of the $b_v$'s, it holds that
     \begin{equation*}
    \norm{\tilde{A}}_2  \leq O\left(\sqrt{\frac{\ell \log(\abs{\F^*}n)}{d}}\right) \mper
     \end{equation*}
    From \cref{fact:avgdegree}, we have $d \geq \frac{\abs{\F^*}}{2} \left(\frac{\ell}{\abs{\F^*}n}\right)^{k/2} \cdot \abs{\mcH}$. Hence, if $\abs{\mcH} \geq C n \log(\abs{\F^*}n) \left(\frac{\abs{\F^*}n}{\ell}\right)^{k/2-1} \varepsilon^{-2}$ for a sufficiently large constant $C$, then $\norm{\tilde{A}}_2 \leq \eps$ with probability $1 - 1/\poly(n)$. This proves \cref{item:refutation2}.
\end{proof}

\parhead{The trace moment method.} It remains to prove \cref{lem:countingbacktrackingwalks}, which we do using the trace moment method.
\begin{proof}[Proof of \cref{lem:countingbacktrackingwalks}]
By \cref{fact:hermitian}, we have that $\norm{\tilde{A}}_2 \leq \tr((\Gamma^{-1} A)^{2t})^{1/2t}$ for any positive integer $t$. Because the $b_v$'s are drawn independently from $\F$, the matrix $\tilde{A}$ is a random matrix. By Markov's inequality,
\begin{equation*}
\Pr\left[\tr((\Gamma^{-1} A)^{2t}) \geq N \cdot \E[\tr((\Gamma^{-1} A)^{2t})]\right] \leq \frac{1}{N} \mper
\end{equation*}
We note this event is the same as $\tr((\Gamma^{-1} A)^{2t})^{1/2t} \geq N^{1/2t} \cdot \E[\tr((\Gamma^{-1} A)^{2t})]^{1/2t}$, and for $2t \geq \log N $ we have $N^{1/2t} \leq O(1)$. This immediately gives us that with probability $\geq 
1- \frac{1}{N}$, $\norm{\tilde{A}}_2 \leq O\left(\E[\tr((\Gamma^{-1} A)^{2t})]^{1/2t}\right)$. We then have that
\begin{flalign*}
    \E\left[\tr\left(\left(\Gamma^{-1} A\right)^{2t}\right)\right] 
    &= \E\left[\tr\left(\left(\Gamma^{-1} \sum_{v \in \mcH, \beta \in \F^*} c_{v, \beta} \cdot A_{v, \beta}\right)^{2t}\right) \right]\\
    &= \E\left[\tr\left(\sum_{(v_1, \beta_1) ,..., (v_{2t}, \beta_{2t}) \in \mcH \times \F^*} \prod_{i = 1}^{2t} \Gamma^{-1} \cdot c_{v_i, \beta_i} \cdot A_{v_i, \beta_i} \right) \right]\\
    &= \sum_{(v_1, \beta_1) ,..., (v_{2t}, \beta_{2t}) \in \mcH \times \F^*} \E\left[\tr\left(\prod_{i = 1}^{2t} \Gamma^{-1} \cdot c_{v_i, \beta_i} \cdot A_{v_i, \beta_i} \right) \right] \\
        &= \sum_{(v_1, \beta_1) ,..., (v_{2t}, \beta_{2t}) \in \mcH \times \F^*} \E\left[\prod_{i=1}^{2t} c_{v_i, \beta_i} \right] \cdot \tr\left(\prod_{i = 1}^{2t} \Gamma^{-1}A_{v_i, \beta_i} \right)\mper
    \end{flalign*}
    Let us now make the following observation. Let $(v_1, \beta_1) ,..., (v_{2t}, \beta_{2t}) \in \mcH \times \F^*$ be a term in the above sum. Fix $v \in \mcH$, and let $R(v)$ denote the set of $i \in [2t]$ such that $v_i = v$. We observe that if for some $v \in \mcH$, $\sum_{i \in R(v)} \beta_i \ne 0$, then $\E\left[\prod_{i=1}^{2t} c_{v_i, \beta_i} \right] = 0$. Indeed, this is because $b_v$ is independent for each $v \in \mcH$, and so $\E\left[\prod_{i=1}^{2t} c_{v_i, \beta_i} \right]  = \prod_{v \in \mcH} \E\left[\prod_{i \in R(v)} c_{v, \beta_i} \right]$, and 
    \begin{equation*}
    \E\left[\prod_{i \in R(v)}c_{v, \beta_i} \right] = \E\left[\prod_{i \in R(v)} \omega_p^{\Tr(\beta_i b_v)}\right] = \E\left[\omega_p^{\Tr((\sum_{i \in R(v)}\beta_i) b_v)}\right] \mper
    \end{equation*}
    Then, since $b_v$ is uniform from $\F$, it follows that $\E\left[\omega_p^{\Tr((\sum_{i \in R(v)}\beta_i) b_v)}\right] = 0$ if $\sum_{i \in R(v)} \beta_i \ne 0$, and $\E\left[\omega_p^{\Tr((\sum_{i \in R(v)}\beta_i) b_v)}\right] = 1$ if $\sum_{i \in R(v)} \beta_i = 0$. This motivates the following definition.
    \begin{definition}[Trivially closed sequence]
    \label{def:triviallyclosedwalks}
    Let $(v_1, \beta_1) ,..., (v_{2t}, \beta_{2t}) \in \mcH \times \F^*$. We say that $(v_1, \beta_1) ,..., (v_{2t}, \beta_{2t}) \in \mcH \times \F^*$ is trivially closed with respect to $v$ if it holds that $\sum_{i \in R(v)} \beta_i = 0$. We say that the sequence is trivially closed if it is trivially closed with respect to all $v \in \mcH$.
    \end{definition}
    With the above definition in hand, we have shown that
\begin{flalign*}
  & \E\left[\tr\left(\left(\Gamma^{-1} A\right)^{2t}\right)\right] 
    = \sum_{\substack{(v_1, \beta_1) ,..., (v_{2t}, \beta_{2t}) \\ \text{trivially closed}}} \tr\left(\prod_{i = 1}^{2t} \Gamma^{-1}A_{v_i, \beta_i} \right)\mper
\end{flalign*}

The following lemma yields the desired bound on $\E[\tr((\Gamma^{-1} A)^{2t})]$.
\begin{lemma}
    \label{lem:maincountingbacktrackingwalks}
    $\sum_{\substack{(v_1, \beta_1) ,..., (v_{2t}, \beta_{2t}) \\ \text{trivially closed}}} \tr\left(\prod_{i = 1}^{2t} \Gamma^{-1}A_{v_i, \beta_i} \right) \leq N \cdot 2^{2t} \cdot \left(\frac{2t}{d}\right)^t$.
\end{lemma}
With \cref{lem:maincountingbacktrackingwalks}, we thus have the desired bound $\E[\tr((\Gamma^{-1} A)^{2t})]$. Taking $t$ to be $c \log_2 N$ for a sufficiently large constant $c$ and applying Markov's inequality finishes the proof.
\end{proof}

\begin{proof}[Proof of  \cref{lem:maincountingbacktrackingwalks}]
    We bound the sum as follows. First, we observe that for a trivially closed sequence $(v_1, \beta_1) ,..., (v_{2t}, \beta_{2t})$, we have
    \begin{flalign*}
    \tr\left(\prod_{i = 1}^{2t} \Gamma^{-1}A_{v_i, \beta_i} \right) = \sum_{U_0, U_1, \dots, U_{2t-1}} \prod_{i = 0}^{2t - 1} \Gamma^{-1}_{U_i} \cdot \1\left(U_i\xrightarrow{\text{$v_{i+1}, \beta_{i+1}$}} U_{i+1}\right) \mcom
    \end{flalign*}
    where we define $U_{2t} = U_0$. Thus, the sum that we wish to bound in \cref{lem:maincountingbacktrackingwalks} simply counts the total weight of ``trivially closed walks'' $U_0, v_1, \beta_1, U_1, \dots, U_{2t-1}, v_{2t}, \beta_{2t}, U_{2t}$ (where $U_{2t} = U_0$) in the Kikuchi graph $A$, where the weight of a walk is simply $\prod_{i = 0}^{2t-1} \Gamma^{-1}_{U_i}$.
    
    Let us now bound this total weight by uniquely encoding a walk $U_0, v_1, \beta_1, U_1, \dots, U_{2t-1}, v_{2t}, \beta_{2t}, U_{2t}$ as follows.
    \begin{itemize}
        \item First, we write down the start vertex $U_0$.
        \item For $i = 1, \dots, 2t$, we let $z_i$ be $1$ if $v_i = v_j$ for some $j < i$. In this case, we say that the edge is ``old''. Otherwise $z_i = 0$, and we say that the edge is ``new''.
	\item For $i = 1, \dots, 2t$, if $z_i$ is $1$ then we encode $U_{i}$ by writing down the smallest $j \in [2t]$ such that $v_i = v_j$. We note that we \emph{do not} need to specify the element $\beta_i$, as for any vertex $U$, there is at most one $V$ and one $\beta \in \F^*$ such that $\1(U \xrightarrow{\text{$v_{i}, \beta$}} V)$. As pointed out in \cref{rem:nontrivialkikuchi}, this crucially saves us a factor of $\abs{\F^*}$ in the total number of constraints that we require.
	\item For $i = 1, \dots, 2t$, if $z_i$ is $0$ then we encode $U_i$ by writing down an integer in $1, \dots, \deg(U_{i-1})$ that specifies the edge we take to move to $U_i$ from $U_{i-1}$ (we associate $[\deg(U_{i-1})]$ to the edges adjacent to $U_{i-1}$ with an arbitrary fixed map).
    \end{itemize}
    With the above encoding, we can now bound the total weight of all trivially closed walks as follows. First, let us consider the total weight of walks for some fixed choice of $z_1, \dots, z_{2t}$. We have $N$ choices for the start vertex $U_0$. For each $i = 1, \dots, 2t$ where $z_i = 0$, we have $\deg(U_{i-1})$ choices for $U_i$, and we multiply by a weight of $\Gamma^{-1}_{U_{i-1}} \leq \frac{1}{\deg(U_{i-1})}$. For each $i = 1, \dots, 2t$ where $z_i = 1$, we have at most $2t$ choices for the index $j < i$, and we multiply by a weight of $\Gamma^{-1}_{U_{i-1}} \leq \frac{1}{d}$. Hence, the total weight for a specific $z_1, \dots, z_{2t}$ is at most $N \left(\frac{2t}{d}\right)^{r}$, where $r$ is the number of $z_i$ such that $z_i = 1$.
    
    Finally, we observe that any trivially closed walk must have $r \geq t$. Hence, after summing over all $z_1, \dots, z_{2t}$, we have the final bound of $N 2^{2t} \left(\frac{2t}{d}\right)^{t}$, which finishes the proof.
\end{proof}

\section{Refuting Semirandom \texorpdfstring{$k$-$\LIN$}{k-LIN} over Abelian Groups for Even $k$}
\label{sec:grouprefutation}

In this section, we prove the even case of \cref{thm:grouprefutation}, extending our refutation techniques in the finite field setting to general finite Abelian groups. A naive application of the Kikuchi matrix refutation in \cref{sec:refutation} produces a version of \cref{thm:refutation} with an extra factor of $\abs{G}$ in the number of equations required in \cref{item:grouprefutation2}, coming from zero divisors in the group inducing self-loops or multiedges in the Kikuchi graph. In the case where $\varepsilon < \frac{1}{\abs{G}}$, this extra factor is only a mild $\varepsilon^{-1}$, as stated in the final result of \cref{thm:grouprefutation}, but this fails when $\abs{G} \gg \frac{1}{\varepsilon}$.

This leaves us in the curious position where the ``easier'' case (with $\varepsilon \gg \frac{1}{\abs{G}}$) has a larger dependence on $\abs{G}$. The reason is simply that we are doing extra work to obtain a bound $\frac{1}{\abs{G}} + \varepsilon$, even though the factor $\frac{1}{\abs{G}}$ is dominated by $\varepsilon$. A reasonable fix is to refute the instance in a quotient group of size $\frac{1}{\varepsilon^{O(1)}}$. A natural question is what to do if such a quotient group does not exist, and it turns out that the only such case is when $G$ is a field, where the extra $O(\abs{G})$ factor can be removed naturally as in \cref{sec:refutation}. 

\subsection{Properties of \texorpdfstring{$k$-$\LIN$}{k-LIN} over finite Abelian groups}

Throughout this section, we think of $G = \bigotimes_{i=1}^r \Z_{m_i}$ and recall that, as in \cref{def:groupsemirandomklin}, multiplication within $G$ is defined as direct product multiplication between the cyclic groups. For $x \in G$ we denote by $x^{(i)}$ the component from $\Z_{m_i}$.

\begin{definition}[$\lambda$-thin equations]
    Given a $k$-$\LIN(G)$ equation represented by a $k$-sparse vector $v \in G^n$, we collect the non-zero coefficients associated with a particular subgroup $\{v^{(i)}_j\}_{j \in [n], v_j \neq 0}$ for $i \in [r]$ and call the subgroup of $\Z_{m_i}$ generated by this collection $H_i(v)$. We could also define $H_i(v)$ equivalently as $\mu\Z_{m_i}$ where $\mu = \mathrm{gcd}(m_i, \{v_j^{(i)}\}_{j \in [n], v_j \neq 0})$. We define the representative group of the equation as the product $H(v) = \bigotimes_{i=1}^r H_i(v)$. We define the \textit{thinness} of a equation as
    \begin{equation*}
        \lambda(v) := \frac{\abs{H(v)}}{\abs{G}}= \frac{1}{\abs{G}}\prod_{i=1}^r \abs{H_i(v)} \mcom
    \end{equation*}
    and say a equation is $\lambda$-thin if $\lambda(v) \geq \lambda$.
\end{definition}

\begin{example}
    Let $G = \Z_{2p}$ for some prime $p \in \Z$. The equation ``$px_1 + px_2 + px_3$'' would be $\frac{2}{\abs{G}}$-thin since the representative group is $p\Z_{2p}$ with size $2$. In contrast, the equation ``$2x_1 + 2x_2 + 2x_3$'' is $\frac{1}{2}$-thin.
\end{example}

\begin{remark}
    \label{rem:size}
    Note that in the case $G = \F_p$ for prime $p$ then every equation is necessarily $1$-thin, so \cref{sec:refutation} can be seen as a special case of what follows.
\end{remark}

The thinness of a equation turns out to be crucial in understanding when equations are satisfiable.

\begin{observation}
    \label{obs:groupprob}
    Given a $k$-$\LIN(G)$ equation $v \in G^n$ with associated $b_v \in G$, it is only possible to satisfy $v$ by any assignment if $b_v$ is in the representative group $H(v)$.
\end{observation}

This can be seen by noting that multiplication of any coefficient by any $x_i$ is equivalent to generation by that element, so the left hand side falls in $H(v)$. An immediate consequence of above is that if $b_v$ is chosen uniformly from $G$ for $\lambda$-thin $v$, the probability that the equation is satisfiable is exactly $\lambda$. We also prove the following lemma that characterizes the structure of our Kikuchi matrix for $G$ according to thinness.

\begin{observation}
    \label{obs:groupdivisors}
    Let $H(v) = \bigotimes_{i=1}^r \mu_i\Z_{m_i}$. Then if $\beta \in G$ has $\beta^{(i)} \in (m_i/\mu_i)\Z_{m_i} + \alpha^{(i)}$ (where by convention $m_i\Z_{m_i} = \{0\}$) then $\beta v = \alpha v$.
\end{observation}

\begin{proof}
    Note that $\mu_i = \mathrm{gcd}(m_i, \{v^{(i)_j}\}_{j \in [n], v_j \neq 0})$, so every such $v^{(i)}_j$ has $\mu_i$ as a divisor, so we write it as $\mu_i a$. For any $\beta^{(i)} \in (m_i/\mu_i)\Z_{m_i} + \alpha^{(i)}$ we may write it as $(m_i/\mu_i)b + \alpha^{(i)}$ for some $b \in \Z_{m_i}$. We then have $\beta^{(i)}v^{(i)}_j = m_iab + \alpha^{(i)}v_j^{(i)}$ which is the same as $\alpha^{(i)}v_j^{(i)}$ mod $m_i$.
\end{proof}

The result of the above is that we can group elements of $G$ according to the equivalence classes of their action on $v \in G^n$, with each class consisting of exactly $1/\lambda$ identical elements. Later this shows up as every edge from a $\lambda$-thin equation in our Kikuchi graph actually clusters into a multi-edge of $1/\lambda$ edges. The degree of the Kikuchi matrix is also further characterized by the following stricter definition.

\begin{definition}[$\lambda$-robust equations]
    Given a $k$-$\LIN(G)$ equation represented by a $k$-sparse vector $v \in G^n$, we consider the set of $k/2$-size subequations of $v$, $\{v_S\}_{S \in {[k] \choose k/2}}$. We say $v$ is $\lambda$-robust if all $v_S$ are at worst $\lambda$-thin and an instance is $\lambda$-robust if all equations are $\lambda$-robust.
\end{definition}

We are now ready to state our main technical result, which is our Kikuchi matrix refutation algorithm for even $k$.

\begin{lemma}[Kikuchi matrix refutation of semirandom $k$-$\LIN(G)$]
    \label{lem:grouprefutation}
    Fix $\ell \geq k/2$. There is an algorithm that takes as input a $k$-$\LIN(G)$ subinstance $\mcI = (\mcH, \{b_v\}_{v \in \mcH})$ in $n$ variables and outputs a number $\algval(\mcI) \in [0,1]$ in time $(\abs{G} n)^{O(\ell)}$ with the following two guarantees:
    \begin{enumerate}[(1)]
        \item \label{item:lemgrouprefutation1} $\algval(\mcI) \geq \val(\mcI)$ for every instance $\mcI$;
        \item \label{item:lemgrouprefutation2} If $\abs{\mcH} \geq O(n) \cdot \frac{1}{\lambda} \cdot \left(\frac{n\abs{G}}{\ell}\right)^{k/2-1} \cdot \log(\abs{G}n) \cdot \varepsilon^{-2}$ and $\mcI$ is $\lambda$-robust and drawn from the semirandom distribution described in \cref{def:groupsemirandomklin}, then with probability $\geq 1-\frac{1}{\poly(n)}$ over the draw of the semirandom instance, i.e., the randomness of $\{b_v\}_{v \in \mcH}$, it holds that $\algval(\mcI) \leq \frac{1}{\abs{G}} + \varepsilon$.
    \end{enumerate}
\end{lemma}

\begin{remark}
    The key difference between \cref{lem:grouprefutation} and \cref{thm:grouprefutation} is the $\lambda$-robust assumption, which yields an extra factor of $1/\lambda$ in the number of constraints. As every instance is trivially $1/\abs{G}$-robust, this extra factor is at most $\abs{G}$. The rest of this section is devoted to preprocessing the instance so that this blowup is instead $\poly(1/\varepsilon)$.
\end{remark}

Our main goal is to apply \cref{lem:grouprefutation} to an instance that is $\lambda$-robust where $\lambda$ is not too small. This means if the original instance is $\lambda$-robust only for a small $\lambda$, which corresponds to $G$ being large and having many zero divisors, we would like to switch to a smaller (but still larger than $O(\varepsilon^{-1})$) or more robust group. The reason we can do this is the following observation.

\begin{observation}
    Let $G$ be an Abelian group and fix a $k$-$\LIN(G)$ instance $\mathcal{I}$. For any subgroup $H$ of $G$, let $\tilde{\mcI}$ be the result of modding each equation in $\mcI$ by $H$, i.e., we map the equations over $G$ to equations over $G / H$ using the natural surjective group homomorphism. Then $\val(\mcI) \leq \val(\tilde{\mcI})$.
\end{observation}

\begin{proof}
    Note that for any assignment $x \in G^n$, if $x$ satisfies an equation in $\mcI$ then $x$ modded by $H$ satisfies the equation as well.
\end{proof}

We can thus mod the equations over $G$ by some subgroup $H$ without decreasing the value of the instance. There are now two ways we can succeed in lowering $\frac{1}{\lambda}$: (1) find a subgroup $H$ such that of $G/H$ has size $\approx \varepsilon^{-1}$, or (2) find a subgroup $H$ such that $G/H$ is a field, i.e., it is a group of prime order, and so the $\frac{1}{\lambda}$ can be removed completely. The following lemma shows that one of these must always be possible.

\begin{lemma}
    \label{lem:groupchoice}
     For any $t \in \N$ and Abelian group $G$, there is a subgroup $H$ of $G$ such that (1) $\abs{G/H} \leq t$, or (2) $\abs{G/H} > t$ and $G/H$ is prime order, or (3) $t < \abs{G/H} \leq t^2$ and any non-trivial subgroup of $G/H$ has size at least $\frac{1}{t}\abs{G/H}$. Moreover, the subgroup $H$ can be found in time $\poly(\abs{G})$.
\end{lemma}

\begin{proof}
    First, we assume that $\abs{G} \geq t$, as if $\abs{G} < t$ then we may take $H$ to be the group with one element so that $G/H = G$.

    By the Fundamental Theorem of Finite Abelian groups, we can assume that $G$ is isomorphic to $\bigotimes_{i = 1}^r \Z_{m_i}$, where the $m_i$ are prime powers $p_i^{e_i}$ and are arranged such that the $p_i$ are in decreasing order.
    Let $(\mu_1, ..., \mu_r)$ with each $\mu_i \in \Z$ with $1 \leq \mu_i \leq m_i$ for each $i$, and consider the subgroup $H$ given by $\bigotimes_{i = 1}^r \mu_i \Z_{m_i}$. We observe that $\abs{H} = \bigotimes_{i = 1}^r \frac{m_i}{\mu_i}$, and that $\abs{G/H} = \prod_{i = 1}^r \mu_i$.

    Next, we assume that $p_i \leq t$ for all $i$, as if $p_s > t$ for some $s$, then we may take $H$ to be the subgroup $(\bigotimes_{i \ne s} \Z_{m_i}) \bigotimes (p_s\Z_{m_s})$, so that $G/H$ is isomorphic to $\Z_{p_s}$, which is prime order and has size $p_s > t$.

    Now, we can search through elements $(\mu_1, ..., \mu_r) \in G$ of the form $(m_1,...,m_{s-1}, \mu_s, 1, ..., 1)$ where $\mu_s = p_s^{e'_s}$ satisfies $1 < \mu_s \leq m_s$. Note that there are at most $\abs{G}$ such elements. Let $H$ given by $\bigotimes_{i = 1}^r \mu_i \Z_{m_i}$. By the above, the quotient group $G/H$ has size $\left(\prod_{i = 1}^{s-1} m_i \right) \mu_s$. Fix $H$ to be the largest subgroup of the form $\bigotimes_{i=1}^r \mu_i \Z_{m_i}$ where the $\mu_i$'s are in the above set and $\abs{H} \leq \abs{G}/t$ (equivalently $\abs{G/H} > t$). Note that such a subgroup $H$ exists since $\abs{G} \geq t$, so this is well-defined.

   Let  $H = \bigotimes_{i = 1}^{s-1} m_i \Z_{m_i} \bigotimes (\mu_s \Z_{m_s}) \bigotimes_{i = s + 1}^r \Z_{m_i}$. Recall that we have shown that we may assume $p_s \leq t$. It remains to argue that $\abs{G/H} \leq t^2$ and that any non-trivial subgroup of $G/H$ has size at least $\frac{1}{t} \abs{G/H}$. We have that $G/H$ is isomorphic to $\left(\bigotimes_{i = 1}^{s-1} \Z_{m_i}\right) \bigotimes \Z_{\mu_s}$, where recall that $m_i = p_i^{e_i}$ and $\mu_s = p_s^{e'_s}$ and the $p_i$'s are in decreasing order. Since $\mu_s > 1$, it then follows that the smallest subgroup of $G/H$ is isomorphic to $\Z_{p_s}$. By the maximality assumption on $H$, the subgroup $H'$ given by the tuple 
    $(m_1,...,m_{s-1}, \mu_s/p_s, 1, ..., 1)$ has size $\abs{H'} \geq \abs{G}/t$. Note that $\abs{H'} = p_s \abs{H}$, and so it follows that $p_s \abs{H} \geq \abs{G}/t$, which implies that $\abs{G/H} \leq p_s t$. Since any non-trivial subgroup of $G/H$ has size at least $p_s$, this is at least $\abs{G/H}/t$ by the above. Furthermore, since $p_s \leq t$, we have $\abs{G/H} \leq p_s t \leq t^2$. 

    Finally, we observe that the above operations to find the subgroup $H$ (and thereby the quotient group $G/H$ also) can be clearly done in time $\poly(\abs{G})$ when $G$ is given in the form $\bigotimes_{i = 1}^r \Z_{m_i}$. This finishes the proof.
\end{proof}

We are now ready to state our full algorithm.

\begin{tcolorbox}[
    width=\textwidth,   
    colframe=black,  
    colback=white,   
    title=Reduction to $\lambda$-robust $k$-$\LIN(G)$ Refutation,
    colbacktitle=white, 
    coltitle=black,      
    fonttitle=\bfseries,
    center title,   
    enhanced,       
    frame hidden,           
    borderline={1pt}{0pt}{black},
    sharp corners,
    toptitle=2.5mm
]
\textbf{Input:} A $k$-$\LIN(G)$ instance $\mcI = (\mcH, \{b_v\}_{v \in \mcH})$.\\

\textbf{Output:} $\algval(\mcI) \in [0,1]$ with guarantee $\algval(\mcI) \geq \max_{x \in G^n} \val(\mcI, x)$.\\

\textbf{Algorithm:}
\begin{enumerate}
    \item Find a subgroup $H$ of $G$ with $t = 4\varepsilon^{-1}$ according to \cref{lem:groupchoice} and mod $\mcI$ into $G/H$ as $\mcI_1$.
    \item Move any equations with multiple zero coefficients after modding into a separate instance $\mathcal{I}_2$ (with the coefficients deleted). If they have none or a single non-zero coefficient left, additionally move them to $\mathcal{I}_0$.
    \item If $\frac{\abs{\mcH_1}}{\abs{\mcH}} \leq \frac{\varepsilon}{4}$ then let $\algval(\mcI_1)$ (and do the same for $\mcI_2$). Otherwise, run the Kikuchi matrix refutation on $\mcI_1$ as a $k$-$\LIN(G/H)$ instance and on $\mathcal{I}_2$ as a $(k-2)$-$\LIN(G/H)$ instance. Compute the value of $\mathcal{I}_0$ exactly.\footnote{For trivial equations, this is done by checking if the right hand side is $0$. For equations with a single coefficient, this can be done choosing the assignment to each variable separately to maximize the number of equations satisfied.}
    \item Output $\frac{\abs{\mcH_{\mcI_1}}}{\abs{\mcH}} \algval(\mcI_1) + \frac{\abs{\mcH_{\mathcal{I}_2}}}{\abs{\mcH}} \algval(\mathcal{I}_2) + \frac{\abs{\mcH_{\mathcal{I}_0}}}{\abs{\mcH}} \val(\mathcal{I}_0)$.
\end{enumerate}

\end{tcolorbox}

\begin{proof}[Proof of \cref{thm:grouprefutation} from \cref{lem:grouprefutation}]
    We compute
    \begin{flalign*}
        \val(\mcI) &= \max_{x \in G^n} \val(\mcI, x) = \max_{x \in G^n} \frac{\abs{\mcH_{\mcI_1}}}{\abs{\mcH}} \val(\mcI_1,x) + \frac{\abs{\mcH_{\mathcal{I}_2}}}{\abs{\mcH}} \val(\mathcal{I}_2,x) + \frac{\abs{\mcH_{\mathcal{I}_0}}}{\abs{\mcH}} \val(\mathcal{I}_0,x)\\
        &\leq \frac{\abs{\mcH_{\mcI_1}}}{\abs{\mcH}} \val(\mcI_1) + \frac{\abs{\mcH_{\mathcal{I}_2}}}{\abs{\mcH}} \val(\mathcal{I}_2) + \frac{\abs{\mcH_{\mathcal{I}_0}}}{\abs{\mcH}} \val(\mathcal{I}_0)\mper
    \end{flalign*}
    By \cref{item:lemgrouprefutation1} of \cref{lem:grouprefutation}, $\frac{\abs{\mcH_{\mcI_1}}}{\abs{\mcH}} \algval(\mcI_1) + \frac{\abs{\mcH_{\mathcal{I}_2}}}{\abs{\mcH}} \algval(\mathcal{I}_2) + \frac{\abs{\mcH_{\mathcal{I}_0}}}{\abs{\mcH}} \val(\mathcal{I}_0)$ is indeed a certificate on $\val(\mcI)$, giving \cref{item:grouprefutation1}.

    Since modding preserves semirandomness, we have
    \begin{equation*}
        \frac{\abs{\mcH_{\mcI_1}}}{\abs{\mcH}} \algval(\mcI_1) \leq \frac{\abs{\mcH_{\mcI_1}}}{\abs{\mcH}} \cdot \frac{1}{\abs{H}} + \frac{\abs{\mcH_{\mcI_1}}}{\abs{\mcH}} \cdot \sqrt{\frac{\abs{\mcH}}{\abs{\mcH_{\mcI_1}}}} \frac{\varepsilon}{4} \leq \frac{\abs{\mcH_{\mcI_1}}}{\abs{\mcH}} \cdot \frac{1}{\abs{H}} +\frac{\varepsilon}{4}\mcom
    \end{equation*}
    w.h.p. if $\abs{\mcH} \geq C n \log(\abs{G}n) \left(\frac{\abs{G}n}{\ell}\right)^{k/2-1} \varepsilon^{-3}$ for sufficiently large $C > 0$, by absorbing the $\frac{\abs{\mcH_{\mcI_1}}}{\abs{\mcH}}$ into the loss $\varepsilon$ in \cref{lem:grouprefutation}. Note the factor $\frac{1}{\lambda} \leq \frac{4}{\varepsilon}$ here since \cref{lem:groupchoice} guarantees there is no subgroup of $G/H$ of thinness less than $\frac{\varepsilon}{4}$ and since there is at most a single zero coefficient, it cannot form a thin subequation when $k > 2$. Similarly, we have $\frac{\abs{\mcH_{\mcI_2}}}{\abs{\mcH}} \algval(\mcI_2) \leq \frac{\abs{\mcH_{\mcI_2}}}{\abs{\mcH}} \cdot \frac{1}{\abs{H}} + \sqrt{\frac{\abs{\mcH_{\mcI_2}}}{\abs{\mcH}}} \frac{\varepsilon}{4}$ w.h.p. if $\abs{\mcH} \geq C n \log(\abs{G}n) \left(\frac{\abs{G}n}{\ell}\right)^{k/2-1} \varepsilon^{-2}$. Unlike in the previous case, we do not assume $\frac{1}{\lambda} < \abs{G}$ in our use of \cref{lem:grouprefutation}, instead relying on the fact that we reduced the arity by $2$, making the threshold required a factor $\abs{G}$ lower already.

    For $\mcI_0$, we can appeal to standard Chernoff bounds (see \cref{lem:unsatisfiability}) to claim that the value concentrates within $\frac{1}{\abs{H}} + \frac{\varepsilon}{4}$ at the above equation density so long as $\frac{\abs{\mcH_{\mcI_0}}}{\abs{\mcH}} \geq \frac{\varepsilon}{4}$. We conclude $\frac{\abs{\mcH_{\mcI_0}}}{\abs{\mcH}} \val(\mcI_0) \leq \max\left(\frac{\abs{\mcH_{\mcI_0}}}{\abs{\mcH}} (\frac{1}{\abs{H}} + \frac{\varepsilon}{4}), \frac{\varepsilon}{4}\right)$.

    Putting the bounds together we get with high probability
    \begin{equation*}
        \frac{\abs{\mcH_{\mcI_1}}}{\abs{\mcH}} \algval(\mcI_1) + \frac{\abs{\mcH_{\mathcal{I}_2}}}{\abs{\mcH}} \algval(\mathcal{I}_2) + \frac{\abs{\mcH_{\mathcal{I}_0}}}{\abs{\mcH}} \val(\mathcal{I}_0) \leq \frac{1}{\abs{H}} + \frac{3\varepsilon}{4}\mper
    \end{equation*}
    If we have $G/H = G$, this achieves \cref{item:grouprefutation2}. Otherwise, by the choice of $H$, we have $\frac{1}{\abs{G/H}} \leq \frac{\varepsilon}{4}$, also achieving \cref{item:grouprefutation2}.
\end{proof}

\subsection{Finite Abelian Kikuchi matrices}

Following the previous section, we have reduced the proof of \cref{thm:grouprefutation} to showing how to accomplish $\lambda$-robust refutation in $G$ as in \cref{lem:grouprefutation}, which we prove here. As in the proof of \cref{thm:refutation}, the key is in defining a suitable Kikuchi matrix over $G = \bigotimes_{i=1}^r \Z_{m_i}$. 

\begin{definition}{(Even-arity Kikuchi matrix over $G = \bigotimes_{i=1}^r \Z_{m_i}$).}
\label{def:groupkikuchimatrix}
Let $k/2 \leq \ell \leq n/2$ be a parameter, and let $N = \abs{G}^\ell {n \choose \ell}$. For each $k$-sparse\footnote{Technically, we allow vectors that are less than $k$-sparse in order to accommodate zeroed out coefficients when modding.} vector $v \in G^n$ and $\beta \in G^*$ for $G^* = G \setminus \{0\}$, we define a matrix $A_{v, \beta} \in \C^{N \times N}$ as follows. First, we identify $N$ with the set of pairs $(U, S)$ where $U \in G^n$ and $S \in {[n] \choose \ell}$ with the condition $\supp(U) \subseteq S$. We call the set of all such pairs $I(G, n, \ell)$. Then, for any pair $(U,S)$ and $(V,T)$ we define
    \begin{equation*}
        A_{v, \beta}((U, S), (V, T)) = \begin{cases}
                      1  & (U, S)\xrightarrow{\text{$v, \beta$}} (V,T)\\
                      0 & \text{otherwise}
                    \end{cases}
    \end{equation*}
    where we say $(U,S)\xrightarrow{\text{$v, \beta$}} (V,T)$ if $U-V = \beta v$ and $S \oplus T = \supp(v)$.
    
Let $\Phi(x) = \frac{1}{\abs{G} \abs{\mcH}}\sum_{v \in \mcH} \sum_{\beta \in G^*} c_{v, \beta} \cdot \chi_{\beta v}$ be a polynomial defined by a set $\mcH$ of $k$-sparse vectors from $G^n$ and complex coefficients $\cbra{c_{v, \beta}}_{\substack{v \in \mcH \\ \beta \in G^*}}$. We define the level-$\ell$ Kikuchi matrix for this polynomial to be $A = \sum_{v \in \mcH} \sum_{\beta \in G^*} c_{v, \beta} \cdot A_{v, \beta}$.
\end{definition}

\begin{observation}
    \label{obs:groupadvantageoverpolynomial}
    For a $k$-$\LIN(G)$ instance $\mcI = (\mcH, \{b_v\}_{v \in \mcH})$, let $\val(\mcI, x)$ denote the fraction of constraints satisfied by an assignment $x \in G^n$. Then, we have
    \begin{equation*}
        \val(\mcI, x) = \frac{1}{\abs{G}} + \frac{1}{\abs{\mcH} \abs{G}}\sum_{v \in \mcH} \sum_{\beta \in G^*} \chi_{\beta}(b_v) \cdot \overline{\chi_{\beta v}(x)} := \frac{1}{\abs{G}} + \Phi(x) \mcom
    \end{equation*}
\end{observation}
\begin{proof}
    Recall that a constraint in $\mcI$ takes the form $\langle v, x \rangle = b_v$ for $v \in \mcH$, where $x \in G^n$ are the variables. The indicator variable for this event is simply:
    \begin{align*}
        \1(\langle v, x \rangle = b_v) &= \prod_{i=1}^r \1(\langle v^{(i)}, x^{(i)} \rangle = b_v^{(i)})\\
        &= \prod_{i=1}^r \E_{\beta_i \sim \Z_{m_i}}\sbra{\omega_{m_i}^{\beta_i b_v^{(i)} - \beta_i \langle v^{(i)}, x^{(i)} \rangle}}\\ 
        &= \prod_{i=1}^r \frac{1}{\abs{\Z_{m_i}}}\sum_{\beta_i \in \Z_{m_i}}\omega_{m_i}^{\beta_i b_v^{(i)} - \beta_i \langle v^{(i)}, x^{(i)} \rangle}\\ 
        &=\frac{1}{\abs{G}}\sum_{\beta \in G} \prod_{i=1}^r \omega_{m_i}^{\beta^{(i)} b_v^{(i)} - \beta^{(i)} \langle v^{(i)}, x^{(i)} \rangle}\\ 
        &= \frac{1}{\abs{G}}\sum_{\beta \in G} \chi_\beta(b_v) \cdot \overline{\chi_{\beta v}(x)} \mper
    \end{align*}
    If $\langle v, x \rangle = b_v$ at each coordinate $i \in [r]$, then $\beta^{(i)} b_v^{(i)} - \beta^{(i)}\ip{v,x}^{(i)} = 0$ for all $\beta \in G$. If $b_v - \ip{v,x} \ne 0$, we have some coordinate $i$ that is non-zero, then $\E_{\beta_i \sim \Z_{m_i}}\sbra{\omega_{m_i}^{\beta_i (b_v^{(i)} - \langle v, x \rangle^{(i)})}} = 0$. Hence, it follows that
    \begin{flalign*}
    &\val(\mcI, x) = \frac{1}{\abs{\mcH}} \sum_{v \in \mcH}  \1(\langle v, x \rangle = b_v) = \frac{1}{\abs{\mcH}} \sum_{v \in \mcH}  \frac{1}{\abs{G}} \sum_{\beta \in G} \chi_\beta(b_v) \cdot \overline{\chi_{\beta v}(x)} \\
    &= \frac{1}{\abs{G}} + \frac{1}{\abs{\mcH} \abs{G}} \sum_{v \in \mcH} \sum_{\beta \in G^*} \chi_\beta(b_v) \cdot \overline{\chi_{\beta v}(x)}  \mper  \end{flalign*}
\end{proof}

The following observation shows that we can express $\Phi(x)$ as a quadratic form on the matrix $A$ defined in \cref{def:groupkikuchimatrix}.
\begin{observation}
    \label{fact:groupbilinearforms}
    For $x \in G^n$ define $y \in \C^N$ as follows. For each $(U, S) \in I(G, n, \ell)$ as in \cref{def:groupkikuchimatrix}, we set $y_U = \chi_U(x)$. Then
    \begin{equation*}
        \Phi(x) = \frac{1}{\abs{\mcH}\abs{G} \Delta} y^\dagger A y \mper
    \end{equation*}
    where $\Delta := {k \choose k/2} {n-k \choose \ell - k/2}\abs{G}^{\ell-k/2}$.
\end{observation}

\begin{proof}
    \begin{align*}
        y^\dagger A y &= \sum_{(U,S), (V,T) \in I(G, n, \ell)} A((U,S), (V,T)) \cdot \overline{\chi_U(x)} \cdot \chi_V(x)\\
        &= \sum_{(U,S), (V,T) \in I(G, n, \ell)} \sum_{v \in \mcH, \beta \in G^*}
        \1\left((U,S)\xrightarrow{\text{$v, \beta$}} (V,T)\right) \cdot c_{v, \beta} \cdot \overline{\chi_U(x)} \cdot \chi_V(x)\\
        &= \sum_{(U,S), (V,T) \in I(G, n, \ell)} \sum_{v \in \mcH, \beta \in G^*} \1\left((U,S)\xrightarrow{\text{$v, \beta$}} (V,T)\right) \cdot c_{v, \beta} \cdot \overline{\chi_{U - V}(x)} \\
        &= \sum_{(U,S), (V,T) \in I(G, n, \ell)} \sum_{v \in \mcH, \beta \in G^*} \1\left((U,S)\xrightarrow{\text{$v, \beta$}} (V,T)\right) \cdot c_{v, \beta} \cdot \overline{\chi_{\beta v}(x)} \mper
    \end{align*}
    For each $v \in \mcH$ and $\beta \in G^*$, the term $c_{v, \beta} \cdot \overline{\chi_{\beta v}(x)}$ appears once for each pair of vertices $((U,S), (V,T))$ with $(U,S)\xrightarrow{\text{$v, \beta$}} (V,T)$. Let us now argue that the number of such pairs is exactly $\Delta = {k \choose k/2} {n-k \choose \ell - k/2}\abs{G}^{\ell-k/2}$. We count the number by first specifying $\supp(U)$ and $\supp(V)$, and then by specifying $U_i$ for each $i \in \supp(U)$ (and same for $V$). We first require that $\supp(U) \oplus \supp(V) = \supp(v)$, which in turn means that $\supp(U)$ has intersection exactly $k/2$ with $\supp(v)$ and likewise for $\supp(V)$. Thus, we can pay ${k \choose k/2}$ to count the number of ways to split $\supp(v)$ into two equal parts. Second, we need to specify $\supp(U) \setminus \supp(v)$, which is equal to $\supp(V) \setminus \supp(v)$, which is ${n - k\choose \ell - k/2}$ choices. Finally, we need to specify $U_i$ for each $i \in \supp(U)$ and $V_i$ for each $i \in \supp(V)$. For each $i \in \supp(U) \cap \supp(v)$, we set $U_i = (\beta v)_i$, and for each $i \in \supp(U) \setminus \supp(v)$, we can set $U_i$ to be any element in $G$. Note that specifying $U$ then determines $V$, so we have $\abs{G}^{\ell - k/2}$ choices. This finishes the proof.
 \end{proof}

Next, we compute the average degree (or number of non-zero entries) in a row/column in $A$.
\begin{observation}
    \label{fact:groupavgdegree}
    For $(U,S) \in I(G, n, \ell)$ we define the graph degree as normal:
    \begin{equation*}
    \deg((U, S)) := \abs{\{(\beta, v) \mid \beta \in G^*, v \in \mcH \text{ s.t. } \exists (V, T) \in I(G, n, \ell), (U,S) \xrightarrow{v, \beta} (V,T)\}} \mper
    \end{equation*}
    Then $\E[\deg(U)] \geq \frac{\abs{G^*}}{2} \left({\frac{\ell}{\abs{G}n}}\right)^{k/2} \cdot \abs{\mcH}$.
\end{observation}

\begin{proof}
    Each $v \in \mcH$ contributes $\abs{G^*}\Delta$ to the total degree, so the average degree is $\E[\deg((U,S))] = \frac{\abs{\mcH}\abs{G^*}\Delta}{N}$. We then have:
    \begin{equation*}
    \E[\deg((U,S))] = \frac{\abs{G^*}\Delta}{N} \cdot \abs{\mcH} = \abs{G^*}\frac{\abs{G}^{\ell-k/2} {k \choose k/2}{n-k \choose \ell-k/2}}{\abs{G}^\ell {n \choose \ell}} \cdot \abs{\mcH} \geq \frac{\abs{G^*}}{2} \left({\frac{\ell}{\abs{G}n}}\right)^{k/2} \cdot \abs{\mcH} \mcom
    \end{equation*}
    where the last inequality follows from \cref{fact:binomest}.
\end{proof}

Unlike in the field case, we do not have the guarantee that for fixed $(U, S)$ and $v$ we have at most one $\beta $-labeled adjacent edge. Instead, this number is categorized by the robustness of the equations.

\begin{observation}
        \label{obs:groupdeg}
        Given $(U,S) \in I(G, n, \ell)$ of a Kikuchi matrix from a $\lambda$-robust equation set and $v \in \mcH$, there can be at most $\frac{1}{\lambda}$ choices of $\beta \in G$ such that there is $(V, T) \in I(G, n, \ell)$ with
        \begin{equation*}
            (U, S) \xrightarrow{v, \beta} (V, T) \mper
        \end{equation*}
\end{observation}

\begin{proof}
    This follows from \cref{obs:groupdivisors}. Suppose there is at least one $\beta$ for which the above holds and that $U$ matches on $v_U$ with representative group $H(v_U) = \bigotimes_{i=1}^r \mu_i \Z_{m_i}$. Note that $\beta^{(i)}$ is then interchangeable within the coset $(m_i/\mu_i)\Z_{m_i} + \beta^{(i)}$, giving $\mu_i$ options to change $\beta^{(i)}$ to without changing the value of $\beta v_U$, thus inducing another edge on $(U, S)$. Since this holds for all $i \in [r]$, we have a total of $\prod_{i=1}^r \mu_i$ options. From here we note $\abs{H_i(v_U)} = m_i/\mu_i$ so rewrite the number to $\prod_{i=1}^r \frac{m_i}{\abs{H_i(v_U)}}$ which is exactly $\frac{1}{H(v_U)}$ so by $\lambda$-robustness is at most $\frac{1}{\lambda}$.
\end{proof}

The following spectral norm bound immediately implies \cref{thm:grouprefutation} in the even case.
\begin{lemma}
    \label{lem:groupcountingbacktrackingwalks}
    Let $A$ be the level-$\ell$ Kikuchi matrix over $G^n$ defined in \cref{def:groupkikuchimatrix} for the $k$-$\LIN(G)$ instance $\mcI = (\mcH, \{b_v\}_{v \in \mcH})$ for some choice of $\{G^*\}_{v \in \mcH}$. Let $\Gamma \in \C^{N \times N}$ be the diagonal matrix $\Gamma = D + d \Id$ where $D_{(U,S), (U,S)} := \deg((U,S))$ and $d = \E[\deg((U,S))]$. Suppose the $b_v$'s are drawn independently and uniformly from $G$ and the instance $\mcI$ is $\lambda$-robust. Then, with probability $\geq 1 - \frac{1}{\poly(n)}$, it holds that
    \begin{equation*}
    \norm{\Gamma^{-1/2} A \Gamma^{-1/2}}_2 \leq O\left(\sqrt{\frac{\ell \log(\abs{G}n)}{\lambda d}}\right) \mper
    \end{equation*}
\end{lemma}

We further postpone the proof of \cref{lem:groupcountingbacktrackingwalks} and now finish the proof of \cref{thm:grouprefutation}.

\begin{proof}[Proof of \cref{lem:grouprefutation} from \cref{lem:groupcountingbacktrackingwalks}]
Let $\mcI = (\mcH, \{b_v\}_{v \in \mcH})$ be the input to the algorithm. Given $\ell$, the algorithm constructs the matrix $A$ for $\mcI$ and computes $\frac{2\abs{G^*}}{\abs{G}} \norm{\tilde{A}}_2$. Let $\Phi(x)$ be the polynomial defined in \cref{obs:groupadvantageoverpolynomial}. For each $x \in \F^n$, letting $y \in \C^n$ be the vector defined in \cref{fact:groupbilinearforms}, we have
    \begin{flalign*}
 &\Phi(x) = \frac{1}{\abs{G}  \abs{\mcH} \Delta} \cdot y^\dagger A y =  \frac{1}{\abs{G}  \abs{\mcH} \Delta} \cdot (\Gamma^{1/2} y)^\dagger \tilde{A} (\Gamma^{1/2} y) \leq \frac{1}{\abs{G}  \abs{\mcH} \Delta} \cdot \norm{\tilde{A}}_2 \norm{\Gamma^{1/2} y}_2^2 \\
 &=  \frac{1}{\abs{G}  \abs{\mcH} \Delta} \cdot \norm{\tilde{A}}_2 \cdot \tr(\Gamma) \leq \frac{2\abs{G^*}}{\abs{G}}\norm{\tilde{A}}_2  \mcom
     \end{flalign*}
     where we use that $\norm{\Gamma^{1/2} y}_2^2 = y^{\dagger} \Gamma y = \sum_{U} (\Gamma)_U \abs{y_U}^2 = \sum_U (\Gamma)_U = \tr(\Gamma)$ since $\abs{y_U} = 1$ for all $U$, and that $\tr(\Gamma) = 2 \abs{\mcH} \abs{G^*} \Delta$. Hence, 
     \begin{equation*}
     \max_{x \in G^n} \Phi(x) \leq \frac{2\abs{G^*}}{\abs{G}}\norm{\tilde{A}}_2  \mcom
     \end{equation*}
     which proves \cref{item:lemgrouprefutation1} in \cref{lem:grouprefutation} for $\mcI$. To prove \cref{item:lemgrouprefutation2}, we observe that by \cref{lem:groupcountingbacktrackingwalks}, since $\mcI$ is semirandom and $\lambda$-robust then with high probability over the draw of the $b_v$'s, it holds that
     \begin{equation*}
        \norm{\tilde{A}}_2  \leq O\left(\sqrt{\frac{\ell \log(\abs{G}n)}{\lambda d}}\right) \mcom
     \end{equation*}
     From \cref{fact:groupavgdegree}, we have $d \geq \frac{\abs{G^*}}{2} \left(\frac{\ell}{\abs{G}n}\right)^{k/2} \cdot \abs{\mcH}$ and can plug this in to get
     \begin{equation*}
        \frac{2\abs{G^*}}{\abs{G}}\norm{\tilde{A}}_2 \leq O(1) \cdot \sqrt{\frac{\ell \log(\abs{G}n)}{\lambda \abs{G}\left(\frac{\ell}{\abs{G}n}\right)^{k/2} \abs{\mcH}}}\mper
     \end{equation*}
     If $\abs{\mcH} \geq C \frac{n}{\lambda} \log(\abs{G}n) \left(\frac{\abs{G}n}{\ell}\right)^{k/2-1} \varepsilon^{-2}$ for a sufficiently large constant $C$, then $\frac{2\abs{G^*}}{\abs{G}}\norm{\tilde{A}}_2 \leq \varepsilon$ with probability $1 - 1/\poly(n)$, which is good enough for \cref{item:lemgrouprefutation2}.
\end{proof}

\begin{proof}[Proof of \cref{lem:groupcountingbacktrackingwalks}]

Since the matrix is Hermitian, we have that $\norm{\tilde{A}}_2 \leq \tr((\Gamma^{-1} A)^{2t})^{1/2t}$ for any positive integer $t$. Because the $b_v$'s are drawn independently from $\F$, the matrix $\tilde{A}$ is a random matrix. By Markov's inequality,
\begin{equation*}
\Pr\left[\tr((\Gamma^{-1} A)^{2t}) \geq N \cdot \E[\tr((\Gamma^{-1} A)^{2t})]\right] \leq \frac{1}{N} \mper
\end{equation*}
We note this event is the same as $\tr((\Gamma^{-1} A)^{2t})^{1/2t} \geq N^{1/2t} \cdot \E[\tr((\Gamma^{-1} A)^{2t})]^{1/2t}$, and for $2t \geq \log N $ we have $N^{1/2t} \leq O(1)$. This immediately gives us that with probability $\geq 
1- \frac{1}{N}$, $\norm{\tilde{A}}_2 \leq O\left(\E[\tr((\Gamma^{-1} A)^{2t})]^{1/2t}\right)$. We then have that
\begin{flalign*}
    \E\left[\tr\left(\left(\Gamma^{-1} A\right)^{2t}\right)\right] 
    &= \E\left[\tr\left(\left(\Gamma^{-1} \sum_{v \in \mcH, \beta \in G^*} c_{v, \beta} \cdot A_{v, \beta}\right)^{2t}\right) \right]\\
    &= \E\left[\tr\left(\sum_{(v_1, \beta_1) ,..., (v_{2t}, \beta_{2t}) \in \mcH \times G^*} \prod_{i = 1}^{2t} \Gamma^{-1} \cdot c_{v_i, \beta_i} \cdot A_{v_i, \beta_i} \right) \right]\\
    &= \sum_{(v_1, \beta_1) ,..., (v_{2t}, \beta_{2t}) \in \mcH \times G^*} \E\left[\tr\left(\prod_{i = 1}^{2t} \Gamma^{-1} \cdot c_{v_i, \beta_i} \cdot A_{v_i, \beta_i} \right) \right] \\
        &= \sum_{(v_1, \beta_1) ,..., (v_{2t}, \beta_{2t}) \in \mcH \times G^*} \E\left[\prod_{i=1}^{2t} c_{v_i, \beta_i} \right] \cdot \tr\left(\prod_{i = 1}^{2t} \Gamma^{-1}A_{v_i, \beta_i} \right)\mper
    \end{flalign*}
    Let us now make the following observation. Let $(v_1, \beta_1) ,..., (v_{2t}, \beta_{2t}) \in \mcH \times G^*$ be a term in the above sum. Fix $v \in \mcH$, and let $R(v)$ denote the set of $i \in [2t]$ such that $v_i = v$. We observe that if for some $v \in \mcH$, $\sum_{i \in R(v)} \beta_i \neq 0$, then $\E\left[\prod_{i=1}^{2t} c_{v_i, \beta_i} \right] = 0$. Indeed, this is because $b_v$ is independent for each $v \in \mcH$, and so $\E\left[\prod_{i=1}^{2t} c_{v_i, \beta_i} \right]  = \prod_{v \in \mcH} \E\left[\prod_{i \in R(v)} c_{v, \beta_i} \right]$, and 
    \begin{equation*}
    \E\left[\prod_{i \in R(v)}c_{v, \beta_i} \right] = \E\left[\prod_{i \in R(v)} \chi_{\beta_i}(b_v)\right] = \E\left[\chi_{\sum_{i \in R(v)} \beta_i}(b_v)\right]\mper
    \end{equation*}
    Then, since $b_v$ is uniform from $G$, it follows that $\E\left[\chi_{\sum_{i \in R(v)} \beta_i}(b_v)\right] = 0$ if $\sum_{i \in R(v)} \beta_i \neq 0$, and $1$ if $\sum_{i \in R(v)} \beta_i = 0$. This motivates the following definition.
    \begin{definition}[Trivially closed sequence]
    \label{def:grouptriviallyclosedwalks}
    Let $(v_1, \beta_1) ,..., (v_{2t}, \beta_{2t}) \in \mcH \times G^*$. We say that $(v_1, \beta_1) ,..., (v_{2t}, \beta_{2t}) \in \mcH \times G^*$ is trivially closed with respect to $v$ if it holds that $\sum_{i \in R(v)} \beta_i = 0$. We say that the sequence is trivially closed if it is trivially closed with respect to all $v \in \mcH$.
    \end{definition}
    With the above definition in hand, we have shown that
    \begin{flalign*}
      & \E\left[\tr\left(\left(\Gamma^{-1} A\right)^{2t}\right)\right] 
        = \sum_{\substack{(v_1, \beta_1) ,..., (v_{2t}, \beta_{2t}) \\ \text{trivially closed}}} \tr\left(\prod_{i = 1}^{2t} \Gamma^{-1}A_{v_i, \beta_i} \right)\mper
    \end{flalign*}

The following lemma yields the desired bound on $\E[\tr((\Gamma^{-1} A)^{2t})]$.
\begin{lemma}
    \label{lem:groupmaincountingbacktrackingwalks}
    $\sum_{\substack{(v_1, \beta_1) ,..., (v_{2t}, \beta_{2t}) \\ \text{trivially closed}}} \tr\left(\prod_{i = 1}^{2t} \Gamma^{-1}A_{v_i, \beta_i} \right) \leq N \cdot 2^{2t} \cdot \left(\frac{2t}{\lambda d}\right)^t$.
\end{lemma}
With \cref{lem:groupmaincountingbacktrackingwalks}, we thus have the desired bound $\E[\tr((\Gamma^{-1} A)^{2t})]$. Taking $t$ to be $c \log_2 N$ for a sufficiently large constant $c$ and applying Markov's inequality finishes the proof.
\end{proof}

\begin{proof}[Proof of  \cref{lem:groupmaincountingbacktrackingwalks}]
    We bound the sum as follows. First, we observe that for a trivially closed sequence $(v_1, \beta_1) ,..., (v_{2t}, \beta_{2t})$, we have for pairs $P_i \in I(G, n, \ell)$
    \begin{flalign*}
    \tr\left(\prod_{i = 1}^{2t} \Gamma^{-1}A_{v_i, \beta_i} \right) = \sum_{P_0, P_1, \dots, P_{2t-1}} \prod_{i = 0}^{2t - 1} \Gamma^{-1}_{P_i} \cdot \1\left(P_i\xrightarrow{\text{$v_{i+1}, \beta_{i+1}$}} P_{i+1}\right) \mcom
    \end{flalign*}
    where we define $P_{2t} = P_0$. Thus, the sum that we wish to bound in \cref{lem:groupmaincountingbacktrackingwalks} simply counts the total weight of ``trivially closed walks'' $P_0, v_1, \beta_1, P_1, \dots, P_{2t-1}, v_{2t}, \beta_{2t}, P_{2t}$ (where $P_{2t} = P_0$) in the Kikuchi graph $A$, where the weight of a walk is simply $\prod_{i = 0}^{2t-1} \Gamma^{-1}_{P_i}$.
    
    Let us now bound this total weight by uniquely encoding a walk $P_0, v_1, \beta_1, U_1, \dots, P_{2t-1}, v_{2t}, \beta_{2t}, P_{2t}$ as follows.
    \begin{itemize}
        \item First, we write down the start vertex $P_0 = (U_0, S_0)$.
        \item For $i = 1, \dots, 2t$, we let $z_i$ be $1$ if $v_i = v_j$ for some $j < i$. In this case, we say that the edge is ``old''. Otherwise $z_i = 0$, and we say that the edge is ``new''.
	\item For $i = 1, \dots, 2t$, if $z_i$ is $1$ then we encode $P_i = (U_i, S_i)$ by writing down the smallest $j \in [2t]$ such that $v_i = v_j$ and $\tau \in [1/\lambda]$ denoting $\beta_i$. Note that \cref{obs:groupdeg} and the $\lambda$-robust assumption guarantees there are at most $1/\lambda$ choices for what $\beta_i$ can be out of a fixed $(U_i, S_i)$.
	\item For $i = 1, \dots, 2t$, if $z_i$ is $0$ then we encode $P_i$ by writing down an integer in $1, \dots, \deg(P_{i-1})$ that specifies the edge we take to move to $P_i$ from $P_{i-1}$ (we associate $[\deg(P_{i-1})]$ to the edges adjacent to $P_{i-1}$ with an arbitrary fixed map).
    \end{itemize}
    With the above encoding, we can now bound the total weight of all trivially closed walks as follows. First, let us consider the total weight of walks for some fixed choice of $z_1, \dots, z_{2t}$. We have $N$ choices for the start vertex $U_0$. For each $i = 1, \dots, 2t$ where $z_i = 0$, we have $\deg(U_{i-1})$ choices for $U_i$, and we multiply by a weight of $\Gamma^{-1}_{U_{i-1}} \leq \frac{1}{\deg(U_{i-1})}$. For each $i = 1, \dots, 2t$ where $z_i = 1$, we have at most $2t$ choices for the index $j < i$ and $1/\lambda$ choices for $\beta_i$, and we multiply by a weight of $\Gamma^{-1}_{U_{i-1}} \leq \frac{1}{d}$. Hence, the total weight for a specific $z_1, \dots, z_{2t}$ is at most $N \left(\frac{2t}{\lambda d}\right)^{r}$, where $r$ is the number of $z_i$ such that $z_i = 1$.
    
    Finally, we observe that any trivially closed walk must have $r \geq t$. Hence, after summing over all $z_1, \dots, z_{2t}$, we have the final bound of $N 2^{2t} \left(\frac{2t}{\lambda d}\right)^{t}$, which finishes the proof.
\end{proof}

\section{Refuting Semirandom \texorpdfstring{$k$-$\LIN$}{k-LIN} over Fields/Abelian Groups for Odd $k$}
\label{sec:mainrefutation}

In this section, we prove \cref{thm:refutation} and \cref{thm:grouprefutation} for the case of odd $k$. We primarily tailor our exposition to the finite field case of \cref{thm:refutation} with interspersed commentary on where the proof of \cref{thm:grouprefutation} diverges. 

Our refutation algorithm follows the framework set out in \cite{GuruswamiKM22, HsiehKM23}, once again utilizing a generalization of the Kikuchi matrix for arbitrary finite fields $\F$. Since naive reductions or modifications to the even-arity case seem to fail, this Kikuchi matrix differs greatly and requires additional setup.

\subsection{Refuting bipartite polynomials}

We begin our proof by defining a structured type of $k$-$\LIN$ instance that lends itself well to the more delicate odd-arity Kikuchi matrix.

\begin{definition}[$\mcU$-bipartite decompositions of vector sets]
    A $t$-sparse $\mcU$-bipartite vector set $\mcH$ over $\F^n$ is a collection $\{\mcH_u\}_{u \in \mcU}$ where $\mcU$ is a set of $t$-sparse vectors in $\F^n$ and each $\mcH_u$ is a set of vectors $v$ in $\F^n$ with the property that $u \sqsubseteq v$. Given a vector set $\mcH$ over $\F^n$ with a partition $\{\mcH_u\}_{u \in \mcU}$ forming a $t$-sparse $\mcU$-bipartite vector set, we say that this partition forms a $t$-sparse $\mcU$-bipartite decomposition for $\mcH$. By slight abuse of notation, we will refer to both the ambient set of vectors and the collection $\{\mcH_u\}_{u \in \mcU}$ as $\mcH$.
\end{definition}

\begin{definition}[$(\varepsilon, \ell)$-regularity in $\mcU$-bipartite sets]
    Given a $\mcU$-bipartite vector set $\{\mcH_u\}_{u \in \mcU}$, a partition $\mcH_u$ of vectors in $\F^n$ is $(\varepsilon, \ell)$-regular if there does not exist non-zero $w \in \F^n$ with $\wt(w) > \wt(u)$ and a subset $\mcH'$ with the property that $w \sqsubseteq v$ for all $v \in \mcH'$ and $\abs{\mcH'} > \max\left(\left(\frac{n \abs{\F}}{\ell}\right)^{k/2-1-\wt(w)}, 1\right) \cdot \varepsilon^{-2}$. The entire collection $\{\mcH_u\}_{u \in \mcU}$ is said to be $(\varepsilon, \ell)$-regular if all partitions $\mcH_u$ are $(\varepsilon, \ell)$-regular. If $\varepsilon = 1$, we abbreviate to just $\ell$-regular.
\end{definition}

Thinking of the vectors as representing coefficients of $k$-$\LIN$ equations, the above can be informally interpreted as a decomposition of equations such that within each partition $\mcH_u$, there is no larger subequation than $u$ that overlaps many equations in $\mcH_u$.

Given a $k$-$\LIN$ instance $\mcI$ and a bipartite decomposition of its underlying equation set, the following trick gives us a new way to bound $\val(\mcI)$.

\begin{lemma}[Cauchy-Schwarz Trick]
    \label{obs:bipartitepolynomial}
    Recall for a $k$-$\LIN(\F)$ instance $\mcI = (\mcH, \{b_v\}_{v \in \mcH})$ we relate its satisfiability to the polynomial
    \begin{equation*}
        \val(\mcI, x) = \frac{1}{\abs{\F}} + \frac{1}{\abs{\mcH} \abs{\F}}\sum_{v \in \mcH} \sum_{\beta \in \F^*} \omega_p^{\Tr(\beta b_v)} \cdot \overline{\chi_{\beta v}(x)} := \frac{1}{\abs{\F}} + \Phi(x)\mper
    \end{equation*}
    Given a $\mcU$-bipartite decomposition $\mcH = \cbra{\mcH_u}_{u \in \mcU}$ we have the bound
    \begin{equation*}
        \abs{\Phi(x)}^2 \leq \frac{\abs{\mcU}}{\abs{\mcH}^2\abs{\F}} \sum_{u \in \mcU} \sum_{\beta \in \F^*} \sum_{v, v' \in \mcH_u} \omega_p^{\Tr(\beta b_{v})} \cdot \omega_p^{-\Tr(\beta b_{v'})} \cdot \overline{\chi_{\beta v}(x)} \cdot \chi_{\beta v'}(x)\mper
    \end{equation*}
    We refer to polynomials in this form as bipartite polynomials.
\end{lemma}

\begin{proof}[Proof of \cref{obs:bipartitepolynomial}]
    First we partition the sum across $\mcH$ based on the decomposition $\cbra{\mcH_u}_{u \in \mcU}$. For any $v \in \mcH_u$, we can rewrite a term $\overline{\chi_{\beta v}(x)} = \overline{\chi_{\beta u}(x)} \cdot \overline{\chi_{\beta v-\beta u}(x)}$. This allows us to pull out the common term $\overline{\chi_{\beta u}(x)}$ across the entire partition $\mcH_u$.
    \begin{align*}
        \abs{\Phi(x)}^2 &= \frac{1}{\abs{\mcH}^2 \abs{\F}^2} \left|\sum_{v \in \mcH} \sum_{\beta \in \F^*} \omega_p^{\Tr(\beta b_v)} \cdot \overline{\chi_{\beta v}(x)}\right|^2 \\
        &= \frac{1}{\abs{\mcH}^2 \abs{\F}^2} \left|\sum_{\beta \in \F^*} \sum_{u \in \mcU}\overline{\chi_{\beta u}(x)} \sum_{v \in \mcH_u}\omega_p^{\Tr(\beta b_{v})} \cdot \overline{\chi_{\beta (v-u)}(x)}\right|^2\mper
    \end{align*}

    Now we can apply the Cauchy-Schwarz inequality to the above quantity in such a way that the shared term $\overline{\chi_{\beta u}(x)}$ cancels.
    \begin{align*}
        \left|\sum_{\beta \in \F^*, u \in \mcU}\overline{\chi_{\beta u}(x)} \sum_{v \in \mcH_u}\omega_p^{\Tr(\beta b_v)} \cdot \overline{\chi_{\beta (v-u)}(x)}\right|^2 
        &\leq \sum_{\beta \in \F^*, u \in \mcU} \left|\overline{\chi_{\beta u}(x)}\right|^2 \sum_{\beta \in \F^*, u \in \mcU} \left|\sum_{v \in \mcH_u} \omega_p^{\Tr(\beta b_{v})} \cdot \overline{\chi_{\beta (v-u)}(x)}\right|^2 \\
        &= \abs{\F^*} \abs{\mcU} \sum_{\beta \in \F^*, u \in \mcU} \sum_{v, v' \in \mcH_u}  \omega_p^{\Tr(\beta b_{v})} \cdot \omega_p^{-\Tr(\beta b_{v'})} \cdot \overline{\chi_{\beta (v-u)}(x)} \cdot \chi_{\beta (v'-u)}(x)\mper
    \end{align*}
    Notice that we may rewrite $\overline{\chi_{\beta (v - u)}(x)} = \overline{\chi_{\beta v}(x)} \cdot \chi_{\beta u}(x)$ and $\chi_{\beta (v' - u)}(x) = \chi_{\beta v'}(x) \cdot \overline{\chi_{\beta u}(x)}$. Since $\chi_u(x) \cdot \overline{\chi_u(x)} = 1$, we have $\overline{\chi_{\beta (v-u)}(x)} \cdot \chi_{\beta (v'-u)}(x) = \overline{\chi_{\beta v}(x)} \cdot \chi_{\beta v'}(x)$ and after combining the above this rewrite finishes the claim.
\end{proof}

Given any instance $\mcI$, we can always apply \cref{obs:bipartitepolynomial} using the trivial bipartite decomposition of $\mcH$ as itself, but since this decomposition is not guaranteed to be regular, our refutation algorithm and resulting analysis end up failing. Instead, our goal is to find a non-trivial regular bipartite decomposition to apply the trick to. With this roadmap in mind, we are ready to state our main technical results to prove \cref{thm:refutation} in full.

\begin{proof}[Proof of \cref{thm:refutation}]
    To achieve $\varepsilon$-refutation of a $k$-$\LIN$ instance $\mcI = (\mcH, \{b_v\}_{v \in \mcH})$, we partition $\mcI$ in to $k$ subinstances $\cbra{\mcI^{(t)}}_{t \in [k]}$ and refute all simultaneously. In this context, a subinstance $\mcI^{(t)}$ is a subset $\mcH^{(t)}$ of equations from $\mcH$ and the corresponding right-hand sides. We observe that given a collection of subinstances and an $\varepsilon^{(t)}$-refutation of each we have
    \begin{equation*}
        \val(\mcI, x) = \sum_{t \in [k]} \frac{\abs{\mcH^{(t)}}}{\abs{\mcH}} \val(\mcI^{(t)}, x) \leq \sum_{t \in [k]} \frac{\abs{\mcH^{(t)}}}{\abs{\mcH}} \left(\frac{1}{\abs{\F}} + \varepsilon^{(t)}\right) = \frac{1}{\abs{\F}} + \sum_{t \in [k]} \frac{\abs{\mcH^{(t)}}}{\abs{\mcH}}\varepsilon^{(t)} \mcom
    \end{equation*}
    using $\sum_{t=1}^k \frac{\abs{\mcH^{(t)}}}{\abs{\mcH}} = 1$. We aim to bound the contribution of each $\frac{\abs{\mcH^{(t)}}}{\abs{\mcH}}\varepsilon^{(t)} \leq \frac{\varepsilon}{k}$. In order to tightly refute each subinstance, we need the regularity guaranteed by the following algorithm.

    \begin{lemma}[Regular decomposition algorithm for $k$-$\LIN$ instances]
        \label{lem:decompositionalg} 
        Fix $k \geq 2$. There is an algorithm that takes as input a $k$-$\LIN(\F)$ instance $\mcI = (\mcH, \{b_v\}_{v \in \mcH})$ and outputs a set of instances $\mcI^{(t)} = (\mcH^{(t)}, \{b_v\}_{v \in \mcH^{(t)}})$ each equivalent \footnote{Equivalent here means that we may scale constraints. For example, we could take $x_1 + 2x_2 = 1$ in $\F_3$ and replace it with $2x_1 + x_2 = 2$, which is scaled by $2$. Vitally, this never changes satisfiability.} to a subinstance in a fixed partition of $\mcI$ and $t$-sparse $\mcU^{(t)}$-bipartite decompositions $\{\mcH^{(t)}_u\}_{u \in \mcU^{(t)}}$ for each $\mcH^{(t)}$ in time $\left(\abs{\F^*}n\right)^{O(\ell)}$ with the guarantees:
        \begin{enumerate}
            \item \label{item:decompositionalg2} For $t \neq 1$ and all $u \in \mcU^{(t)}$, $\abs{\mcH^{(t)}_u} = \tau_t := \max(1, \left(\frac{n\abs{\F^*}}{\ell}\right)^{k/2-t}) \cdot 4k^2\varepsilon^{-2}$.
            \item \label{item:decompositionalg3} For all $u \in \mcU^{(1)}$, $\abs{\mcH^{(1)}_u} \leq \tau_1$.
            \item \label{item:decompositionalg4} For all $t$, $\abs{\mcU^{(t)}} \leq \frac{2\abs{\mcH}}{\tau_t}$.
            \item \label{item:decompositionalg5} $\mcH$ is $\left(\frac{\varepsilon}{2k}, \ell\right)$-regular.
        \end{enumerate}
    \end{lemma}

We prove \cref{lem:decompositionalg} in \cref{sec:vectorsetreg}. Using the significant structure guaranteed on the subinstances $\mcI^{(t)}$ identified by this algorithm, namely the existence of an $\left(\frac{\varepsilon}{2k}, \ell\right)$-regular decomposition $\{\mcH^{(t)}_u\}_{u \in \mcU^{(t)}}$, we provide sufficient refutations on the components arising from \cref{obs:bipartitepolynomial}. More formally, we have that the advantage polynomial is
\begin{equation*}
    \val(\mcI, x) = \frac{1}{\abs{\F}} + \frac{1}{\abs{\mcH}\abs{\F}} \sum_{v \in \mcH^{(t)}} \sum_{\beta \in \F^*} \omega_p^{\Tr(\beta b_v)} \cdot \overline{\chi_{\beta v}(x)} := \frac{1}{\abs{\F}} + \Phi_t(x) \mper
\end{equation*}
To limit the contribution, we scale by $\frac{\abs{\mcH^{(t)}}}{\abs{\mcH}}$ and show $\frac{\abs{\mcH^{(t)}}}{\abs{\mcH}} \max_{x \in \F^n} \abs{\Phi_t(x)} \leq \frac{\varepsilon}{k}$. Since each $\mcI^{(t)}$ output above has a regular decomposition, we can apply \cref{obs:bipartitepolynomial} directly to write
\begin{equation*}
    \frac{k^2\abs{\mcH^{(t)}}^2}{\abs{\mcH}^2}\abs{\Phi_t(x)}^2 \leq \frac{k^2\abs{\mcU^{(t)}}}{\abs{\mcH}^2\abs{\F}}\sum_{u \in \mcU^{(t)}} \sum_{\beta \in \F^*} \sum_{v, v' \in \mcH_u^{(t)}} \omega_p^{\Tr(\beta b_{v})} \cdot \omega_p^{-\Tr(\beta b_{v'})} \cdot \overline{\chi_{\beta v}(x)} \cdot \chi_{\beta v'}(x) \mper
\end{equation*}
To finish the proof, it suffices to show the latter term is bound by $\varepsilon^2$, which we accomplish through the following refutation algorithm.

\begin{lemma}[$\varepsilon^2$-refutation of semirandom bipartite polynomials]
    \label{lem:mainrefutation}
    Fix $k \geq 2$, $\ell \geq k/2$, and $1 \leq t \leq k$. There is an algorithm that takes as input an instance $\mcI^{(t)} = (\mcH^{(t)}, \{b_v\}_{v \in \mcH^{(t)}})$ with an $\mcU^{(t)}$-bipartite decomposition $\mcH^{(t)} = \{\mcH^{(t)}_u\}_{u \in \mcU^{(t)}}$ coming from $\mcH$ describing an associated polynomial $\Psi_t = \frac{k^2\abs{\mcU^{(t)}}}{\abs{\mcH}^2\abs{\F}} \sum_{u \in \mcU^{(t)}} \sum_{\beta \in \F^*} \sum_{v, v' \in \mcH_u^{(t)}} \omega_p^{\Tr(\beta b_{v})} \cdot \omega_p^{-\Tr(\beta b_{v'})} \cdot \overline{\chi_{\beta v}} \cdot \chi_{\beta v'}$ and outputs a certificate $\algval(\Psi_t) \in \R$ in time $(\abs{\F^*}n)^{O(\ell)}$ with the guarantees:
    \begin{enumerate}
        \item \label{item:ref1} $\algval(\Psi_t) \geq \max_{x \in \F^n} \abs{\Psi_t(x)}$.
        \item \label{item:ref2} $\algval(\Psi_t) < \varepsilon^2$ with high probability (over the randomness in $b_{v}$) given:
        \begin{enumerate}
            \item \label{item:ref2a} The hypothesis of \cref{item:refutation2} in \cref{thm:refutation} holds.
            \item \label{item:ref2b} The output guarantees of \cref{lem:decompositionalg} hold.
        \end{enumerate}
    \end{enumerate}
\end{lemma}

We delay the proof of \cref{lem:mainrefutation} to \cref{sec:kikuchirefutation}. The algorithm for \cref{thm:refutation} then falls into place. First, we run the algorithm in \cref{lem:decompositionalg} which gives a set of regular subinstances $\{\mcI^{(t)}\}_{t \in [k]}$ with regular bipartite decompositions. We then run the algorithm of \cref{lem:mainrefutation} on each, and by \cref{item:ref2} succeeds in showing $\algval(\Psi_t) < \varepsilon^2$ for each $\mcI^{(t)}$ with high probability given the regularity properties of \cref{lem:decompositionalg} and semirandomness. Using \cref{item:ref1} then bounds the true value of $\Psi_t$.
\end{proof}

\begin{tcolorbox}[
    width=\textwidth,   
    colframe=black,  
    colback=white,   
    title=Semirandom $k$-$\LIN(\F)$ Refutation Algorithm,
    colbacktitle=white, 
    coltitle=black,      
    fonttitle=\bfseries,
    center title,   
    enhanced,       
    frame hidden,           
    borderline={1pt}{0pt}{black},
    sharp corners,
    toptitle=2.5mm
]
\textbf{Input:} A $k$-$\LIN(\F)$ instance $\mcI = (\mcH, \{b_v\}_{v \in \mcH})$.\\

\textbf{Output:} $\algval(\mcI) \in [0,1]$ with guarantee $\algval(\mcI) \geq \max_{x \in \F^n} \val(\mcI, x)$.\\

\textbf{Algorithm:}

\begin{enumerate}
    \item Using the algorithm $\mathcal{A}$ specified in \cref{lem:decompositionalg} let $\cbra{\mcI^{(t)}}_{t \in [k]} := \mathcal{A}(\mcH)$.
    \item For each $t \in [k]$, using the algorithm $\mathcal{B}$ specified in \cref{lem:mainrefutation} let $\textrm{adv}(\mcI^{(t)}) := \sqrt{\mathcal{B}\left(\mcI^{(t)}\right)}$.
    \item Output $\algval(\mcI) := \frac{1}{\abs{\F}} + \sum_{t \in [k]} \textrm{adv}(\mcI^{(t)})$.
\end{enumerate}

\end{tcolorbox}

\subsection{Vector set regularity algorithm}
\label{sec:vectorsetreg}

\begin{tcolorbox}[
    width=\textwidth,   
    colframe=black,  
    colback=white,   
    title=$k$-$\LIN(\F)$ Regularity Decomposition Algorithm,
    colbacktitle=white, 
    coltitle=black,      
    fonttitle=\bfseries,
    center title,   
    enhanced,       
    frame hidden,           
    borderline={1pt}{0pt}{black},
    sharp corners,
    toptitle=2.5mm
]
\textbf{Input:} A $k$-$\LIN(\F)$ instance $\mcI = (\mcH, \{b_v\}_{v \in \mcH})$.\\

\textbf{Output:} A set of instances $\cbra{\mcI^{(t)}}_{t \in [k]}$ satisfying the criteria of \cref{lem:decompositionalg}.\\

\textbf{Algorithm:}
\begin{enumerate}
    \item Let $t = k$, and while $\exists u \in \F^n$ with $\wt(u) = t$ such that $\abs{\{\beta v \in \mcH \mid \beta \in \F^*, u \sqsubseteq \beta v\}} \geq \tau_t := \max(1, (\frac{n \abs{\F^*}}{\ell})^{k/2-t}) \cdot 4k^2 \varepsilon^{-2}$, do the following. Otherwise, decrement $t$.
    \begin{enumerate}
        \item Let $\mcH_u^{(t)}$ hold $\beta v$ for exactly $\tau_t$ such vectors and move $\mcH_u^{(t)}$ from $\mcH$ to $\mcH^{(t)}$.
        \item For each $v \in \mcH^{(t)}_u$, set $b_{\beta v} = \beta b_v$.
    \end{enumerate}
    \item When $t = 0$, for all remaining $v \in \mcH$:
    \begin{enumerate}
        \item Identify the first $i \in [n]$ such that $v_i \neq 0$.
        \item Let $\beta = v_i^{-1}$ and add $\beta v$ to $\mcH^{(1)}_{e_i}$.
    \end{enumerate}
\end{enumerate}

\end{tcolorbox}

\begin{proof}[Proof of \cref{lem:decompositionalg}]
    The $t$-sparsity of $\mcU^{(t)}$ and \cref{item:decompositionalg2} follow simply from the loop condition, which enforces that only $\tau_t$ sized sets with $\wt(u) = t$ get added to $\mcH^{(t)}$. For \cref{item:decompositionalg3}, note that by the greediness of the algorithm any set added outside the for loop must have size $< \tau_1$ otherwise it would have been added within.
    
    For \cref{item:decompositionalg4}, assume $t > 1$ and note that by \cref{item:decompositionalg2} we have $\abs{\mcU^{(t)}} \leq \frac{\abs{\mcH}}{\tau_t}$. When $t = 1$, we have two kinds of sets $\mcH_u$, ones added in the for loop and those added outside. The number added in the for loop is $\leq \frac{\abs{\mcH}}{\tau_1}$ following the case above. The number added outside is at most $n$ counting possible shared entries. Since $\abs{\mcH} \geq C n \cdot \log(\abs{\F^*}n) \left(\frac{n \abs{\F^*}}{\ell}\right)^{k/2-1} \cdot 4k^2 \varepsilon^{-2}$ and $\tau_1 = \left(\frac{n \abs{\F^*}}{\ell}\right)^{k/2-1} \cdot 4k^2\varepsilon^{-2}$ we have $n \leq \frac{\abs{\mcH}}{\tau_1}$ for $C \geq 1$, which gives in total $\abs{\mcU^{(1)}} \leq \frac{2 \abs{\mcH}}{\tau_1}$.

    \cref{item:decompositionalg5} follows by the greediness of the algorithm. Assume for sake of contradiction there is some partition $\mcH^{(t)}$ which is not $\left(\frac{\varepsilon}{2k}, \ell\right)$-regular. By definition, there is some $\mcH' \subseteq \mcH_u^{(t)}$ and $w \in \F^n$ such that all $v \in \mcH'$ have $w \sqsubseteq v$ and moreover $\wt(w) = t' > t$ and $\abs{\mcH'} \geq \tau_{t'}$. Since iteration $t'$ happens before $t$, such a set would have been available on iteration $t'$ and would have been added then instead, a contradiction.
\end{proof}

\begin{groupcase}
    In the general Abelian case, we cannot necessarily invert elements as we do in Step 2. To get around this, we need to group zero divisors separately. Luckily, the only case we have to do this is when we are embeeded in a group $H$ with $\abs{H} \leq \frac{1}{\varepsilon^2}$, so this only induces $\varepsilon^{-2}$ more buckets, which is absorbed into the equation density.
\end{groupcase}

\subsection{Odd-arity Kikuchi matrices}
\label{sec:kikuchirefutation}

We now prove the main technical component, \cref{lem:mainrefutation}. The proof proceeds as in \cref{sec:refutation} by constructing a generalization of the odd-arity Kikuchi matrix in \cite{GuruswamiKM22, HsiehKM23}. At a high-level, the proof is much the same, reducing to counting certain walks in our defined Kikuchi matrix, but some more complicated design choices in the Kikuchi matrix and an extra degree-regularization step are needed, the latter of which we use the regularity of the instances.

\begin{lemma}[$\varepsilon^2$-refutation of semirandom bipartite polynomials (\cref{lem:mainrefutation} restated)]
    Fix $k \geq 2$, $\ell \geq k/2$, and $1 \leq t \leq k$. There is an algorithm that takes as input an instance $\mcI^{(t)} = (\mcH^{(t)}, \{b_v\}_{v \in \mcH^{(t)}})$ with a $t$-sparse $\mcU^{(t)}$-bipartite decomposition $\mcH^{(t)} = \{\mcH^{(t)}_u\}_{u \in \mcU^{(t)}}$ coming from $\mcH$ describing an associated polynomial $\Psi_t = \frac{k^2\abs{\mcU^{(t)}}}{\abs{\mcH}^2\abs{\F}} \sum_{u \in \mcU^{(t)}} \sum_{\beta \in \F^*} \sum_{v, v' \in \mcH_u^{(t)}} \omega_p^{\Tr(\beta b_{v})} \cdot \omega_p^{-\Tr(\beta b_{v'})} \cdot \overline{\chi_{\beta v}} \cdot \chi_{\beta v'}$ and outputs a certificate $\algval(\Psi_t) \in \R$ in time $(\abs{\F^*}n)^{O(\ell)}$ with the guarantees:
    \begin{enumerate}
        \item $\algval(\Psi_t) \geq \max_{x \in \F^n} \abs{\Psi_t(x)}$.
        \item $\algval(\Psi_t) < \varepsilon^2$ with high probability (over the randomness in $b_{v}$) given:
        \begin{enumerate}
            \item \label{item:newref2a} The hypothesis of \cref{item:refutation2} in \cref{thm:refutation} holds.
            \item \label{item:newref2b} The following criteria holds for $\tau_t = \max(1, \left(\frac{n\abs{\F^*}}{\ell}\right)^{k/2-t}) \cdot 4k^2\varepsilon^{-2}$.
            \begin{enumerate}
                \item For all $u \in \mcU^{(t)}$, $\abs{\mcH^{(t)}_u} \leq \tau_t$.
                \item $\abs{\mcU^{(t)}} \leq \frac{2\abs{\mcH}}{\tau_t}$.
                \item $\mcH$ is $\left(\frac{\varepsilon}{2k}, \ell\right)$-regular.
            \end{enumerate}
        \end{enumerate}
    \end{enumerate}
\end{lemma}

    Note first that if $v = v'$ in the sum then $\omega_p^{\Tr(\beta b_{v})} \cdot \omega_p^{-\Tr(\beta b_{v'})} \cdot \overline{\chi_{\beta v}} \cdot \chi_{\beta v'} = 1$, so we trivially have
    \begin{flalign*}
        &\Psi_t(x) = \frac{k^2\abs{\mcU^{(t)}}}{\abs{\mcH}^2\abs{\F}} \sum_{u \in \mcU^{(t)}} \abs{\mcH^{(t)}_u} + \frac{k^2\abs{\mcU^{(t)}}}{\abs{\mcH}^2\abs{\F}} \sum_{u \in \mcU^{(t)}} \sum_{\beta \in \F^*} \sum_{v \neq v' \in \mcH^{(t)}_u} \omega_p^{\Tr(\beta b_{v})} \cdot \omega_p^{-\Tr(\beta b_{v'})} \cdot \overline{\chi_{\beta v}} \cdot \chi_{\beta v'}\\
        &= \frac{k^2\abs{\mcH^{(t)}}\abs{\mcU^{(t)}}}{\abs{\mcH}^2\abs{\F}} + \frac{k^2\abs{\mcU^{(t)}}}{\abs{\mcH}^2\abs{\F}} \sum_{u \in \mcU^{(t)}} \sum_{\beta \in \F^*} \sum_{v \neq v' \in \mcH^{(t)}_u} \omega_p^{\Tr(\beta b_{v})} \cdot \omega_p^{-\Tr(\beta b_{v'})} \cdot \overline{\chi_{\beta v}} \cdot \chi_{\beta v'}\mper
    \end{flalign*}

    To achieve \cref{item:ref1} we compute the former term directly and provide a spectral certificate on the latter term in the vein of \cref{lem:countingbacktrackingwalks}. We begin by giving our new Kikuchi matrix definition here.

\begin{definition}{(Odd-arity Kikuchi matrix over $\F$).}
\label{def:oddkikuchimatrix}
Let $k/2 \leq \ell \leq n/2$ be a parameter and let $N = \abs{\F^*}^\ell {2n \choose \ell}$. For each pair $(v, v')$ of $(k-t)$-sparse vectors in $\F^n$ and $\beta \in \F^*$, we define a matrix $A_{v, v', \beta} \in \C^{N \times N}$ as follows. First, we identify $N$ with the set of pairs of vectors $(U^{(1)}, U^{(2)})$ in $\F^n$ with the condition $\wt(U^{(1)}) + \wt(U^{(2)}) = \ell$. Then for any such pairs $U$ and $V$ we let
    \begin{equation*}
        A_{v, v', \beta}(U, V) = \begin{cases}
                      1  & U\xrightarrow{\text{$v, v', \beta$}} V\\
                      0 & \text{otherwise}
                    \end{cases}
    \end{equation*}
    where we say $U\xrightarrow{\text{$v, v', \beta$}} V$ if the following conditions hold
    \begin{enumerate}
        \item $U^{(1)}\xrightarrow{\text{$v, \beta$}} V^{(1)}$.
        \item $U^{(2)}\xrightarrow{\text{$v', -\beta$}} V^{(2)}$.
        \item $\abs{\supp(U^{(1)}) \oplus \supp(v)} = \lfloor\frac{k-t}{2}\rfloor$ and $\abs{\supp(U^{(2)}) \oplus \supp(v')} = \lceil \frac{k-t}{2} \rceil$ or vice versa.
    \end{enumerate}
    
Let $\Psi_t(x) = \frac{k^2\abs{\mcH^{(t)}}\abs{\mcU^{(t)}}}{\abs{\mcH}^2\abs{\F}} + \frac{k^2\abs{\mcU^{(t)}}}{\abs{\mcH}^2\abs{\F}} \sum_{u \in \mcU^{(t)}} \sum_{\beta \in \F^*} \sum_{v \neq v' \in \mcH^{(t)}_u} \omega_p^{\Tr(\beta b_{v})} \cdot \omega_p^{-\Tr(\beta b_{v'})} \cdot \overline{\chi_{\beta v}} \cdot \chi_{\beta v'}$ be a polynomial defined by a set $\mcH^{(t)} \subseteq \mcH$ of $k$-sparse vectors from $\F^n$ with a decomposition $\mcH^{(t)} = \{\mcH^{(t)}_u\}_{u \in \mcU^{(t)}}$ and complex coefficients $\cbra{c_{v, \beta}}_{\substack{v \in \mcH \\ \beta \in \F^*}}$. We define the level-$\ell$ Kikuchi matrix for this polynomial to be $A^{(t)} = \sum_{u \in \mcU^{(t)}} \sum_{\beta \in \F^*} \sum_{v \neq v' \in \mcH^{(t)}_u} c_{v, \beta} \cdot \overline{c_{v', \beta}} \cdot A_{v-u, v'-u, \beta}$. For our purposes, we often shorthand $A_{v-u, v'-u, \beta} = A^u_{v, v', \beta}$ and let it be understood that $v, v'$ intersect $u$. We refer to the graph (with complex weights) defined by the underlying adjacency matrix as the Kikuchi graph.
\end{definition}

\begin{remark}
    For a fixed $\mcH^{(t)}_u$, the corresponding term in $A^{(t)}$ is defined as a graph with $(k-t)$-sparse vectors labels on the edges. This is because the ambient vectors $v$ in $\mcH$ are $k$-sparse and we subtract out the $t$-sparse shared part $u \in \mcU^{(t)}$. For notational convenience, we refer to $v-u$ as just $v_u$.
\end{remark}

\begin{groupcase}
        The odd-arity Kikuchi matrix generalizes by enforcing support constraints as in \cref{def:groupkikuchimatrix}.
    \end{groupcase}

Similar to \cref{fact:bilinearforms} in \cref{sec:refutation} we make the following connection between the bilinear forms of our constructed Kikuchi matrix and the value of the polnomial that we are seeking to refute.

\begin{observation}
    \label{fact:genbilinearforms}
    For $x \in \F^n$ define $y \in \C^N$ which we index by pair $U = (U^{(1)}, U^{(2)})$ with $U^{(1)}, U^{(2)} \in \F^n$ and $\wt(U^{(1)}) + \wt(U^{(2)}) = \ell$ and set $y_U = \chi_{U^{(1)}}(x) \cdot \chi_{U^{(2)}}(x)$. Then:
    \begin{equation*}
        \Psi_t(x) = \frac{k^2\abs{\mcH^{(t)}}\abs{\mcU^{(t)}}}{\abs{\mcH}^2\abs{\F}} + \frac{k^2\abs{\mcU^{(t)}}}{\Delta\abs{\mcH}^2\abs{\F}} \cdot y^\dagger A^{(t)} y\mcom
    \end{equation*}
    where $\Delta := {k-t \choose \lceil \frac{k-t}{2} \rceil}{k-t \choose \lfloor \frac{k-t}{2} \rfloor}{2n-2(k-t) \choose \ell-k-t} \cdot \abs{\F^*}^{\ell - k-t} \cdot 2^{\1(\text{$k-t$ odd})}$.
\end{observation}

\begin{proof}
    Compute
    \begin{align*}
        y^\dagger A^{(t)} y
        &= \sum_{\substack{U, V \in \F^n\\ \wt(U^{(1)}) + \wt(U^{(2)}) = \ell\\ \wt\left(V^{(1)}\right) + \wt\left(V^{(2)}\right) = \ell}} A^{(t)}(U, V) \cdot \overline{\chi_{U^{(1)}}(x)} \cdot \overline{\chi_{U^{(2)}}(x)} \cdot \chi_{V^{(1)}}(x) \cdot \chi_{V^{(2)}}(x)\\
        &= \sum_{\substack{U, V \in \F^n\\ \wt(U^{(1)}) + \wt(U^{(2)}) = \ell\\ \wt\left(V^{(1)}\right) + \wt\left(V^{(2)}\right) = \ell}} \sum_{u \in \mcU} \sum_{v \neq v' \in \mcH_u, \beta \in \F^*} \1\left(U\xrightarrow{\text{$v_u, v'_u, \beta$}} V\right) \cdot c_{v, \beta} \cdot \overline{c_{v', \beta}} \cdot \overline{\chi_{U^{(1)}-V^{(1)}}(x)} \cdot \overline{\chi_{U^{(2)}-V^{(2)}}(x)}\\
        &= \sum_{\substack{U, V \in \F^n\\ \wt(U^{(1)}) + \wt(U^{(2)}) = \ell\\ \wt\left(V^{(1)}\right) + \wt\left(V^{(2)}\right) = \ell}} \sum_{u \in \mcU} \sum_{v \neq v' \in \mcH_u, \beta \in \F^*} \1\left(U\xrightarrow{\text{$v_u, v'_u, \beta$}} V\right) \cdot c_{v, \beta} \cdot \overline{c_{v', \beta}} \cdot \overline{\chi_{\beta v}(x)} \cdot \chi_{\beta v'}\mper
    \end{align*}

    Each term appears once for every pair $(U, V)$ with $U\xrightarrow{\text{$v_u, v'_u, \beta$}} V$. We count the number of such pairs as $\Delta := {k-t \choose \lceil \frac{k-t}{2} \rceil}{k-t \choose \lfloor \frac{k-t}{2} \rfloor}{2n-k-t \choose \ell-k-t} \cdot \abs{\F^*}^{\ell - k-t} \cdot 2^{\1(\text{$k-t$ odd})}$. We first specify the supports of $U^{(1)}$, $U^{(2)}$, $V^{(1)}$, and $V^{(2)}$ and then the non-zero elements $U^{(1)}_i$ for $i \in \supp(U^{(1)})$ (and the same for $U^{(2)}$, $V^{(1)}$, and $V^{(2)}$). The condition $\abs{\supp(U^{(1)}) \oplus \supp(v_u)} = \lfloor\frac{k-t}{2}\rfloor$ and $\abs{\supp(U^{(2)}) \oplus \supp(v'_u)} = \lceil \frac{k-t}{2} \rceil$ or the other way around gives $2$ options when $k-t$ is odd (otherwise the conditions are the same). Without loss of generality we let $\abs{\supp(U^{(1)}) \oplus \supp(v_u)} = \lfloor\frac{k-t}{2}\rfloor$, then there are ${k-t \choose \lfloor \frac{k-t}{2} \rfloor}$ choices for $\supp(U^{(1)}) \cap \supp(v_u)$. Similarly we have ${k-t \choose \lceil \frac{k-t}{2} \rceil}$ ways to pick the intersection $\supp(U^{(2)}) \cap \supp(v'_u)$. Among $\supp(U^{(1)}) \setminus v_u$ and $\supp(U^{(2)}) \setminus v'_u$, there are then $\ell - k-t$ indices, and each falls outside of $[n] \setminus \supp(v_u)$ and $[n] \setminus \supp(v'_u)$, giving ${2n-2k-t \choose \ell - k-t}$ options. The condition $U^{(1)}\xrightarrow{\text{$v, \beta$}} V^{(1)}$ enforces that $\supp(U^{(1)}) \oplus \supp(V^{(1)}) = \supp(v_u)$ which requires $\supp(V^{(1)}) \cap \supp(v_u) = \supp(v_u) \setminus (\supp(U^{(1)}) \cap \supp(v_u))$ and $\supp(V^{(1)}) \setminus \supp(v_u) = \supp(U^{(1)}) \setminus \supp(v_u)$, fully determining $\supp(V^{(1)})$. Likewise $\supp(V^{(2)})$ is fully determined. Now we specify the actually values. Note for $i \in \supp(U^{(1)}) \cap \supp(v_u)$ we have $U_i = \beta v_i$, and likewise for the other vectors. For the other $\ell - k-t$ non-zero entries across $U^{(1)}$ and $U^{(2)}$, we may choose any value in $\F^*$. The corresponding entries in $V^{(1)}$ and $V^{(2)}$ are then determined as they must cancel, totaling to $\abs{\F^*}^{\ell -k-t}$ points. This finishes the proof.
\end{proof}

We also establish a bound on the average degree of the underlying Kikuchi graph akin to \cref{fact:avgdegree}.

\begin{observation}
    \label{fact:genavgdegree}
    Let $A^{(t)}$ be a level-$\ell$ Kikuchi matrix defined from $\mcH^{(t)}$. Then the average degree $d_t \geq \abs{\F^*} \left(\frac{\ell}{2n\abs{\F^*}}\right)^{k-t} \sum_{u \in \mcU^{(t)}} {\abs{\mcH^{(t)}_u} \choose 2}$.
\end{observation}

\begin{proof}
    Note that every triple $(v, v', \beta)$ uniformly adds $\Delta = {k-t \choose \lceil \frac{k-t}{2} \rceil}{k-t \choose \lfloor \frac{k-t}{2} \rfloor}{2n-2(k-t) \choose \ell-k-t} \cdot \abs{\F^*}^{\ell - k-t} \cdot 2^{\1(\text{$k-t$ odd})}$ edges. The total degree can then be written as $\Delta \abs{\F^*} \sum_{u \in \mcU^{(t)}} {\abs{\mcH^{(t)}_u} \choose 2}$ and the average using standard binomial approximations is:
    \begin{flalign*}
        &d_t = \frac{\Delta\abs{\F^*}}{N} \sum_{u \in \mcU} {\abs{\mcH^{(t)}_u} \choose 2}
        = \frac{{k-t \choose \lceil \frac{k-t}{2} \rceil}{k-t \choose \lfloor \frac{k-t}{2} \rfloor}{2n-(k-t) \choose \ell-k-t} \cdot \abs{\F^*}^{\ell - k-t} \cdot 2^{\1(\text{$k-t$ odd})}\abs{\F^*}}{\abs{\F^*}^\ell {2n \choose \ell}} \sum_{u \in \mcU} {\abs{\mcH^{(t)}_u} \choose 2}\\
        &\geq \abs{\F^*}\left({\frac{\ell}{2n \abs{\F^*}}}\right)^{k-t}  \sum_{u \in \mcU} {\abs{\mcH^{(t)}_u} \choose 2}\mper
    \end{flalign*}
    The inequality follows from \cref{fact:binomest}.
\end{proof}

There is one other property that is important to track, which we can think of as a typed local degree.

\begin{definition}[Local degree]
    Let $U = (U^{(1)}, U^{(2)})$ be a vertex of the level-$\ell$ Kikuchi graph and $v \in \mcH^{(t)}_u$. We define the $v$-local degrees of a vertex:
    \begin{align*}
        &d_{U, v, 0} = \{v' \in \mcH^{(t)}_u \mid \exists V \text{ vertex and } \beta \in \F^* \text{ such that } U  \xrightarrow{\text{$v_u, v'_u, \beta$}} V\}\mper\\
        &d_{U, v, 1} = \{v' \in \mcH^{(t)}_u \mid \exists V \text{ vertex and } \beta \in \F^* \text{ such that } U  \xrightarrow{\text{$v'_u, v_u, \beta$}} V\}\mper
    \end{align*}
    We call a Kikuchi graph $\eta$-bounded local degree if for all $d_{U, v, b} \leq \eta$ for all $U, v, b$.
\end{definition}

With these definitions in place, we are able to specify our spectral certificate (an analog of \cref{lem:countingbacktrackingwalks}) and show how it can be used to achieve \cref{lem:mainrefutation}.

\begin{lemma}
    \label{lem:maincountbackwalks}
    Let $A$ be the level-$\ell$ Kikuchi matrix for a bipartite polynomial specified by an $\mcU$-bipartite vector collection $\mcH = \cbra{\mcH_u}_{u \in \mcU}$ of $(k-t)$-sparse vectors from $\F^n$ and complex coefficients $\{c_{v, \beta}\}_{\substack{v \in \mcH \\ \beta \in \F^*}}$. Let $\Gamma \in \C^{N \times N}$ be $\Gamma = D + d \Id$ where $D_{U, U} := \deg(U)$ and average degree $d$. Suppose additionally that the underlying Kikuchi graph is $\eta$-bounded local degree and that the $c_{v, \beta}$ are drawn from the distribution specified in \cref{item:refutation2} of \cref{thm:refutation}. Then,
    \begin{equation*}
    \norm{\Gamma^{-1/2} A \Gamma^{-1/2}}_2 \leq 8\sqrt{\frac{\eta\ell \log (\abs{\F^*} n)}{d}}\mper
    \end{equation*}
\end{lemma}

\begin{groupcase}
    The main lemma can be generalized for $\lambda$-robust $k$-$\LIN(G)$ equations as in \cref{lem:groupcountingbacktrackingwalks} at the cost of a factor $\sqrt{\frac{1}{\lambda}}$.
\end{groupcase}

\begin{proof}[Proof of \cref{lem:mainrefutation} from \cref{lem:maincountbackwalks}]

By \cref{fact:genbilinearforms} we have
\begin{equation*}
    \Psi_t(x) = \frac{k^2\abs{\mcH^{(t)}}\abs{\mcU^{(t)}}}{\abs{\mcH}^2\abs{\F}} + \frac{k^2\abs{\mcU^{(t)}}}{\Delta\abs{\mcH}^2\abs{\F}} \cdot y^\dagger A^{(t)} y\mper
\end{equation*}
Letting $\tilde{A}^{(t)} = \Gamma^{-1/2}A^{(t)} \Gamma^{-1/2}$ we write
\begin{flalign*}
    &\frac{k^2\abs{\mcU^{(t)}}}{\Delta\abs{\mcH}^2\abs{\F}} \cdot y^\dagger A^{(t)} y = \frac{k^2\abs{\mcU^{(t)}}}{\Delta\abs{\mcH}^2\abs{\F}} \cdot (\Gamma^{1/2}y)^\dagger \tilde{A}^{(t)} (\Gamma^{1/2} y) \leq \frac{k^2\abs{\mcU^{(t)}}}{\Delta\abs{\mcH}^2\abs{\F}} \cdot \norm{\tilde{A}^{(t)}}_2 \norm{\Gamma^{1/2} y}_2^2 \\
    &= \frac{k^2\abs{\mcU^{(t)}}}{\Delta\abs{\mcH}^2\abs{\F}} \cdot \norm{\tilde{A}^{(t)}}_2 \cdot \tr(\Gamma) = \frac{2k^2\abs{\mcU^{(t)}}\abs{\F^*}}{\abs{\mcH}^2\abs{\F}} \sum_{u \in \mcU^{(t)}} {\abs{\mcH^{(t)}_u} \choose 2} \cdot \norm{\tilde{A}^{(t)}}_2  \mcom
\end{flalign*}
where we use that $\norm{\Gamma^{1/2} y}_2^2 = y^{\dagger} \Gamma y = \sum_{U} \Gamma_U \abs{y_U}^2 = \sum_U \Gamma_U = \tr(\Gamma)$ since $\abs{y_U} = 1$ for all $U$, and that $\tr(\Gamma) = 2 \Delta \abs{\F^*} \sum_{u \in \mcU^{(t)}} {\abs{\mcH^{(t)}_u} \choose 2}$. We use this inequality directly to define our certificate $\algval(\Psi_t)$. For every $x \in \F^n$ we have
\begin{flalign*}
    &\Psi_t(x) = \frac{k^2\abs{\mcH^{(t)}}\abs{\mcU^{(t)}}}{\abs{\mcH}^2\abs{\F}} + \frac{k^2\abs{\mcU^{(t)}}}{\Delta\abs{\mcH}^2\abs{\F}} \cdot y^\dagger A^{(t)} y\\ 
    &\leq \frac{k^2\abs{\mcH^{(t)}}\abs{\mcU^{(t)}}}{\abs{\mcH}^2\abs{\F}} + \frac{2k^2\abs{\mcU^{(t)}}\abs{\F^*}}{\abs{\mcH}^2\abs{\F}} \sum_{u \in \mcU^{(t)}} {\abs{\mcH^{(t)}_u} \choose 2} \cdot \norm{\tilde{A}^{(t)}}_2 := \algval^*(\Psi_t)\mper
\end{flalign*}
To fulfill \cref{item:ref1}, we essentially compute and output $\algval^*(\Psi_t)$ (we will end up computing a slightly different $\algval(\Psi)$ that we define shortly). This can be accomplished with simple arithmetic operations, with the bulk of the computation coming from computing $\norm{\tilde{A}^{(t)}}_2$. The matrix is $N \times N$ where $N = (\abs{\F^*}n)^{O(\ell)}$, meaning this can be computed in time $(\abs{\F^*}n)^{O(\ell)}$ using standard linear algebraic techniques.

To achieve \cref{item:ref2}, we bound the two parts of $\algval^*(\Psi_t)$ independently by $\frac{\varepsilon^2}{2}$ as
\begin{equation*}
    \algval^*(\Psi_t) \leq \frac{k^2\abs{\mcH^{(t)}}\abs{\mcU^{(t)}}}{\abs{\mcH}^2\abs{\F}} + \frac{2k^2\abs{\mcU^{(t)}}\abs{\F^*}}{\abs{\mcH}^2\abs{\F}} \sum_{u \in \mcU^{(t)}} {\abs{\mcH^{(t)}_u} \choose 2} \cdot \norm{\tilde{A}^{(t)}}_2 \leq \frac{\varepsilon^2}{2} + \frac{\varepsilon^2}{2} \leq \varepsilon^2 \mper
\end{equation*}
By the assumption \cref{item:newref2b}, we have $\abs{\mcU^{(t)}} \leq \frac{2 \abs{\mcH}}{\tau_t}$, and since $\tau_t \geq 4k^2\varepsilon^{-2}$, this is enough to conclude $\frac{k^2\abs{\mcH^{(t)}}\abs{\mcU^{(t)}}}{\abs{\mcH}^2\abs{\F}} \leq \frac{\varepsilon^2}{4\abs{\F}} \cdot \frac{\abs{\mcH^{(t)}}}{\abs{\mcH}}$ which is a stronger condition than $\leq \frac{\varepsilon^2}{2}$.

For the latter bound, we must show $\frac{2k^2\abs{\mcU^{(t)}}\abs{\F^*}}{\abs{\mcH}^2\abs{\F}} \sum_{u \in \mcU^{(t)}} {\abs{\mcH^{(t)}_u} \choose 2} \cdot \norm{\tilde{A}^{(t)}}_2 \leq \frac{\varepsilon^2}{2}$ with high probability. We show that this is quantitatively achieved by \cref{lem:maincountbackwalks} when $\abs{\mcH}$ is sufficiently large. To begin, note that a direct application of \cref{lem:maincountbackwalks} and \cref{fact:genavgdegree} in tandem yields
\begin{align*}
    \frac{2k^2\abs{\mcU^{(t)}}\abs{\F^*}}{\abs{\mcH}^2\abs{\F}} \sum_{u \in \mcU^{(t)}} {\abs{\mcH^{(t)}_u} \choose 2} \cdot \norm{\tilde{A}^{(t)}}_2 
    &\leq \frac{16k^2\abs{\mcU^{(t)}}\abs{\F^*}}{\abs{\mcH}^2\abs{\F}} \sum_{u \in \mcU^{(t)}} {\abs{\mcH^{(t)}_u} \choose 2} \cdot \sqrt{\frac{\eta_t\ell \log (\abs{\F^*} n)}{d}}\\
    &\leq \frac{16k^2\abs{\mcU^{(t)}}\abs{\F^*}}{\abs{\mcH}^2\abs{\F}} \sum_{u \in \mcU^{(t)}} {\abs{\mcH^{(t)}_u} \choose 2} \cdot \sqrt{\frac{\eta_t\ell \log (\abs{\F^*} n)}{\abs{\F^*} \sum_{u \in \mcU^{(t)}} {\abs{\mcH^{(t)}_u} \choose 2}} \cdot \left(\frac{2n\abs{\F^*}}{\ell}\right)^{k-t}}\\
    &\leq \frac{16k^2\abs{\mcU^{(t)}}\abs{\F^*}}{\abs{\mcH}^2\abs{\F}} \cdot \sqrt{\frac{\eta_t\ell \log (\abs{\F^*} n)}{\abs{\F^*}} \cdot \left(\frac{2n\abs{\F^*}}{\ell}\right)^{k-t} \sum_{u \in \mcU^{(t)}} {\abs{\mcH^{(t)}_u} \choose 2}}\\
    &\leq \frac{16k^2\abs{\mcU^{(t)}}\abs{\F^*}}{\abs{\mcH}^2\abs{\F}} \cdot \sqrt{\frac{\eta_t\ell \log (\abs{\F^*} n)}{\abs{\F^*}} \cdot \left(\frac{2n\abs{\F^*}}{\ell}\right)^{k-t} \abs{\mcU^{(t)}} \tau_t^2}\mper
\end{align*}
In the last line we use that $\abs{\mcH^{(t)}_u} \leq \tau_t$ for all $u \in \mcU^{(t)}$. We now use the property $\abs{\mcU^{(t)}} \leq \frac{2\abs{\mcH}}{\tau_t}$ to claim
\begin{equation*}
    \frac{2k^2\abs{\mcU^{(t)}}\abs{\F^*}}{\abs{\mcH}^2\abs{\F}} \sum_{u \in \mcU^{(t)}} {\abs{\mcH^{(t)}_u} \choose 2} \cdot \norm{\tilde{A}^{(t)}}_2 \leq \frac{64k^2\abs{\F^*}}{\abs{\F}} \cdot \sqrt{\frac{\eta_t\ell \log (\abs{\F^*} n)}{\abs{\F^*}\abs{\mcH}\tau_t} \cdot \left(\frac{2n\abs{\F^*}}{\ell}\right)^{k-t}}\mper
\end{equation*}
Finally we use that $\tau_t = \left(\frac{n\abs{\F^*}}{\ell}\right)^{k/2-t} \cdot 4k^2\varepsilon^{-2}$ we observe that when $\abs{\mcH} \geq C^k \cdot n \log(\abs{\F^*}n) \left(\frac{n\abs{\F^*}}{\ell}\right)^{k/2-1} \cdot \varepsilon^{-4}$ for sufficiently large constant $C > 0$, we obtain a bound of $\frac{\varepsilon^2}{2}\sqrt{\frac{\eta_t}{\tau_t} \left(\frac{n\abs{\F^*}}{\ell}\right)^{k/2-t}}$. We may trivially bound $\eta_t \leq \abs{\mcH^{(t)}_u} \leq \tau_t$ and when $t \geq \frac{k}{2}$ we have $\left(\frac{n\abs{\F^*}}{\ell}\right)^{k/2-t} \leq 1$, finishing our bound.

When $t < \frac{k}{2}$ the trivial bound is not sufficient. In particular if $\eta_t$ is sufficiently close to its max $\tau_t$ then this value diverges as $n \to \infty$. Instead we need the following lemma which gives us a way to reduce the local degree of an arbitrary Kikuchi matrix while maintaining high global degree.

    \begin{lemma}
        \label{lem:edgedeletion}
        Let $K_\ell$ be a level-$\ell$ Kikuchi graph with average degree $d(K_\ell)$, for each $v \neq v' \in \mcH^{(t)}$ $(k-t)$-sparse and $\beta \in \F^*$ there are $\Delta$ edges of type $(v, v', \beta)$, and satisfying the output criteria in \cref{lem:decompositionalg}. Then we can find a subgraph $\hat{K}_\ell$ in time $(\abs{\F^*}n)^{O(\ell)}$ with the following properties:
        \begin{itemize}
            \item $\hat{K}_\ell$ is $D^k \varepsilon^{-2}$-bounded degree for some constant $D > 0$.
            \item The number of $(v, v', \beta)$ type edges is $\geq (1-\rho) \Delta$ for some $\rho \in [0, \frac{1}{2}]$.
        \end{itemize}
    \end{lemma}

    We think of $\rho$ as the fraction of edges missing, and the second condition guarantees an important uniformity in how many edges of each type. In particular, an inspection of the above section shows \cref{fact:genbilinearforms} and \cref{fact:genavgdegree} hold with $\Delta$ replaced with $(1-\rho)\Delta$. This allows the proof of \cref{lem:maincountbackwalks} to go through scaling $d_t$ by $1-\rho$ but with $\eta_t \leq \exp(k) \cdot \varepsilon^{-2}$. As a result, in the case $t < \frac{k}{2}$ we have for slightly larger constant $CD > 0$ in $\abs{\mcH} \geq (CD)^k \cdot n \log(\abs{\F^*}n) \left(\frac{n\abs{\F^*}}{\ell}\right)^{k/2-1} \cdot \varepsilon^{-4}$ that
    \begin{equation*}
        \frac{2k^2\abs{\mcU^{(t)}}\abs{\F^*}}{\abs{\mcH}^2\abs{\F}} \sum_{u \in \mcU^{(t)}} {\abs{\mcH^{(t)}_u} \choose 2} \cdot \norm{\tilde{A}^{(t)}}_2\leq \frac{\varepsilon^2}{2}\sqrt{\frac{\eta_t}{D^k\tau_t} \left(\frac{n\abs{\F^*}}{\ell}\right)^{k/2-t}} \leq \frac{\varepsilon^2}{2}\sqrt{\frac{\eta_t}{D^k \varepsilon^{-2}}} \leq \frac{\varepsilon^2}{2}\mper
    \end{equation*}

    \begin{tcolorbox}[
    width=\textwidth,   
    colframe=black,  
    colback=white,   
    title=Bipartite Polynomial Refutation Algorithm,
    colbacktitle=white, 
    coltitle=black,      
    fonttitle=\bfseries,
    center title,   
    enhanced,       
    frame hidden,           
    borderline={1pt}{0pt}{black},
    sharp corners,
    toptitle=2.5mm
]
\textbf{Input:} A $\mcU^{(t)}$-bipartite vector collection $\mcH^{(t)} \subseteq \mcH$ of $k$-sparse vectors from $\F^n$ and complex coefficients $\{c_{v, \beta}\}_{\substack{v \in \mcH \\ \beta \in \F^*}}$ specifying a polynomial $\Psi_t = \frac{k^2\abs{\mcU^{(t)}}}{\abs{\mcH}^2\abs{\F}} \sum_{u \in \mcU^{(t)}} \sum_{\beta \in \F^*} \sum_{v, v' \in \mcH_u^{(t)}} \omega_p^{\Tr(\beta b_{v})} \cdot \omega_p^{-\Tr(\beta b_{v'})} \cdot \overline{\chi_{\beta v}} \cdot \chi_{\beta v'}$ along with the integer $\abs{\mcH}$.\\

\textbf{Output:} $\algval(\Psi_t) \in [0,1]$ with guarantee $\algval(\Psi_t) \geq \max_{x \in \F^n} \abs{\Psi_t(x)}$.\\

\textbf{Algorithm:}
\begin{enumerate}
    \item Construct the $N \times N$ Kikuchi matrix (\cref{def:oddkikuchimatrix}) $A^{(t)}$ from the input.
    \item Utilizing the algorithm from \cref{lem:edgedeletion}, construct an adjacency submatrix $\hat{A}^{(t)}$.
    \item Compute $\Gamma$ for $\hat{A}^{(t)}$ and compute $\norm{\Gamma^{-1/2}\hat{A}^{(t)} \Gamma^{-1/2}}_2$.
    \item Output $\frac{k^2\abs{\mcH^{(t)}}\abs{\mcU^{(t)}}}{\abs{\mcH}^2\abs{\F}} + \frac{2k^2\abs{\mcU^{(t)}}\abs{\F^*}}{\abs{\mcH}^2\abs{\F}} \sum_{u \in \mcU^{(t)}} {\abs{\mcH^{(t)}_u} \choose 2} \cdot \norm{\Gamma^{-1/2}\hat{A}^{(t)} \Gamma^{-1/2}}_2$.
\end{enumerate}

\end{tcolorbox}

\end{proof}

We finish this section by proving \cref{lem:maincountbackwalks} and \cref{lem:edgedeletion}, which proceeds via a more complicated version of the trace moment method analysis in the proof of \cref{lem:countingbacktrackingwalks}.

\begin{proof}[Proof of \cref{lem:maincountbackwalks}]

By assumption we have $b_v \sim \F$ independently and may treat our Kikuchi matrix $A = \sum_{u \in \mcU} \sum_{\beta \in \F^*} \sum_{v \neq v' \in \mcH_u} c_{v, \beta} \cdot \overline{c_{v', \beta}} \cdot A^u_{v, v', \beta}$ as a random matrix with respect to $c_{v,\beta}$ and $c_{v', \beta}$. As in \cref{lem:countingbacktrackingwalks}, we use the trace power method $\norm{\Gamma^{-1/2} A \Gamma^{-1/2}}_2 \leq \tr((\Gamma^{-1}A)^{2t})^{1/2t}$ for $t = \log N$. We compute
\begin{align*}
    \E\left[\tr\left(\left(\Gamma^{-1} A\right)^{2t}\right)\right] 
    &= \E\left[\tr\left(\left(\Gamma^{-1} \sum_{u \in \mcU} \sum_{\beta \in \F^*} \sum_{v \neq v' \in \mcH_u} c_{v, \beta} \cdot \overline{c_{v', \beta}} \cdot A^u_{v, v', \beta}\right)^{2t}\right) \right]\\
    &= \E\left[\tr\left(\sum_{\substack{(v_1, v_1', \beta_1) ,..., (v_{2t}, v_{2t}', \beta_{2t}) \\ \in \cbra{\mcH_u \times \mcH_u}_{u \in \mcU} \times \F^*}} \prod_{i = 1}^{2t} \Gamma^{-1} \cdot c_{v_i, \beta_i} \cdot \overline{c_{v_i', \beta_i}} \cdot A^{u_i}_{v_i, v_i', \beta_i} \right) \right]\\
    &= \sum_{\substack{(v_1, v_1', \beta_1) ,..., (v_{2t}, v_{2t}', \beta_{2t}) \\ \in \cbra{\mcH_u \times \mcH_u}_{u \in \mcU} \times \F^*}} \E\left[\tr\left(\prod_{i = 1}^{2t} \Gamma^{-1} \cdot c_{v_i, \beta_i} \cdot \overline{c_{v_i', \beta_i}} \cdot A^{u_i}_{v_i, v_i', \beta_i} \right) \right]\\
    &= \sum_{\substack{(v_1, v_1', \beta_1) ,..., (v_{2t}, v_{2t}', \beta_{2t}) \\ \in \cbra{\mcH_u \times \mcH_u}_{u \in \mcU} \times \F^*}} \E\left[\prod_{i = 1}^{2t}c_{v_i, \beta_i} \cdot \overline{c_{v_i', \beta_i}}\right] \cdot \tr\left(\prod_{i = 1}^{2t} \Gamma^{-1} A^{u_i}_{v_i, v_i', \beta_i} \right)\mper
\end{align*}

We now define a notion of trivially closed sequence analogous to that in the proof of \cref{lem:countingbacktrackingwalks}. Let $(v_1, v_1', \beta_1) ,..., (v_{2t}, v'_{2t}, \beta_{2t})$ be a term in the above sum. Fix $v \in \mcH$, and let $R^+(v)$ denote the set of $i \in [2t]$ such that either $v_i = v$ and $R^-(v)$ those that $v_i' = v$. We observe that if for some $v \in \mcH$, $\sum_{i \in R^+(v)} \beta_i - \sum_{i \in R^-(v)} \beta_i \ne 0$, then $\E\left[\prod_{i=1}^{2t} c_{v_i, \beta_i} \right] = 0$. Indeed, this is because $b_v$ is independent for each $v \in \mcH$, and so $\E\left[\prod_{i=1}^{2t} c_{v_i, \beta_i} \right]  = \prod_{v \in \mcH} \E\left[\prod_{i \in R^+(v)} c_{v, \beta_i} \prod_{i \in R^-(v)} \overline{c_{v, \beta_i}}\right]$, and 
    \begin{equation*}
    \E\left[\prod_{i \in R^+(v)} c_{v, \beta_i} \prod_{i \in R^-(v)} \overline{c_{v, \beta_i}}\right] = \E\left[\prod_{i \in R^+(v)} \omega_p^{\Tr(\beta_i b_v)} \prod_{i \in R^-(v)} \omega_p^{-\Tr(\beta_i b_v)}\right] = \E\left[\omega_p^{\Tr((\sum_{i \in R^+(v)}\beta_i -\sum_{i \in R^-(v)} \beta_i) b_v)}\right] \mper
    \end{equation*}
    Then, since $b_v$ is uniform from $\F$, it follows that $\E\left[\omega_p^{\Tr((\sum_{i \in R^+(v)}\beta_i -\sum_{i \in R^-(v)} \beta_i) b_v)}\right] = 0$ if $\sum_{i \in R^+(v)}\beta_i -\sum_{i \in R^-(v)} \beta_i \ne 0$, and $\E\left[\omega_p^{\Tr((\sum_{i \in R^+(v)}\beta_i -\sum_{i \in R^-(v)} \beta_i) b_v)}\right] = 1$ if $\sum_{i \in R^+(v)}\beta_i -\sum_{i \in R^-(v)} \beta_i = 0$. This motivates the following definition.
\begin{definition}[Trivially closed sequence]
    \label{def:gentriviallyclosedwalks}
    Let $(v_1, v_1', \beta_1) ,..., (v_{2t}, v_{2t}',  \beta_{2t}) \in \mcH \times \mcH \times \F^*$. We say that $(v_1, v_1', \beta_1) ,..., (v_{2t}, v'_{2t}, \beta_{2t}) \in \mcH \times \mcH \times \F^*$ is trivially closed with respect to $v$ if it holds that $\sum_{i \in R^+(v)} \beta_i - \sum_{i \in R^-(v)} \beta_i = 0$. We say that the sequence is trivially closed if it is trivially closed with respect to all $v \in \mcH$.
\end{definition}
    With the above definition in hand, we have shown that
\begin{flalign*}
  &  \E\left[\tr\left(\left(\Gamma^{-1} A\right)^{2t}\right)\right] 
    = \sum_{\substack{(v_1, v_1', \beta_1) ,..., (v_{2t}, v_{2t}', \beta_{2t}) \\ \in \cbra{\mcH_u \times \mcH_u}_{u \in \mcU} \times \F^* \\ \text{trivially closed}}} \tr\left(\prod_{i = 1}^{2t} \Gamma^{-1} A^{u_i}_{v_i, v_i', \beta_i} \right)\mper
\end{flalign*}

\begin{remark}
    In the general group case, the above $0$-$1$ characterization of the walks extends since the instances we are refuting have $b_v$ drawn randomly from the corresponding representative group, which still corresponds to a symmetric set of roots of unity.
\end{remark}

The following lemma yields the desired bound.
\begin{lemma}
    \label{lem:genmaincountingbacktrackingwalks}
    $\sum_{\substack{(v_1, v_1', \beta_1) ,..., (v_{2t}, v_{2t}', \beta_{2t}) \\ \in \cbra{\mcH_u \times \mcH_u}_{u \in \mcU} \times \F^* \\ \text{trivially closed}}} \tr\left(\prod_{i = 1}^{2t} \Gamma^{-1} A^{u_i}_{v_i, v_i', \beta_i} \right) \leq N \cdot 2^{2t} \cdot \left(\frac{2t}{d}\right)^t$.
\end{lemma}
Taking $t$ to be $c \log_2 N$ for a sufficiently large constant $c$ and applying Markov's inequality finishes the proof.
\end{proof}

\begin{proof}[Proof of \cref{lem:genmaincountingbacktrackingwalks}]
    We bound the sum as follows. First, we observe that for a trivially closed sequence $(v_1, v'_1, \beta_1) ,..., (v_{2t}, v'_{2t}, \beta_{2t})$, we have
    \begin{flalign*}
    \sum_{\substack{(v_1, v_1', \beta_1) ,..., (v_{2t}, v_{2t}', \beta_{2t}) \\ \in \cbra{\mcH_u \times \mcH_u}_{u \in \mcU} \times \F^* \\ \text{trivially closed}}} \tr\left(\prod_{i = 1}^{2t} \Gamma^{-1} A^{u_i}_{v_i, v_i', \beta_i} \right) = \sum_{U_0, U_1, \dots, U_{2t-1}} \prod_{i = 0}^{2t - 1} \Gamma^{-1}_{U_i} \cdot \1\left(U_i\xrightarrow{\text{$u, v_{i+1}, v'_{i+1}, \beta_{i+1}$}} U_{i+1}\right) \mper
    \end{flalign*}
    with $U_0 = U_{2t}$. Thus, the sum that we wish to bound in \cref{lem:genmaincountingbacktrackingwalks} simply counts the total weight of ``trivially closed walks'' $U_0, v_1, v'_1, \beta_1, U_1, \dots, U_{2t-1}, v_{2t}, v'_{2t}, \beta_{2t}, U_{2t}$in the Kikuchi graph $A$, where the weight of a walk is simply $\prod_{i = 0}^{2t-1} \Gamma^{-1}_{U_i}$.
    
    Let us now bound this total weight by uniquely encoding a walk $U_0, v_1, v'_1, \beta_1, U_1, \dots, v_{2t}, v'_{2t}, \beta_{2t}, U_{2t}$ as follows.
    \begin{itemize}
        \item First, we write down the start vertex $U_0$.
        \item For $i = 1, \dots, 2t$, we let $z_i$ be $1$ if (1) $v_i = v_j$, (2) $v_i = v_j'$, (3) $v'_i = v_j$, or (4) $v'_i = v'_j$ for some $j < i$. In this case, we say that the edge is ``old''. Otherwise $z_i = 0$, and we say that the edge is ``new''.
	\item For $i = 1, \dots, 2t$, if $z_i$ is $1$ then we encode $U_{i}$ by writing down the smallest $j \in [2t]$, the least $c \in [4]$ specifying which of the 4 cases above holds, and lastly the ``other'' parallel edge to take. For instance if we have $v_i = v_j$ we specify must $v'_i$ from $\mcH$. Note these parameters uniquely define the edge as there is only one $\beta \in \F^*$ and $V$ such that $U \xrightarrow{v_i, v'_i, \beta} V$.
	\item For $i = 1, \dots, 2t$, if $z_i$ is $0$ then we encode $U_i$ by writing down an integer in $1, \dots, \deg(U_{i-1})$ that specifies the edge we take to move to $U_i$ from $U_{i-1}$ (we associate $[\deg(U_{i-1})]$ to the edges adjacent to $U_{i-1}$ with an arbitrary fixed map).
    \end{itemize}
    With the above encoding, we can now bound the total weight of all trivially closed walks as follows. First, let us consider the total weight of walks for some fixed choice of $z_1, \dots, z_{2t}$. We have $N$ choices for the start vertex $U_0$. For each $i = 1, \dots, 2t$ where $z_i = 0$, we have $\deg(U_{i-1})$ choices for $U_i$, and we multiply by a weight of $\Gamma^{-1}_{U_{i-1}} \leq \frac{1}{\deg(U_{i-1})}$. For each $i = 1, \dots, 2t$ where $z_i = 1$, we have at most $2t$ choices for the index $j < i$ and $4$ for the bit $c$. For the choice of partner, without loss of generality we say $v'_i$, we note that $U_i, v_i$, and $v_i$'s position are fixed. The number of viable $v'_i$ is then just the local degree, and by the assumption of $\eta$-bounded local degree we can limit this to $\eta$ choices. Finally, we multiply by a weight of $\Gamma^{-1}_{U_{i-1}} \leq \frac{1}{d}$. Hence, the total weight for a specific $z_1, \dots, z_{2t}$ is at most $N \left(\frac{8\eta t}{d}\right)^{r}$, where $r$ is the number of $z_i$ such that $z_i = 1$.
    
    Finally, we observe that any trivially closed walk must have $r \geq t$. Hence, after summing over all $z_1, \dots, z_{2t}$, we have the final bound of $N 2^{2t} \left(\frac{8\eta t}{d}\right)^{t}$, which finishes the proof.
\end{proof}

\subsection{Edge deletion algorithm}

Consider the following greedy algorithm.

   \begin{tcolorbox}[
    width=\textwidth,   
    colframe=black,  
    colback=white,   
    title=Edge Deletion Algorithm,
    colbacktitle=white, 
    coltitle=black,      
    fonttitle=\bfseries,
    center title,   
    enhanced,       
    frame hidden,           
    borderline={1pt}{0pt}{black},
    sharp corners,
    toptitle=2.5mm
]
\textbf{Input:} A Kikuchi graph $K_\ell$ specified by $\mcH$.\\

\textbf{Output:} A subgraph $\hat{K}_\ell$ of $K_\ell$ of $\eta$-bounded local degree and $d(\hat{K}_\ell) \geq \frac{1}{2}d(K_\ell)$.\\

\textbf{Algorithm:}
\begin{enumerate}
    \item While there is a vertex $U$ and $v \in \mcH$ s.t. $d_{U,v,0}$ or $d_{U, v, 1} > \eta$, delete an arbitrary edge from $U$ that involves $v$.
    \item Let $\rho$ be the largest fraction of edges deleted for any $v \in \mcH$ and $\beta \in \F^*$. Delete edges arbitrarily until a $\rho$ fraction has been deleted for all $v \in \mcH$ and $\beta \in \F^*$.
    \item Output the resultant graph.
\end{enumerate}

\end{tcolorbox}

The algorithm straightforwardly deletes edges from vertices with high local degree until none remain. Thereafter, it deletes $(v, v', \beta)$-labelled edges until there are an equal number of each type. Therefore, we may define the fraction of edges deleted to be exactly the maximum fraction of edges deleted of any fixed type $(v, v', \beta)$ after Step 1. We show the following guarantee on such a quantity.

\begin{lemma}
    \label{lem:edgedeletionanalysis}
    Let $K_\ell$ be a Kikuchi graph for $\mcH^{(t)} = \cbra{\mcH^{(t)}_u}_{u \in \mcU^{(t)}}$ with $1 \leq t \leq k$ and $\mcH^{(t)}_u$ $(k-t)$-sparse vectors and $\mcU^{(t)}$ $t$-sparse vectors from $\F^n$ and satisfying the output criteria in \cref{lem:decompositionalg}. Fix $(v, v' , \beta)$. Let $E_{(v, v' , \beta)}$ be the set of edges in $K_\ell$ associated with this triple. In the second step of the algorithm above, we have the following guarantee for some constant $D > 0$
    \begin{equation*}
        \Pr_{(U, V) \sim E_{(v, v' , \beta)}}[(U, V) \text{ deleted}] \leq \frac{D^k}{\eta} \sum_{s = 0}^{k-t} \tau_{s+t} \min\left(1, \left(\frac{\ell}{\abs{\F^*}n}\right)^{\lfloor \frac{k-t}{2}\rfloor - s}\right)\mper
    \end{equation*}
\end{lemma}

\begin{remark}
    Note that we denote the fraction of deleted of a certain type by $\Pr_{(U, V) \sim E_{(v, v' , \beta)}}[(U, V) \text{ deleted}]$ as we think of the process as sampling an edge and checking whether it is deleted. The randomness is over this hypothetical sampling and not the algorithm itself, which is fully deterministic.
\end{remark}

\begin{proof}[Proof of \cref{lem:edgedeletion} from \cref{lem:edgedeletionanalysis}]
The lemma shows that the algorithm can output an $\eta$-bounded local degree subgraph of $K_\ell$ with $\rho \leq \frac{D^k}{\eta} \sum_{s = 0}^{k-t} \tau_{s+t} \min\left(1, (\frac{\ell}{\abs{\F^*}n})^{\lfloor \frac{k-t}{2}\rfloor - s}\right)$. We show setting $\eta = 8D^k k^3 \varepsilon^{-2}$ achieves $\rho \leq \frac{1}{2}$. In the case $s + t \geq \frac{k+1}{2}$, we can note that $\tau_{s+t} = 4k^2 \varepsilon^{-2}$ and the bound is immediate via
\begin{equation*}
    \rho \leq \frac{D^k}{\eta} \sum_{s = 0}^{k-t} \tau_{s+t} \min\left(1, \left(\frac{\ell}{\abs{\F^*}n}\right)^{\lfloor \frac{k-t}{2}\rfloor - s}\right) \leq \frac{1}{8k^3\varepsilon^{-2}} \cdot  (k-t)\tau_{s+t} \leq \frac{1}{8k^2\varepsilon^{-2}} \cdot  4k^2\varepsilon^{-2} \leq \frac{1}{2}\mper
\end{equation*}

In the case $s + t \leq \frac{k-1}{2}$ we have
\begin{flalign*}
    \rho &\leq \frac{D^k}{\eta} \sum_{s = 0}^{k-t} \tau_{s+t} \min\left(1, \left(\frac{\ell}{\abs{\F^*}n}\right)^{\lfloor \frac{k-t}{2}\rfloor - s}\right) \leq \frac{D^k}{\eta} \sum_{s = 0}^{k-t} 4k^2\varepsilon^{-2}\left(\frac{n \abs{\F^*}}{\ell}\right)^{k/2-s-t} \left(\frac{\ell}{\abs{\F^*}n}\right)^{\lfloor \frac{k-t}{2}\rfloor - s}\\
    &\leq \frac{D^k}{\eta} \sum_{s = 0}^{k-t} 4k^2\varepsilon^{-2}\left(\frac{\ell}{\abs{\F^*}n}\right)^{\frac{t}{2} - \frac{\1(k-t\text{ is odd})}{2}}\leq \frac{1}{2}\left(\frac{\ell}{\abs{\F^*}n}\right)^{\frac{t}{2} - \frac{\1(k-t\text{ is odd})}{2}}\mper
\end{flalign*}

Since $t \geq 1$, we get $\frac{1}{2}$ as our desired bound.
\end{proof}

\begin{proof}[Proof of \cref{lem:edgedeletionanalysis}]
Fix $U$ a vertex in $K_\ell$ and $v \in \mcH^{(t)}_u$ for $u \in \mcU^{(t)}$ and $b \in \{0,1\}$. We will be crude and assume \textit{all} edges with $d_{U, v, b} > \eta$ for any $v \in \mcH$ are deleted. We then think of this process as removing edges involving $v$ from all the vertices that have high $v$-local degree (but only these vertices). Observe
\begin{align*}
    \Pr_{(U, V) \sim E_{(v, v' , \beta)}}[(U, V) \text{ deleted}]
    &\leq \Pr_{(U, V) \sim E_{(v, v' , \beta)}}[d_{U, v, 0} > \eta \cup d_{U, v', 1} > \eta]\\
    &\leq 2 \Pr_{(U, V) \sim E_{(v, v' , \beta)}}[d_{U, v, 0} > \eta]\\
    &\leq \frac{2}{\eta}\E_{(U,V) \sim E_{(v, v' , \beta)}}[d_{U, v, 0}-1] \tag{Markov's Ineq.}\\
    &\leq \frac{2}{\eta}\sum_{v'' \neq v, v' \in \mcH^{(t)}_u} \Pr_{(U,V) \sim E_{(v, v' , \beta)}}[\exists V' \text{ s.t. } U  \xrightarrow{\text{$v, v'', \beta$}} V']\mper
\end{align*}

We can think of this process like this: after fixing $(v, v', \beta)$ we sample an edge $(U, V)$. We are now interested in the probability that the edge we sampled is incident to an edge $(U, V')$ involving $v$. The main fact we establish is a bound on the probability such an edge exists given uniform sampling.
\begin{lemma}
    \label{fact:crossterms}
    Fix $v, v', v'' \in \mcH_u^{(t)}$.
    \begin{equation*}
    \Pr_{(U,V) \sim E_{(v, v' , \beta)}}[\exists V' \text{ s.t. } U  \xrightarrow{\text{$v, v'', \beta$}} V'] \leq {k-t \choose \lfloor \frac{k-t}{2} \rfloor}\min\left(1, \left(\frac{\ell}{\abs{\F^*}n}\right)^{\lfloor \frac{k-t}{2}\rfloor - \abs{v' \perp v''}+t}\right)\mcom
    \end{equation*}
    where we define $v \perp v' \subseteq [n]$ as $\supp(v) \cap \supp(v') \setminus \supp(v + v')$.
\end{lemma}

Using \cref{fact:crossterms} we continue above by writing
\begin{align*}
    \Pr_{(U, V) \sim E_{(v, v' , \beta)}}[(U, V) \text{ deleted}] 
    &\leq \frac{2}{\eta}\sum_{v'' \neq v, v' \in \mcH^{(t)}_u} {k-t \choose \lfloor \frac{k-t}{2} \rfloor}\min\left(1, \left(\frac{\ell}{\abs{\F^*}n}\right)^{\lfloor \frac{k-t}{2}\rfloor - \abs{v' \perp v''}+t}\right)\\
    &\leq \frac{2}{\eta} \cdot {k-t \choose \lfloor \frac{k-t}{2} \rfloor} \sum_{s = 0}^{k-t} \sum_{\substack{v'' \neq v, v' \in \mcH^{(t)}_u \\ I = v' \perp v'' \\ \abs{I} = s+t}} \min\left(1, \left(\frac{\ell}{\abs{\F^*}n}\right)^{\lfloor \frac{k-t}{2}\rfloor - s}\right)\mper
\end{align*}

In the second line, we are essentially partitioning all vectors $v''$ based on how they intersect with $v'$. Note that since $v, v' \in \mcH^{(t)}_u$, it is guaranteed $\abs{v \perp v'} \geq t$, since they intersect on $u$. We then make the observation that for a fixed $I$ with $\abs{I} = s+t$ for $s > 0$, the maximum number of $v'' \in \mcH^{(t)}_u$ with $I = v' \perp v''$ is $\tau_{s+t}$. This follows straightforwardly from the $(\frac{\varepsilon}{2k}, \ell)$-regularity condition guaranteed in \cref{item:decompositionalg5} of \cref{lem:decompositionalg}. When $\abs{I} = t$, this bound follows more immediately from the fact that $\abs{\mcH^{(t)}_u} \leq \tau_t$. This allows us to bound the number of terms in the inner sum above nicely as
\begin{align*}
    \Pr_{(U, V) \sim E_{(v, v' , \beta)}}[(U, V) \text{ deleted}]
    &\leq \frac{2}{\eta} \cdot {k-t \choose \lfloor \frac{k-t}{2} \rfloor} \sum_{s = 0}^{k-t} \sum_{\abs{I} = s+t, I \subseteq \supp(v')} \tau_{s + t} \min\left(1, \left(\frac{\ell}{\abs{\F^*}n}\right)^{\lfloor \frac{k-t}{2}\rfloor - s}\right)\\
    &\leq \frac{2 {k-t \choose \lfloor \frac{k-t}{2} \rfloor} {k \choose k/2}}{\eta} \sum_{s = 0}^{k-t} \tau_{s+t} \min\left(1, \left(\frac{\ell}{\abs{\F^*}n}\right)^{\lfloor \frac{k-t}{2}\rfloor - s}\right)\mper
\end{align*}

Standard binomial estimates give that for some constant $D > 0$ the leading coefficient may be bound by $\frac{D^k}{\eta}$. We finish by proving \cref{fact:crossterms}.
 \end{proof}

\begin{proof}[Proof of \cref{fact:crossterms}]

For arbitrary $(v, v', \beta)$ we let $U$ be any vertex such that $U \xrightarrow{v, v', \beta} V$. We are interested in the probability that there exists $V'$ with $U \xrightarrow{v, v'', \beta} V'$ for a fixed $v''$.

Taking a look at the structure of $(U, V)$ in $E_{v, v', \beta}$ the conditions of $U \xrightarrow{v, v', \beta} V$ guarantee we have $U^{(1)}$ matches $v_u$ on either $\lfloor \frac{k-t}{2} \rfloor$ or $\lceil \frac{k-t}{2} \rceil$ indices and $U^{(2)}$ matches $v'_u$ similarly. Additionally, it guarantees that $U^{(1)}$ and $U^{(2)}$ are $0$ on all remaining indices intersecting $v_u$ and $v'_u$ respectively. In total, this fixes $2(k-t)$ indices in $U$. Since $\wt(U^{(1)})+\wt(U^{(2)}) = \ell$, there remains $\ell-k-t$ ``free'' non-zero indices outside of these. Note that these may be ``placed'' anywhere that is not already fixed, the only condition is that they match in $U$ and $V$ and therefore cancel. We can then view the process of sampling an edge, conditioned on these $2(k-t)$ fixed indices, as simply sampling the free indices from ${[2n-2k-2t] \choose \ell-k-t}$ and random values from $\F^*$ for each one.

In order for $U \xrightarrow{v, v'', \beta} V'$ to occur, we necessarily have that $U^{(2)}$ matches $v''_u$ on at least $\lfloor \frac{k-t}{2} \rfloor$ indices. Any indices in $v'_u \perp v''_u$ we get ``for free'' since $U^{(2)}$ must already match (in the worst case). The $\lfloor \frac{k-t}{2} \rfloor - \abs{v'_u \perp v''_u}$ remaining must be randomly chosen as per our sampling. Note the number of such indices is $\lfloor \frac{k-t}{2} \rfloor - \abs{v' \perp v''} + t$ since $\abs{v' \perp v''} = \abs{v'_u \perp v''_u} +t$. Fix any choice of indices $T$, of which there are at most ${k-t \choose \lfloor \frac{k-t}{2} \rfloor}$ options. We union bound over all possible $T$. The desired condition is met given (1) $T$ is contained in the set of free indices and (2) all values match those in $v''$. The probability of the latter is clearly $\frac{1}{\abs{\F^*}}^{\lfloor \frac{k-t}{2} \rfloor - \abs{v' \perp v''} + t}$. The former can be seen as ${2n-2k-2t-\gamma \choose \ell-k-t-\gamma}/{2n-2k-2t \choose \ell-k-t}$ letting $\gamma = \lfloor \frac{k-t}{2} \rfloor - \abs{v' \perp v''} + t$ and we can bound this as $\left(\frac{\ell}{n}\right)^{\lfloor \frac{k-t}{2} \rfloor - \abs{v' \perp v''} + t}$ given our binomial estimates \cref{fact:binomest}.
\end{proof}

\section{Short Linear Dependencies of $k$-Sparse Vectors: Proof of \cref{thm:feige}}
\label{sec:feige}

In the Boolean case $\F = \F_2$, \cite{GuruswamiKM22} observed that a cycle in the Kikuchi graph can be seen as an even cover of the corresponding hypergraph, which we can define as the following. Given the following observation, we see that a cycle in our generalized Kikuchi graph of \cref{def:kikuchimatrix} corresponds to a linear dependence in the underlying vector set $\mcH$.

\begin{observation}
    \label{obs:kikuchicycles}
    Consider a cycle in a Kikuchi graph over $\F$, $U_1 \xrightarrow{\text{$v_1, \beta_1$}} U_2 \xrightarrow{\text{$v_2, \beta_2$}} ... \xrightarrow{\text{$v_\ell, \beta_\ell$}} U_1$ where for $i \in [\ell]$, $U_i$ is an $\ell$-sparse vector in $\F^n$, $v_i$ is $k$-sparse from a set $\mcH$, and $\beta \in \F^*$. Since $U \xrightarrow{\text{$v, \beta$}} V$ requires $U - V = \beta v$, we observe the following telescoping sum
    \begin{equation*}
        \sum_{i = 1}^\ell \beta_i v_i = \sum_{i=1}^\ell U_i - U_{(i+1) \mod \ell} = 0\mper
    \end{equation*}
\end{observation}

The linear dependence in \cref{obs:kikuchicycles} tells us that the labels of the edges in any Kikuchi graph cycle are either (1) trivially closed as in \cref{def:triviallyclosedwalks} or (2) a linearly dependent subset. To see this, recall that a vector $v \in \mcH$ is trivially closed in the walk $U_1, v_1, \beta_1, ..., v_{\ell-1}, \beta_{\ell-1}, U_1$ if its set of coefficients sums to 0. Then note that if $R(v)$ is the set of indices with $v_i = v$ in a fixed walk, $v$ being trivially closed implies the contribution $\sum_{i \in R(v)} \beta_i v_i = 0$. Any vector being not trivially closed implies that the sum $\sum_{i=1}^\ell \beta_i v_i$ has nontrivial terms that cancel, forming a linear dependence of size $\leq \ell$ in the ambient set $\mcH$.

If a collection of vectors $\mcH$ does not have any small linearly dependent sets, it follows that all short walks on the corresponding Kikuchi graph must be trivially closed. This means that a spectral certificate on the Kikuchi matrix behaves mirroring \cref{lem:countingbacktrackingwalks}, where we only had to count trivially closed walks. This simplification is strong enough to provide a spectral double counting argument proving an upper bound on the maximum number of vectors in such a set, which we give here.

\subsection{Generalized Feige's conjecture (even-arity case)}

\begin{theorem}[Generalized Feige's Conjecture, even-arity]
     Fix $k$ even and $\ell \geq k/2$. Let $\mcH$ be a set of $\abs{\mcH} \geq O(n) \cdot \log\left(\abs{\F^*}n\right) \left(\frac{n\abs{\F^*}}{\ell}\right)^{k/2-1}$ $k$-sparse vectors in $\F^n$. There exists a set $\mcV \subseteq \mathcal{H}$ with $\abs{\mcV} \leq \ell \log (\abs{\F^*} n)$ and coefficients $\cbra{\alpha_v}_{v \in \mcV}$ in $\F^*$ such that
     \begin{equation*}
         \sum_{v \in \mcV} \alpha_v \cdot v = 0\mper
     \end{equation*}
\end{theorem}

\begin{proof}
Assume $\mcH$ has no linearly dependent subsets of size $\leq \ell \log (\abs{\F^*} n)$. Let $A$ be the unsigned adjacency matrix for the level-$\ell$ Kikuchi graph of $\mcH$ as defined in \cref{def:kikuchimatrix}. Since $\mcH$ has no linearly dependent sets of size $\leq \ell \log (\abs{\F^*} n)$, all closed walks of size $\leq \ell \log (\abs{\F^*} n)$ correspond to trivially closed walks. It follows from the analysis in \cref{lem:countingbacktrackingwalks} with $t = \frac{1}{2} \log N \leq \ell \log (\abs{\F^*} n)$ that $\norm{\Gamma^{-1/2} A \Gamma^{-1/2}}_2 \leq O\left(\sqrt{\frac{ \log N}{d}}\right)$. This implies $A \preceq O\left(\sqrt{\frac{\log N}{d}}\right) \cdot \Gamma$, so
\begin{equation*}
    \1^\top A \1 \leq O\left(\sqrt{\frac{\log N}{d}}\right) \cdot \tr(\Gamma) = O\left(\sqrt{\frac{ \log N}{d}}\right) \cdot 2Nd\mper
\end{equation*}
We double count $\1^\top A \1 = Nd$. This implies that $d \leq O(\ell \log (\abs{\F^*} n))$. Using our alternate count of $d \geq \frac{\abs{\F^*}}{2} \left({\frac{\ell}{\abs{\F^*}n}}\right)^{k/2} \cdot \abs{\mcH}$ from \cref{fact:avgdegree} we conclude $\abs{\mcH} \leq O(n) \cdot \log\left(\abs{\F^*}n\right) \left( \frac{n \abs{\F^*}}{\ell}\right)^{k/2-1}$ as desired.
\end{proof}

\subsection{Generalized Feige's conjecture (odd-arity case)}

In this section we prove the more involved odd-arity case of above, \cref{thm:mainfeige}, which finishes \cref{thm:feige}.

\begin{theorem}[Generalized Feige's Conjecture, odd-arity]
    \label{thm:mainfeige}
     Fix $k \geq 3$ odd $\ell \geq k/2$. Let $\mcH$ be a set of $\abs{\mcH} \geq O(n) \cdot \log\left(\abs{\F^*}n\right) \left(\frac{n\abs{\F^*}}{\ell}\right)^{k/2-1}$ $k$-sparse vectors in $\F^n$. There exists a set $\mcV \subseteq \mathcal{H}$ with $\abs{\mcV} \leq \ell \log (\abs{\F^*} n)$ and coefficients $\cbra{\alpha_v}_{v \in \mcV}$ in $\F^*$ such that
     \begin{equation*}
         \sum_{v \in \mcV} \alpha_v \cdot v = 0\mper
     \end{equation*}
     That is, $\mcV$ is a linearly dependent subset of $\mathcal{H}$.
\end{theorem}

    The first piece in our proof of \cref{thm:mainfeige} is a regularity decomposition on the set of vectors $\mcH$ similar to the one specified in \cref{lem:decompositionalg}. This breaks the $\mcH$ into a few regular chunks, and our goal will be to find a small linearly dependent set in the largest chunk.

    \begin{lemma}[Regular decomposition algorithm for $k$-sparse vector sets]
    \label{lem:vecdecompositionalg} 
    Fix $k \geq 2$. There is an algorithm that takes as input $\mcH$ a collection of $k$-sparse vectors in $\F^n$ and outputs a partition (up to scalar multiple) parameterized by $0 \leq t \leq k-1$ into $t$-sparse $\mcU^{(t)}$-bipartite decompositions $\cbra{\mcH^{(t)}_u}_{u \in \mcU^{(t)}}$ for each $\mcH^{(t)}$ in time $\left(\abs{\F^*}n\right)^{O(\ell)}$ with the guarantees:
    \begin{enumerate}
        \item \label{item:vecdecompositionalg1} For $t \neq 0$ and all $u \in \mcU^{(t)}$, $\abs{\mcH^{(t)}_u} = \tau_t := \max\left(1, \left(\frac{n\abs{\F^*}}{\ell}\right)^{k/2-t}\right)$.
        \item \label{item:vecdecompositionalg2} For all $t \neq 0$, $\abs{\mcU^{(t)}} \leq \frac{\abs{\mcH}}{\tau_t}$.
        \item \label{item:vecdecompositionalg3} For all $t \neq 0$, the decomposition of $\mcH^{(t)}$ is $\ell$-regular.
        \item \label{item:vecdecompositionalg4} The largest $\mcH^{(t)}$ for $0 \leq t \leq k-1$ is not $\mcH^{(0)}$.
    \end{enumerate}
\end{lemma}

Let us first contrast this with \cref{lem:decompositionalg}, which is a nearly identical result. The largest difference is that the former algorithm partitions $k$-$\LIN$ instances, while this partitions sets of vectors, however this just means we have strictly less information to keep track of. We also have an extra partition $\mcH^{(0)}$ for which many guarantees do not hold. We have importantly \cref{item:vecdecompositionalg4} which guarantees that this partition is not the largest. Unlike in \cref{sec:mainrefutation}, where the proof involved refuting every subinstance partitioned, this proof will only look at the largest, meaning \cref{item:vecdecompositionalg4} can be seen as a guarantee that $\mcH^{(0)}$ is a ``garbage collector'' that can be safely ignored. Finally, we trade $(\varepsilon, \ell)$-regularity for $\ell$-regularity. With these differences highlighted, we give the proof of \cref{lem:vecdecompositionalg}.

\begin{proof}[Proof of \cref{lem:vecdecompositionalg}]
We give a modified algorithm from that of \cref{lem:decompositionalg}.

    \begin{tcolorbox}[
    width=\textwidth,   
    colframe=black,  
    colback=white,   
    title=Vector Set Regularity Decomposition Algorithm,
    colbacktitle=white, 
    coltitle=black,      
    fonttitle=\bfseries,
    center title,   
    enhanced,       
    frame hidden,           
    borderline={1pt}{0pt}{black},
    sharp corners,
    toptitle=2.5mm
]
\textbf{Input:} A set $\mcH$ of $k$-sparse vectors in $\F^n$.\\

\textbf{Output:} A partition of $\mcH$ satisfying the criteria of \cref{lem:vecdecompositionalg}.\\

\textbf{Algorithm:}
\begin{enumerate}
    \item Let $t = k-1$ and $\mcH_\text{curr} = \mcH$.
    \item While $\exists u \in \F^n$ with $\wt(u) = t$ such that $\abs{\{\beta v \in \mcH_\text{curr} \mid \beta \in \F^*, u \subseteq \beta v\}} \geq \tau_t := \max\left(2, \left(\frac{n \abs{\F^*}}{\ell}\right)^{k/2-t}\right)$, do the following. Otherwise, decrement $t$.
    \begin{enumerate}
        \item Let $\mcH_u^{(t)}$ hold $\beta v$ for exactly $\tau_t$ such vectors and move $\mcH_u^{(t)}$ from $\mcH_{\text{curr}}$ to $\mcH^{(t)}$.
    \end{enumerate}
    \item When $t = 0$, add all remaining $v \in \mcH_{\text{curr}}$ to $\mcH^{(0)}$.
\end{enumerate}

\end{tcolorbox}
The proof of the first three items mimics exactly the proof of \cref{lem:decompositionalg} up to parameter changes in $\tau$. The following observation encapsulates \cref{item:vecdecompositionalg4}.

\begin{observation}
    \label{obs:onepartlarge}
    Let $\kappa$ index the largest $\mcH^{(\kappa)}$ output above, then $\abs{\mcH^{(\kappa)}} \geq \frac{\abs{\mcH}}{k}$ and $\kappa \neq 0$.
\end{observation}

The size is by pigeonhole principle. $\mcH^{(0)}$ cannot have more than $\tau_1$ vectors with intersecting support by Step 2, otherwise they would have been added to $\mcH^{(1)}$ by the greediness of the algorithm. There are crudely at most $\tau_1 n < \frac{\abs{\mcH}}{k}$ such vectors counting $n$ indices to intersect and $\tau_1$ for each. Finally, we remark in Step 2a, scalar multiples are added, which importantly does not change the appearance of linearly dependent subsets.
\end{proof}

Finally, note that $\mcH^{(\kappa)}$ has no linearly dependent subsets if $\mcH$ has none. Our overarching goal is to show that $\mcH^{(\kappa)}$ not having such a set implies that it is small, and, by the extension of our relative size lower bound, we conclude $\mcH$ is small.

Our proof proceeds similarly to the even case, utilizing a spectral double counting argument, but using the odd-arity Kikuchi matrix defined in \cref{sec:mainrefutation}.

\begin{definition}{(Odd-arity Kikuchi matrix over $\F$).}
Let $k/2 \leq \ell \leq n/2$ be a parameter and let $N = \abs{\F^*}^\ell {2n \choose \ell}$. For each pair $(v, v')$ of distinct $(k-t)$-sparse vectors in $\F^n$ and $\beta \in \F^*$, we define a matrix $A_{v, v', \beta} \in \C^{N \times N}$ as follows. First, we identify $N$ with the set of pairs of vectors $(U^{(1)}, U^{(2)})$ in $\F^n$ with the condition $\wt(U^{(1)}) + \wt(U^{(2)}) = \ell$. Then for any such pairs $U$ and $V$ we let
    \begin{equation*}
        A_{v, v', \beta}(U, V) = \begin{cases}
                      1  & U\xrightarrow{\text{$v, v', \beta$}} V\\
                      0 & \text{otherwise}
                    \end{cases}
    \end{equation*}
    where we say $U\xrightarrow{\text{$v, v', \beta$}} V$ if the following conditions hold
    \begin{enumerate}
        \item $U^{(1)}\xrightarrow{\text{$v, \beta$}} V^{(1)}$.
        \item $U^{(2)}\xrightarrow{\text{$v', -\beta$}} V^{(2)}$.
        \item $\abs{\supp(U^{(1)}) \oplus \supp(v)} = \lfloor\frac{k-t}{2}\rfloor$ and $\abs{\supp(U^{(2)}) \oplus \supp(v')} = \lceil \frac{k-t}{2} \rceil$ or vice versa.
    \end{enumerate}
    
For $\mcH^{(t)} = \cbra{\mcH_u^{(t)}}_{u \in \mcU^{(t)}}$ we let $A^{(t)} = \sum_{u \in \mcU^{(t)}} \sum_{\beta \in \F^*} \sum_{v \neq v' \in \mcH^{(t)}_u} A^u_{v, v', \beta}$ be the corresponding level-$\ell$ Kikuchi matrix.
\end{definition}

As with \cref{sec:mainrefutation}, the Kikuchi matrix allows us to provide a relevant spectral bound.

\begin{lemma}
    \label{lem:feigecountbackwalks}
    Let $A$ be the level-$\ell$ Kikuchi matrix for a $\mcU$-bipartite vector collection $\mcH = \cbra{\mcH_u}_{u \in \mcU}$ of $k$-sparse vectors from $\F^n$ and complex coefficients $\{c_{v, \beta}\}_{\substack{v \in \mcH \\ \beta \in \F^*}}$. Let $\Gamma \in \C^{N \times N}$ be $\Gamma = D + d \Id$ where $D_{U, U} := \deg(U)$ and average degree $d$. Suppose additionally that the underlying Kikuchi graph is $\eta$-bounded local degree and that $\mcH$ contains no linearly dependent subsets of size $\leq \ell \log(\abs{\F^*} n)$. Then,
    \begin{equation*}
    \norm{\Gamma^{-1/2} A \Gamma^{-1/2}}_2 \leq 8\sqrt{\frac{\eta\ell \log (\abs{\F^*} n)}{d}}\mper
    \end{equation*}
\end{lemma}

With the two main lemmas identified, we are able to prove \cref{thm:mainfeige}. We delay the proof of \cref{lem:feigecountbackwalks}.

\begin{proof}[Proof of \cref{thm:mainfeige} from \cref{lem:feigecountbackwalks}]
    First, the algorithmic result in \cref{lem:vecdecompositionalg} guarantees a partition of $\mcH$ into regular $\mcH^{(t)}$ for $t \in [k-1]$ and an extra set $\mcH^{(0)}$. Let $\mcH^{(\kappa)}$ be the largest regular partition and we have by \cref{item:vecdecompositionalg4} that $\kappa \in [k-1]$, so $\mcH^{(\kappa)}$ is regular. Similar to in \cref{sec:mainrefutation}, \cref{lem:feigecountbackwalks} requires bounded local degree, so we identify a subgraph of the Kikuchi graph of $\mcH^{(\kappa)}$.

    \begin{lemma}
        \label{lem:genedgedeletion}
        Let $K_\ell$ be a level-$\ell$ Kikuchi graph with average degree $d(K_\ell)$, for each $v \neq v' \in \mcH^{(t)}$ $(k-t)$-sparse and $\beta \in \F^*$ there are $\Delta$ edges of type $(v, v', \beta)$, and satisfying the output criteria in \cref{lem:vecdecompositionalg}. Then we can find a subgraph $\hat{K}_\ell$ in time $(\abs{\F^*}n)^{O(\ell)}$ with the following properties:
        \begin{itemize}
            \item $\hat{K}_\ell$ is $D^k$-bounded degree for some constant $D > 0$.
            \item $d(\hat{K}_\ell) \geq \frac{1}{2}d(K_\ell)$.
        \end{itemize}
    \end{lemma}

    We finish the proof by a spectral double counting argument. \cref{lem:genedgedeletion} guarantees a suitable subgraph $\hat{K}_\ell$ of the Kikuchi graph associated with $\mcH^{(\kappa)}$. We apply \cref{lem:feigecountbackwalks} on the adjacency matrix $\hat{A}^{(\kappa)}$ for $\hat{K}_\ell$. This yields
    \begin{equation*}
        \1^\top \hat{A}^{(\kappa)} \1 \leq 8\sqrt{\frac{\eta \ell \log (\abs{\F^*} n)}{d(\hat{K}_\ell)}} \cdot \tr(\Gamma) = 16\sqrt{\frac{\eta \ell \log (\abs{\F^*} n)}{d(\hat{K}_\ell)}} \cdot Nd(\hat{K}_\ell)\mper
    \end{equation*}
    since $\hat{A}^{(\kappa)} \preceq \norm{\Gamma^{-1/2} \hat{A}^{(\kappa)} \Gamma^{-1/2}}_2 \cdot \Gamma$. Since $\onevec^\top \hat{A}^{(\kappa)} \onevec = Nd(\hat{K}_\ell)$ we get that $d(\hat{K}_\ell) \leq 256\eta \ell \log (\abs{\F^*}n)$. Further, using that $\eta \leq D^k$ since $\hat{A}^{(k)}$ is output from \cref{lem:genedgedeletion}, we conclude $d(\hat{K}_\ell) \leq 256D^k\ell \log(\abs{\F^*}n)$. We compare this to the following lower bound.

    \begin{observation}
\label{clm:degreelowerbound}
$d(\hat{K}_\ell) \geq \frac{\abs{\F^*}}{4k^22^k} \left({\frac{\ell}{n \abs{\F^*}}}\right)^{k/2} \cdot \abs{\mcH}$.
\end{observation}

\begin{proof}
Starting with \cref{fact:genavgdegree} we have
\begin{flalign*}
    d(K_\ell) &\geq \abs{\F^*}\left({\frac{\ell}{2n \abs{\F^*}}}\right)^{k-\kappa} \sum_{u \in \mcU^{(\kappa)}} {\abs{\mcH^{(\kappa)}_u} \choose 2} 
    \geq \abs{\F^*}\left({\frac{\ell}{2n \abs{\F^*}}}\right)^{k-\kappa} \abs{\mcU^{(\kappa)}} {\abs{\mcH}/k\abs{\mcU^{(\kappa)}} \choose 2} \\
    &\geq \abs{\F^*}\left({\frac{\ell}{2n \abs{\F^*}}}\right)^{k-\kappa} \frac{\abs{\mcH}^2}{4k^2\abs{\mcU^{(\kappa)}}}\mper
\end{flalign*}

Using $\abs{\mcU^{(\kappa)}} \leq \frac{\abs{\mcH}}{\tau_\kappa} \leq \abs{\mcH} \cdot \left(\frac{\ell}{n \abs{\F^*}}\right)^{k/2-\kappa}$, Jensen's inequality, and finally that $d(\hat{K}_\ell) \geq \frac{1}{2}d(K_\ell)$ we obtain the claimed bound.
\end{proof}

Combining the upper and lower bounds yields $\abs{\mcH} \leq 256D^k n \log (\abs{\F^*} n) \left(\frac{n \abs{\F^*}}{\ell}\right)^{k/2-1}$, completing the proof of \cref{thm:mainfeige}.
\end{proof}

We finish this section by completing the outstanding proofs of \cref{lem:feigecountbackwalks} and \cref{lem:genedgedeletion}. The proof of \cref{lem:feigecountbackwalks} follows from the proof of \cref{lem:maincountbackwalks} given the following fact connecting linearly dependent subsets and trivially closed walks.

\begin{fact}
    In a level-$\ell$ Kikuchi graph $K_\ell$ specified from $\mcH$ a set of $k$-sparse vectors in $\F^n$, any walk of length $\leq \ell \log (\abs{\F^*} n)$ that is not trivially closed immediately implies a linearly dependent subset of size $\leq 2\ell \log (\abs{\F^*}n)$ in $\mcH$.
\end{fact}

\begin{proof}
    Observe that a closed walk $U_1 \xrightarrow{v_1, v'_1, \beta_1} ... \xrightarrow{v_{2t-1}, v'_{2t-1}, \beta_{2t-1}} U_{2t}$ has for $b \in \{1,2\}$
    \begin{equation*}
        \sum_{i = 1}^{2t} U^{(b)}_i - U^{(b)}_{i+1 \mod 2t} = U^{(b)}_1 - U^{(b)}_{2t} = 0\mcom
    \end{equation*}
    by closure. By definition $U_i^{(b)} \xrightarrow{v_i, \beta_i} U_{i+1}^{(b)}$ means $U_i^{(b)} - U_{i+1}^{(b)} = \beta v_i$ which implies individually that $\sum_{i=1}^{2t} \beta_i v_i$ and $\sum_{i=1}^{2t} -\beta_i v_i'$ are $0$, allowing us to conclude
    \begin{equation*}
        \sum_{i=1}^{2t} \beta_i v_i -\beta_i v_i' = 0\mper
    \end{equation*}
    If there is any vector $v \in \mcH$ that is not trivially closed then it appears in the above sum with a non-zero coefficient and the above sum is a nontrivial linear dependence.
\end{proof}

The proof of \cref{lem:genedgedeletion} mirrors \cref{lem:edgedeletion} given slight modifications to the algorithm as shown below.

       \begin{tcolorbox}[
    width=\textwidth,   
    colframe=black,  
    colback=white,   
    title=Edge Deletion Algorithm,
    colbacktitle=white, 
    coltitle=black,      
    fonttitle=\bfseries,
    center title,   
    enhanced,       
    frame hidden,           
    borderline={1pt}{0pt}{black},
    sharp corners,
    toptitle=2.5mm
]
\textbf{Input:} A Kikuchi graph $K_\ell$ specified by $\mcH$.\\

\textbf{Output:} A subgraph $\hat{K}_\ell$ of $K_\ell$ of $\eta$-bounded local degree and $d(\hat{K}_\ell) \geq \frac{1}{2}d(K_\ell)$.\\

\textbf{Algorithm:}
\begin{enumerate}
    \item While there is a vertex $U$ and $v \in \mcH$ s.t. $d_{U,v,0}$ or $d_{U, v, 1} > \eta$, delete an arbitrary edge.
    \item Output the resulting graph.
\end{enumerate}

\end{tcolorbox}

\begin{proof}[Proof of \cref{lem:genedgedeletion}]
    The algorithm is the same as the one in \cref{sec:mainrefutation} with parameter $\eta = D^k$ and without the degree regularizing step (which only lowers the degree). As such, the analysis remains exactly the same when we drop $\varepsilon$ terms.
\end{proof}

\section{Sum-of-Squares Lower Bounds for \texorpdfstring{$k$-$\LIN(\F)$}{k-LIN(F)}: Proof of \cref{thm:infsoslowerbound}}
\label{sec:sos}

In this section we prove \cref{thm:infsoslowerbound}. Our proof follows and generalizes the low-width resolution framework developed in \cite{Grigoriev01, Schoenebeck08} and its exposition in \cite{BarakS16}. We begin by defining two natural forms of polynomial optimization problems to represent $k$-$\LIN(\F)$.

\subsection{Sum-of-Squares relaxations for \texorpdfstring{$k$-$\LIN(\F)$}{k-LIN(F)}}

\begin{definition}[Indicator $k$-$\LIN(\F)$ constraints]
    \label{def:indicator}
    Fix a $k$-$\LIN(\F)$ instance $\mcI = (\mcH, \{b_v\}_{v \in \mcH})$ with $n$ variables and where $\textrm{char}(\F) = p$ and consider the following optimization problem.

    \begin{itemize}
        \item \textbf{Variables:} For each $i \in [n]$, $\alpha \in \F$, $x_{i, \alpha} \in \R$.

        \item \textbf{Constraints:}
        \begin{enumerate}[(1)]
            \item (Booleanity) $\cbra{x_{i, \alpha}^2 = x_{i, \alpha}}_{i \in [n], \alpha \in \F}$;
            \item (Sum) $\cbra{\sum_{\alpha \in \F} x_{i, \alpha} = 1}_{i \in [n]}$.
        \end{enumerate}
        
        \item \textbf{Objective:} $\max_{x \in \R^{\abs{\F}n}} \frac{1}{\abs{\mcH}}\sum_{v \in \mcH} \sum_{{\substack{\alpha \in (\F^*)^{\supp(v)}\\ \sum_{i \in \supp(v)} \alpha_i = b_v}}} \prod_{i \in \supp(v)} x_{i, \alpha_i v_i^{-1}}$.

    \end{itemize}

\end{definition}

In \cref{def:indicator} we have that each $x_{i, \alpha}$ is the $0$-$1$ indicator for whether or not $x_i$ is assigned $\alpha \in \F$. The booleanity constraints enforce that each is truly $0$-$1$ while the sum constraints enforce that there is exactly one assignment for each. We can also observe that the maximization asks how well the best assignment satisfies the $k$-$\LIN$ instance by summing over every satisfying assignment for each constraint. With a proper polynomial system defined, we can now state how Sum-of-Squares fails to refute random instances.

\begin{theorem}[\Cref{thm:infsoslowerbound} restated]
    \label{thm:infsoslowerbound1}
    Fix $k \geq 3$ and $\frac{n}{\max(\abs{\F^*}, k)} \geq \ell \geq k$. Let $\mcI$ be a random $k$-$\LIN(\F)$ instance $\abs{\mcH} \leq O(n) \cdot \left(\frac{n\abs{\F^*}}{\ell}\right)^{k/2-1} \cdot \varepsilon^{-2}$. Then, with high probability over the draw of $\mcI$, it holds that
    \begin{enumerate}
        \item $\val(\mcI) \leq \frac{1}{\abs{\F}} + \varepsilon$.
        \item The degree-$\tilde{O}(\ell)$ Sum-of-Squares relaxation for indicator $k$-$\LIN(\F)$ given in \cref{def:indicator} fails to refute $\mcI$.
    \end{enumerate}
\end{theorem}

Using standard SDP duality, the task of proving \Cref{thm:infsoslowerbound1} reduces to constructing a dual witness pseudo-expectation over $\Bits^{\abs{\F}n}$ as defined in \Cref{def:pseudo-expectation}. The exact operator we construct has the following properties.

\begin{definition}[Pseudo-expectations for indicator $k$-$\LIN(\F)$ over the hypercube] \label{def:indpseudo-expectation}
Our degree $d$ pseudo-expectation $\pE$ over $\Bits^{\abs{\F}n}$ maps degree $\leq d$ polynomials on $\Bits^{\abs{\F}n}$ to reals with the properties:
\begin{enumerate}
    \item (Normalization) $\pE[1] = 1$.
	\item (Booleanity) For any $x_{i, \alpha}$ and any polynomial $f$ of degree $\leq d-2$, $\pE[f (x_{i, \alpha}^2 - x_{i,\alpha})] = 0$. 
	\item (Positivity) For any polynomial $f$ of degree at most $d/2$, $\pE[f^2] \geq 0$.
    \item ($k$-$\LIN(\F)$ Sum) For any $i \in [n]$ and polynomial $f$ of degree at most $d-1$, $\pE[f\cdot(\sum_{\alpha \in \F} x_{i,\alpha} -1)] = 0$.
\end{enumerate} 
\end{definition}

Recall that the degree $d$ Sum-of-Squares algorithm for \cref{def:indicator} given by \cref{fact:sosalg} bounds the value of the instance polynomial, equivalently the value of the instance $\mcI$, by the polynomial's maximum value under any indicator $\pE$ consistent with \cref{def:indpseudo-expectation}. Our goal is to show a random $k$-$\LIN(\F)$ instance $\mcI$ at equation threshold $O(n) \cdot \left(\frac{n\abs{\F^*}}{\ell}\right)^{k/2-1} \cdot \varepsilon^{-2}$ satisfies the following with large probability: (1) the instance is at most $\frac{1}{\abs{\F}} + \varepsilon$-satisfiable and (2) there exists a pseudo-expectation $\pE$ assigning the indicator polynomial value $1$.

It turns to be more natural to construct $\pE$ by reducing to the construction of a slightly different pseudo-expectation built from viewing $\val(\mcI)$ as a complex polynomial optimization problem.

\begin{definition}[Embedding $k$-$\LIN(\F)$ into $\C$]
    \label{def:complexklin}
    Fix a $k$-$\LIN(\F)$ instance $\mcI = (\mcH, \{b_v\}_{v \in \mcH})$ with $n$ variables and where $\textrm{char}(\F) = p$. Define the following related optimization problem over $\C$.

    \begin{itemize}
        \item \textbf{Variables:} For each $i \in [n]$, $\alpha \in \F^*$, $y_{i, \alpha} \in \C$.

        \item \textbf{Constraints:}
        \begin{enumerate}
            \item (Validity) $\cbra{y_{i, \alpha}^p = 1}_{i \in [n], \alpha \in \F^*}$.
            \item (Consistency) $\cbra{y_{i, \alpha} \cdot \overline{y_{i, \beta}} = y_{i, \alpha - \beta}}_{i \in [n], \alpha, \beta \in \F^*}$ letting $y_{i, 0} = 1$ by convention.
        \end{enumerate}
        
        \item \textbf{Objective:} $\max_{y \in \C^{\abs{\F^*}n}} \frac{1}{\abs{\F}} + \frac{1}{\abs{\mcH}\abs{\F}}\sum_{v \in \mcH} \sum_{\beta \in \F^*} \omega_p^{\Tr(\beta b_v)} \cdot \prod_{i \in \supp(v)} \overline{y_{i, \beta v_i}}$.

    \end{itemize}

\end{definition}

With some work one can prove that this system has objective equivalent to $\val(\mcI)$. Now, given a complex polynomial system as above, we can develop a notion of pseudo-expectation over complex polynomials respecting the constraints as follows. First, we introduce the following non-standard notion of degree for the polynomials arising in \cref{def:complexklin}.

\begin{definition}[Variable degree]
    For a monomial $P$ in $n\abs{\F^*}$ indeterminates $y_{i, \alpha}$, we define $\vardeg(P)$ to be the number of $i \in [n]$ such that $y_{i, \alpha}$ appears for some $\alpha \in \F^*$ with non-zero exponent. A polynomial's variable degree is the maximum across all its monomials. A degree $d$ pseudo-expectation then respects the standard pseudo-expectation constraints with respect to variable degree.
\end{definition}

With a suitable notion of degree for polynomials over $\C^{\abs{\F^*}n}$, we define a corresponding pseudo-expectation as such.

\begin{definition}[Pseudo-expectations for $k$-$\LIN(\F)$ over $\C$] \label{def:fourierpseudo-expectation} A degree $d$ pseudo-expectation $\pE$ maps degree $\leq d$ polynomials in $\abs{\F^*}n$ complex indeterminates to complex numbers with the properties:
\begin{enumerate}
    \item (Normalization) $\pE[1] = 1$.
	\item (Positivity) For any polynomial $f$ of variable degree at most $d/2$, $\pE[\abs{f}^2] \geq 0$.
    \item (Validity) For any $y_{i, \alpha}$ and polynomial $f$ of variable degree at most $d-1$, $\pE\sbra{f \cdot (y_{i,\alpha}^p -1)} = 0$.
    \item (Consistency) For any $y_{i, \alpha}$ and $y_{i,\beta}$ and polynomial $f$ of variable degree at most $d-2$, $\pE\sbra{f \cdot (y_{i,\alpha} \cdot \overline{y_{i,\beta}} - y_{i, \alpha - \beta})} = 0$ letting $y_{i, 0} = 1$ by convention.
\end{enumerate} 
\end{definition}

\subsection{Reducing from indicator to complex pseudo-expectations}

Our roadmap for \cref{thm:infsoslowerbound1} is to first construct a ``complex'' pseudo-expectation $\pE$ consistent with \cref{def:fourierpseudo-expectation} and then utilize it to construct a ``Boolean'' $\pE'$ consistent with \cref{def:indpseudo-expectation}.

\begin{lemma}
    \label{lem:pseudoreduction}
Let $\mcI = (\mcH, \{b_v\}_{v \in \mcH})$ be a $k$-$\LIN(\F)$ instance with $n$ variables where $\textrm{char}(\F) = p$ and let $P_\mcI$ be the indicator polynomial and $Q_\mcI$ the complex embedding of $\mcI$. Suppose $\pE$ is a degree $d$ complex pseudo-expectation (as in \cref{def:fourierpseudo-expectation}) under which $Q_\mcI$ has objective value $1$. Then there exists a degree $d$ indicator pseudo-expectation $\pE'$ (as in \cref{def:indpseudo-expectation}) under which $P_\mcI$ objective value $1$.   
\end{lemma}

\cref{lem:pseudoreduction} can be seen as a reduction from constructing indicator pseudo-expectations to complex pseudo-expectations, a more natural task we accomplish in the next section. The main idea is to define a mapping between degree $d$ Boolean polynomials and complex polynomials such that the properties we want, such as positivity, propagate through from $\pE$ to $\pE'$. In a sense what we prove is more than just a Sum-of-Squares lower bound against the standard Boolean relaxation; we also rule out any algorithms based on maximizing across the complex pseudo-expectations seen in \cref{def:fourierpseudo-expectation}.

\begin{proof}[Proof of \cref{lem:pseudoreduction}]
    Let $\pE$ be a degree $d$ complex pseudo-expectation as in the hypothesis. As is standard, we define our degree $d$ pseudo-expectation $\pE'$ on the basis of degree $d$ monomials and extending through linearity. Moreover, since $\pE'$ enforces booleanity we can restrict attention to just the multilinear monomials, as any non-multilinear monomial must match the value of the multilinear monomial generated by making all non-zero exponents $1$. 
    
    To specify the values of $\pE'$, we use the following simple mapping $\varphi$ from Boolean polynomials in $\{x_{i,\alpha}\}_{i \in [n], \alpha \in \F}$ to complex polynomials in $\{y_{i,\beta}\}_{i \in [n], \beta \in \F^*}$. For any Boolean multilinear monomial and every variable $x_{i, \alpha}$, replace $x_{i,\alpha}$ with the complex polynomial $\frac{1}{\abs{\F}}+ \frac{1}{\abs{\F}}\sum_{\beta \in \F^*} \omega_p^{-\Tr(\beta \alpha)} y_{i, \beta}$ and simplify. We can extend $\varphi$ to polynomials by linearity and scalar multiplicativity. Note that it is standardly multiplicative by construction. The majority of constraints propagate nicely thanks to the following observation.

    \begin{observation}
        \label{obs:degreematch}
        Given a Boolean monomial $P$ of degree $d$, $\varphi(P)$ has variable degree at most $d$.
    \end{observation}

    \begin{proof}
        This can be seen by simply expanding $\varphi(P)$. Let $P = \prod_{i \in S} x_{i, \alpha_i}$ for some choice $\{\alpha_i\}_{i \in S}$. The degree of $P$ is $d$, the size of $S$. We compute
        \begin{equation*}
            \varphi(P) = \prod_{i \in S} \frac{1}{\abs{\F}} + \frac{1}{\abs{\F}}\sum_{\beta \in \F^*} \omega_p^{-\Tr(\beta \alpha_i)} y_{i, \beta} = \frac{1}{\abs{\F}^d}\sum_{w \in \F^S} \prod_{i \in S}\omega_p^{-\Tr(w_i\alpha_i)} \cdot y_{i, w_i} \mcom
        \end{equation*}
        where again $y_{i, 0}$ is $1$ by convention. Observe that every term above has variable degree at most $d$ as desired.
    \end{proof}

    Now we claim we can simply let $\pE'\sbra{P} = \pE\sbra{\varphi(P)}$ for any degree $d$ multilinear monomial $P$ and prove that $\pE'$ satisfies the constraints of \cref{def:indpseudo-expectation} and has objective value $1$ for the underlying $k$-$\LIN(\F)$ instance.

    \begin{lemma}[Booleanity]
        For any $x_{i, \alpha}$ and any polynomial $f$ of degree $\leq d-2$, $\pE'[f (x_{i, \alpha}^2 - x_{i,\alpha})] = 0$. 
    \end{lemma}

    \begin{proof}
        By linearity it suffices to prove this for any degree at most $d-2$ monomial $P$. We now show $\pE'\sbra{Px_{i,\alpha}^2} = \pE'\sbra{Px_{i,\alpha}}$. Notice immediately that $Px_{i,\alpha}^2$ and $Px_{i,\alpha}$ give the same multilinear monomial after reducing exponents, thus they are set the same in $\pE'$.
    \end{proof}

    \begin{lemma}[$k$-$\LIN(\F)$ Sum]
        For any $i \in [n]$ and polynomial $f$ of degree at most $d-1$, $\pE'[f\cdot(\sum_{\alpha \in \F} x_{i,\alpha} -1)] = 0$.
    \end{lemma}

    \begin{proof}
        By linearity it suffices to prove for any degree at most $d-1$ monomial $P$ that $\pE'\sbra{P \cdot \sum_{\alpha \in \F} x_{i,\alpha}} = \pE'\sbra{P}$. So we compute
        \begin{align*}
            \pE'\sbra{P \cdot \sum_{\alpha \in \F} x_{i,\alpha}} &= \pE\sbra{\varphi\left(P \cdot \sum_{\alpha \in \F} x_{i,\alpha}\right)}\\
            &= \pE\sbra{\varphi(P) \cdot \sum_{\alpha \in \F} \varphi(x_{i,\alpha})}\\
            &= \pE\sbra{\varphi(P) \cdot \frac{1}{\abs{\F}} \sum_{\alpha \in \F} \sum_{\beta \in \F} \omega_p^{-\Tr(\beta \alpha)} y_{i, \beta}}\\
            &=  \frac{1}{\abs{\F}} \sum_{\alpha \in \F} \sum_{\beta \in \F} \omega_p^{-\Tr(\beta \alpha)} \cdot \pE\sbra{\varphi(P) \cdot y_{i, \beta}}\\
            &=  \pE\sbra{\varphi(P)} + \frac{1}{\abs{\F}} \sum_{\beta \in \F^*} \pE\sbra{\varphi(P) \cdot y_{i, \beta}} \sum_{\alpha \in \F} \omega_p^{-\Tr(\beta \alpha)}\mper
        \end{align*}
        We finish by observing that the sum $\sum_{\alpha \in \F} \omega_p^{-\Tr(\beta \alpha)}$ is symmetric so is $0$.
    \end{proof}

    \begin{lemma}[Positivity]
        \label{lem:positivity1}
        For any polynomial $f$ of degree at most $d/2$, $\pE'[f^2] \geq 0$.
    \end{lemma}

    \begin{proof}
        Since $\pE'\sbra{f^2} = \pE\sbra{\varphi(f)^2}$ the claim reduces to showing we can replace $\varphi(f)$ with $\overline{\varphi(f)}$ and appeal to the positivity of $\pE$. By linearity and the fact that for any real coefficient $c$ we have $c = \overline{c}$ it suffices to show we can freely replace any degree at most $d/2$ monomial $\varphi(P)$ by $\overline{\varphi(P)}$. So we fix $P = \prod_{i \in S} x_{i, \alpha_i}$ for $S \subseteq [n]$ and $\abs{S} \leq d/2$ and compute
        \begin{align*}
            \pE\sbra{\varphi(f) \cdot \varphi(P)} &= \frac{1}{\abs{\F}^d}\sum_{w \in \F^S} \prod_{i \in S}\omega_p^{-\Tr(w_i\alpha_i)} \cdot \pE\sbra{\varphi(f)\cdot \prod_{i \in S} y_{i, w_i}} \mper
        \end{align*}
        Now if we replace $\varphi(P)$ with $\overline{\varphi(P)}$ above we get
        \begin{align*}
            \pE\sbra{\varphi(f) \cdot \overline{\varphi(P)}} &= \frac{1}{\abs{\F}^d}\sum_{w \in \F^S} \prod_{i \in S}\omega_p^{\Tr(w_i\alpha_i)} \cdot \pE\sbra{\varphi(f)\cdot \prod_{i \in S} \overline{y_{i, w_i}}} \mper
        \end{align*}
        Note by the consistency constraint of $\pE$ we can replace $\overline{y_{i, w_i}}$ with $y_{i, -w_i}$. The terms are then identical since we sum over all $w \in \F^S$ since we can just relabel $w_i$ as $-w_i$. 
    \end{proof}

    With all constraints met, we finish by confirming that the objective value of the instance polynomial is $1$ under $\pE'$ given it is $1$ under $\pE$. Recall the instance polynomial looks like $\frac{1}{\abs{\mcH}}\sum_{v \in \mcH} \sum_{{\substack{\alpha \in (\F^*)^{\supp(v)}\\ \sum_{i \in \supp(v)} \alpha_i = b_v}}} \prod_{i \in \supp(v)} x_{i, \alpha_i v_i^{-1}}$. By linearity we have
    \begin{align*}
        &\pE'\sbra{\frac{1}{\abs{\mcH}}\sum_{v \in \mcH} \sum_{{\substack{\alpha \in (\F^*)^{\supp(v)}\\ \sum_{i \in \supp(v)} \alpha_i = b_v}}} \prod_{i \in \supp(v)} x_{i, \alpha_i v_i^{-1}}}\\
        &=  \frac{1}{\abs{\mcH}}\sum_{v \in \mcH} \sum_{{\substack{\alpha \in (\F^*)^{\supp(v)}\\ \sum_{i \in \supp(v)} \alpha_i = b_v}}} \pE'\sbra{\prod_{i \in \supp(v)} x_{i, \alpha_i v_i^{-1}}}\\
        &=  \frac{1}{\abs{\mcH}}\sum_{v \in \mcH} \sum_{{\substack{\alpha \in (\F^*)^{\supp(v)}\\ \sum_{i \in \supp(v)} \alpha_i = b_v}}} \pE\sbra{\varphi\left(\prod_{i \in \supp(v)} x_{i, \alpha_i v_i^{-1}}\right)}\\
        &=  \frac{1}{\abs{\mcH}}\sum_{v \in \mcH} \frac{1}{\abs{\F}^k} \sum_{{\substack{\alpha \in (\F^*)^{\supp(v)}\\ \sum_{i \in \supp(v)} \alpha_i = b_v}}} \sum_{w \in \F^{\supp(v)}}  \prod_{i \in \supp(v)}\omega_p^{-\Tr(w_i\alpha_iv_i^{-1})} \pE\sbra{\prod_{i \in \supp(v)} y_{i, w_i}}\mper
    \end{align*}
    Suppose for any $v \in \mcH$ we take $w_i = \beta v_i$ for some constant $\beta \in \F$ in the sum above. The inner term then simplifies to $\prod_{i \in \supp(v)}\omega_p^{-\Tr(\beta\alpha_i)} \pE\sbra{\prod_{i \in \supp(v)} y_{i, \beta v_i}}$. By the assumption that $\pE$ satisfies the constraints, we get that $\pE\sbra{\prod_{i \in \supp(v)} y_{i, \beta v_i}} = \omega^{\Tr(\beta b_v)}$. Since $\sum_{i \in \supp(v)} \alpha_i = \beta_v$ we also have $\prod_{i \in \supp(v)}\omega_p^{-\Tr(\beta\alpha_i)} = \omega_p^{-\Tr(\beta b_v)}$ and get a contribution of $1$ to the sum. This works for any choice $v \in \mcH$, any choice of satisfying assignment (of which there are $\abs{\F}^{k-1}$ in $k$-$\LIN(\F)$, and any choice of $\beta$, so we get a total of $\abs{\F}^k \abs{\mcH}$ such contributions. In total, this gives a sum of $1$.

    It then suffices to show that all terms outside of these contribute (that is $w$ is not a scalar multiple of $v$) a total of $0$. The main idea is to view the sum over $\alpha$ as choosing a uniformly random satisfying assignment. While this is not completely random, it is uniform on every $k-1$ marginal thanks to the structure of $k$-$\LIN(\F)$. As such, if we are able to eliminate just one $\alpha_i$ term, then suddenly the rest are independent and the $\omega_p^{-\Tr(w_j\alpha_jv_j^{-1})}$ term becomes mean $0$.

    To get such a collapse, we case on what $w \in \F^{\supp(v)}$ looks like. First, assume $w_i = 0$ for any $i \in [n]$. In this case we are immediately done; $\omega_p^{-\Tr(w_i\alpha_iv_i^{-1})} = 1$ so the dependence on $\alpha_i$ is removed. Now we assume $w_i \neq 0$ for all $i \in [n]$. Necessarily $w_1 = \beta v_1$ (assuming without loss of generality $1 \in \supp(v)$) for some $\beta$ while there exists some $i \in \supp(v)$ such that $w_i \neq \beta v_i$. Now we simply rewrite the sum above as such
    \begin{equation*}
        \frac{1}{\abs{\mcH}}\sum_{v \in \mcH} \frac{1}{\abs{\F}^k} \sum_{{\substack{\alpha \in (\F^*)^{\supp(v)}\\ \sum_{i \in \supp(v)} \alpha_i = b_v}}} \sum_{w \in \F^{\supp(v)}} \omega_p^{-\Tr(\beta b_v)} \prod_{i \in \supp(v)}\omega_p^{-\Tr((w_i-\beta v_i)\alpha_iv_i^{-1})} \pE\sbra{\prod_{i \in \supp(v)} y_{i, w_i}}\mper
    \end{equation*}
    We are allowed to factor the root coefficients due to linearity of the trace map. Now we just note that $\prod_{i \in \supp(v)}\omega_p^{-\Tr((w_i-\beta v_i)\alpha_iv_i^{-1})}$ is independent of $\alpha_1$ since $w_1 - \beta v_1 = 0$, thus its means is $0$ under uniform satisfying $\alpha$.
\end{proof}

\subsection{Complex max-entropy pseudo-expectations}

With \cref{lem:pseudoreduction} in mind, the proof of \cref{thm:infsoslowerbound1} is finished by showing that a random $k$-$\LIN(\F)$ instance at the proper equation threshold has a complex pseudo-expectation $\pE$ as in \cref{def:fourierpseudo-expectation} that believes it is satisfiable. 

The key technical result turns out to be the below theorem guaranteeing an expansion property for sets of random vectors. We show that good expansion in the following sense implies the existence of our desired pseudo-expectation $\pE$.

\begin{theorem}[Expansion in random $k$-sparse vector sets]
    \label{thm:inversefeige}
     Fix $k/2 \leq \ell \leq n/(\max(\abs{\F^*}, k))$. Let $\mcH$ be a set of $\abs{\mcH} \leq \Omega(\delta n) \cdot  \left(\frac{n\abs{\F^*}}{\ell}\right)^{k/2-1-\beta}$ uniformly random $k$-sparse vectors in $\F^n$. Then with probability at least $1-\delta$ all sets $\mcV \subseteq \mathcal{H}$ with $\abs{\mcV} \leq \ell$ and coefficients $\cbra{\alpha_v}_{v \in \mcV}$ in $\F^*$ have that
     \begin{equation*}
         \left|\sum_{v \in \mcV} \alpha_v \cdot v\right| > \beta \abs{\mcV}\mper
     \end{equation*}
\end{theorem}

We prove \cref{thm:inversefeige} in \cref{append:expansion}. While seemingly a blanket combinatorial fact, \cref{thm:inversefeige}'s connection to refuting $k$-$\LIN$ instances can be seen through the following definition.

\begin{definition}[Refutation]
    A \textit{refutation} in a $k$-$\LIN$ instance $\mcI = (\mcH, \{b_v\}_{v \in \mcH})$ is a collection $\mcV \subseteq \mcH$ and a choice of coefficients $\alpha_v$ for all $v \in \mcV$ such that
    \begin{equation*}
        \sum_{v \in \mcV} \alpha_v \cdot v = 0 \text{ and } \sum_{v \in \mcV} \alpha_v \cdot b_v \neq 0\mper
    \end{equation*}
    That is, a refutation is a linear combination of equations such that the coefficient vectors cancel out (the left-hand side is 0) but the right-hand side is non-zero.
\end{definition}

Informally, one can think of refutations as the type of contradiction proof that Sum-of-Squares looks for when trying to show a $k$-$\LIN$ instance is not satisfiable. Expansion then guarantees that long refutations require large degree, and as such Sum-of-Squares cannot ``see'' these contradictions, instead believing the instance is satisfiable. This view allows us to prove \cref{thm:mainsoslowerbound} fairly straightforwardly.

We now construct our pseudo-expectation $\pE$ given \cref{thm:inversefeige}. Applying \cref{lem:pseudoreduction} to the result of the following theorem completes \cref{thm:infsoslowerbound1}.

\begin{theorem}[Degree-$\widetilde{\Omega}(\ell)$ Sum-of-Squares lower bounds against $k$-$\LIN(\F)$ refutation]
    \label{thm:mainsoslowerbound}
    Fix $k \geq 3$ and $\frac{n}{\max\left(\abs{\F^*}, k\right)} \geq \ell \geq \max(k,\log\left(\abs{\F^*}n\right))$. Then with large probability a random $k$-$\LIN(\F)$ instance $\mcI$ with $\abs{\mcH} = \Theta(n) \cdot \left(\frac{n\abs{\F^*}}{\ell}\right)^{k/2-1} \cdot \varepsilon^{-2}$ has that:
    \begin{enumerate}
        \item \label{item:sos1} $\val(\mcI) \leq \frac{1}{\abs{\F}} + \varepsilon$.
        \item \label{item:sos2} There is a degree-$\widetilde{\Omega}(\ell)$ pseudo-expectation $\pE$ satisfying:
        \begin{enumerate}
            \item The constraints of \cref{def:fourierpseudo-expectation}.
            \item For all $v \in \mcH, \beta \in \F^*$, $\pE\sbra{\prod_{i\in\supp(v)} y_{i, \beta v_i}} = \omega_p^{\Tr(\beta b_v)}$.
        \end{enumerate}
    \end{enumerate}
\end{theorem}

\begin{proof}[Proof of \cref{thm:mainsoslowerbound} from \cref{thm:inversefeige}]

As a corollary of \cref{thm:inversefeige} with parameter $\beta = \frac{1}{\log(\abs{\F^*}n)}$ we get that a random $\mcH$ with $\abs{\mcH} \geq Cn \cdot \left(\frac{n \abs{\F^*}}{\ell}\right)^{k/2-1} \cdot \varepsilon^{-2}$ for $C > 0$ dependent only on $k$ has the property that all $\mcV \subseteq \mcH$ with $\abs{\mcV} \leq \ell \varepsilon^2$ and choice of coefficients $\{\alpha_v\}_{v \in \mcH}$ has $\left| \sum_{v \in \mcH} \alpha_v \cdot v \right| > \frac{\abs{\mcV}}{\log(\abs{\F^*}n)}$. This immediately implies there are no linearly dependent subsets of size $\ell\varepsilon^2$ or smaller. By choosing $\{b_v\}_{v \in \mcH}$ randomly and letting $\mcI = (\mcH, \{b_v\}_{v \in \mcH})$, \cref{item:sos1} follows from the simple probabilistic fact that semirandom $\mcH$ with $\abs{\mcH} \geq 2n \log \abs{\F} \cdot \varepsilon^{-2}$ equations are highly unsatisfiable (proved as a special case of \cref{lem:unsatisfiability}).

Now that we have an unsatisfiable instance $\mcI$, it suffices to prove \cref{item:sos2} by constructing a degree $d = \frac{\ell\varepsilon^2}{\log(\abs{\F^*}n)}$ pseudo-expectation satisfying all the specified constraints of $\mcI$. To aid in describing our construction, we introduce the following definitions simplifying redundancy in our polynomials.

\begin{definition}[Representative polynomials and vectors]
    Let $P$ be a complex monomial in indeterminates $y_{i, \alpha}$ for $i \in [n]$ and $\alpha \in \F^*$. For a fixed $i \in [n]$, we define the total coefficient $\beta_P(i) \in \F$ as follows. For each $\alpha \in \F^*$, let $a_\alpha \in \N$ be the exponent in which $y_{i, \alpha}$ appears in $P$. Then $\beta_P(i) = \sum_{\alpha \in \F^*} a_\alpha \cdot \alpha$.
    As before $a \cdot \alpha$ for $a \in \N$ denotes adding $a$ copies of $\alpha$. The representative monomial $\mathcal{E}(P)$ for $P$ is then given by the product
    \begin{equation*}
        \mathcal{E}(P) = \prod_{i =1}^n y_{i, \beta_P(i)} \mcom
    \end{equation*}
    where by convention $y_{i, 0} = 1$. The representative polynomial is defined analogously as the sum of corresponding representative monomials. Finally, there is a one-to-one correspondence between representative monomials and vectors in $\F^n$ given by taking the $i$th vector entry to be $\beta_P(i)$.
\end{definition}

    The motivation for defining representative polynomials is as follows. Suppose we have distinct degree $\leq d$ polynomials $P$ and $Q$ with equivalent representative polynomials. Notice that the consistency constraint requires that $P$ and $Q$ are equivalent under any pseudo-expectation. As such, any pseudo-expectation we define respecting the constraints should set $\pE[P] = \pE[Q]$. Accordingly, we define our $\pE$ on representative monomials and use this equivalence plus linearity to extend to all polynomials.

\begin{notation}
    Our pseudo-expectation is defined on monomials $y_u$ for $u \in \F^{\abs{\F^*}n}$ as $y_u = \prod_{(i, \alpha) \in [n] \times \F^*} y_{i, \alpha}^{u_{i, \alpha}}$. Yet, per above, any consistent pseudo-expectation has that all monomials with the same representative monomial have the same value. We abuse notation and write for representative vector $v \in \F^n$, $\pE[y_v] = c$ to mean set $\pE$ for all monomials with representative vector $v$ to $c$.
\end{notation}

We now state our pseudo-expectation $\pE$. Our construction is algorithmic and is based on the principle of max-entropy: given a set of hard constraints (e.g. $k$-$\LIN(\F)$ equations), we set only these constraints (and those derivable from them) and let all others be maximally undetermined.

\begin{tcolorbox}[
    width=\textwidth,   
    colframe=black,  
    colback=white,   
    title=$\widetilde{\Omega}(\ell)$-degree max-entropy pseudo-expectation for $\mcI$,
    colbacktitle=white, 
    coltitle=black,      
    fonttitle=\bfseries,
    center title,   
    enhanced,       
    frame hidden,           
    borderline={1pt}{0pt}{black},
    sharp corners,
    toptitle=2.5mm,
    label=disp:maxentropy
]
\textbf{Input:} A $k$-$\LIN(\F)$ instance $\mcI = (\mcH, \{b_v\}_{v \in \mcH})$.\\

\textbf{Output:} A valid degree $d$ pseudo-expectation $\pE$ respecting the constraints in \cref{item:sos2} where $d = \frac{\ell \varepsilon^2}{2\log(\abs{\F^*}n)}$.\\

\textbf{Algorithm:}
\begin{enumerate}
    \item Let $\pE[1] = \pE[y_0] = 1$.
    \item For every $v \in \mcH, \beta \in \F^*$ set $\pE[y_{\beta v}] = \omega_p^{\Tr(\beta b_v)}$.
    \item Repeat the following until no progress can be made:
    \begin{enumerate}
        \item Choose $U$ and $V \in \F^n$ with $\abs{U - V} \leq d$ and $\pE[y_U]$ and $\pE[y_V] \neq 0$.
        \item Set $\pE[y_{U-V}] = \pE[y_U] \overline{\pE[y_V]}$.
        \item If the value was previously something else, throw an error.
    \end{enumerate}
    \item Set any remaining variable degree $d$ monomials to 0 under the pseudo-expectation.
\end{enumerate}

\end{tcolorbox}

\begin{remark}
    Note that after enforcing linearity, $\pE$ is fully defined by its definition on degree $d$ monomials.
\end{remark}

A priori, it is unclear the max-entropy pseudo-expectation is indeed a pseudo-expectation or even that it is well-defined. Nonetheless, we show that under our expansion assumption both properties hold 

\begin{lemma}
    \label{lem:termination}
    If $\mcI = (\mcH, \{b_v\}_{v \in \mcH})$ where $\mcH$ has for all $\mcV \subseteq \mcH$ with $\abs{\mcV} \leq \ell \varepsilon^2$ and choice of coefficients $\{\alpha_v\}_{v \in \mcV}$ has $\left|\sum_{v \in \mcV} \alpha_v \cdot v\right| > \frac{\abs{\mcV}}{\log(\abs{\F^*}n)}$ then the max-entropy pseudo-expectation $\pE$ is well-defined.
\end{lemma}

\begin{lemma}
    \label{lem:constraints}
    Given the max-entropy pseudo-expectation $\pE$ is well-defined, $\pE$ satisfies the constraints of \cref{item:sos2}.
\end{lemma}

Putting \cref{lem:termination} and \cref{lem:constraints} together immediately yields that the max-entropy $\pE$ satisfies \cref{item:sos2} for random instances.
\end{proof}
    
To make sense of the algorithm above and the proceeding proof, we introduce the notion of derivation, which can be thought of as a transcript for a run of the algorithm.

\begin{definition}[Derivations]
    A degree $d$ derivation of $W \in \F^n$ in a set $\mcH$ of vectors in $\F^n$ is a sequence of vectors $w_1, ...., w_t$ with the properties:
    \begin{itemize}
        \item For all $k \in [t]$, $\abs{w_k} \leq d$.
        \item $w_t = W$.
        \item For all $k \in [t]$ either $\exists i, j < k$ s.t. $w_k = w_i - w_j$ OR $w_k \in \mcH$.
    \end{itemize}
\end{definition}

\begin{observation}
    Since the only derived elements at the start are vectors in $\mcH$, we inductively guarantee all derivable vectors are linear combinations of vectors in $\mcH$. That is, for a derivable vector $v$ we can write $\sum_{i=1}^t \beta_i \cdot v_i$ for some sequence of $\beta_i \in \F^*$ and $v_i \in \mcH$ for some $t \in \N$. We call $t$ the length of the derivation.
\end{observation}

The construction can be thought as greedily $d$-deriving all possible vectors from the original vectors in $\mcH$. To prove \cref{item:sos2} we show that given $\mcI$ satisfies the expansion property from \cref{thm:inversefeige} the algorithm is guaranteed to terminate, and given termination we can guarantee $\pE$ satisfies the constraints.

\begin{proof}[Proof of \cref{lem:termination}]
Note so long as the algorithm never errors, the max-entropy pseudo-expectation $\pE$ is well-defined, so suppose the algorithm errors. Then we have some choice $W = U - V$ with weight $\leq d$ such that $\pE[y_{U -V}]$ is set to two different things. These correspond to equal derivations with unequal right-hand sides, that is
\begin{equation*}
    W = \sum_{u \in \mcU} \beta_u \cdot u = \sum_{v \in \mcV} \beta_v \cdot v \text{ while } \sum_{u \in \mcU} \beta_u \cdot b_u \neq \sum_{v \in \mcV} \beta_v \cdot b_v\mcom
\end{equation*}
for some $\mcU, \mcV \subseteq \mcH$ and coefficient sequences $\cbra{\beta_u}_{u \in \mcU}$ and $\cbra{\beta_v}_{v \in \mcV}$ on $\F^*$. Write $\mcU \cup \mcV$ to denote combining $\mcU$ and $\mcV$ with the sequences $\{\beta_u\}_{u \in \mcU}$ and $\{-\beta_v\}_{v \in \mcV}$ and merging like vectors. This yields a refutation, $\sum_{w \in \mcU \cup \mcV} \beta_w \cdot w = 0$ while $\sum_{w \in \mcU \cup \mcV} \beta_w \cdot b_w \neq 0$. Since there are no linearly dependent subsets of size $\leq \ell \varepsilon^2$, we further have that $\abs{\mcU \cup \mcV} > \ell \varepsilon^2$. We then observe the following.

\begin{observation}
    In the derivation of $\mcU$ there exists a vector $w$ with $\abs{w} > \frac{\ell \varepsilon^2}{2\log(\abs{\F^*}n)}$.
\end{observation}

To observe, let $w_k$ be the first vector in the derivation with length larger than $\ell \varepsilon^2$. Note $w^*$ is guaranteed to exist since at the very least $0$ does. By definition (and assuming $\ell \varepsilon^2 > k$), $w^* = w_i - w_j$ for some $i, j < k$. The length of the derivation for $w_k$ is at most the sums of the lengths for $w_i$ and $w_j$, so by pigeonhole principle one has length in the interval $[\ell \varepsilon^2/2, \ell \varepsilon^2]$, without loss of generality let it be $w_i$. By the lemma assumption the expansion property applies, meaning $\abs{w_i} > \frac{\ell \varepsilon^2}{2\log(\abs{\F^*}n)}$. We finish by noticing this contradicts the fact that the algorithm only considers $d$-derivations.
\end{proof}

\begin{proof}[Proof of \cref{lem:constraints}]

Note that satisfiability is immediately satisfied by construction, given the algorithm terminates. We break the remainder of the proof into showing the max-entropy pseudo-expectation $\pE$ satisfies each constraint of \cref{item:sos2}.

\end{proof}

\begin{lemma}[Validity]
    \label{clm:validity}
    For any $y_{i, \alpha}$ and polynomial $f$ of variable degree at most $d-1$, $\pE\sbra{f \cdot (y_{i,\alpha}^p -1)} = 0$.
\end{lemma}

\begin{proof}[Proof of \cref{clm:validity}]
    Let $f$ be a complex polynomial in indeterminates $\cbra{y_{i, \alpha}}_{i \in [n], \alpha \in \F^*}$ with $\vardeg(y_{i, \alpha}^p f) \leq d$. By linearity $\pE\sbra{(y_{i, \alpha}^p-1)f} = \pE\sbra{y_{i, \alpha}^pf} - \pE\sbra{f}$, so it suffices to argue $\pE\sbra{y_{i, \alpha}^pf} = \pE\sbra{f}$. This follows simply from the fact that (1) the representative polynomial of $y_{i, \alpha}^pf$ is the same as that of $f$ since $y_{i, \alpha}^p$ simply adds $p\alpha \equiv 0$ to the corresponding entry in the representative vector and (2) both $f$ and $y_{i, \alpha}^p f$ have variable degree $d$ or less by assumption and submultiplicativity of variable degree, which means the algorithm sets both the same.
\end{proof}

\begin{lemma}[Consistency]
    \label{clm:consistency}
    For any $y_{i, \alpha}$ and $y_{i,\beta}$ and polynomial $f$ of variable degree at most $d-2$, $\pE\sbra{f \cdot (y_{i,\alpha} \cdot \overline{y_{i,\beta}} - y_{i, \alpha - \beta})} = 0$.
\end{lemma}

\begin{proof}[Proof of \cref{clm:consistency}]
    Let $f$ be a complex polynomial in indeterminates $\cbra{y_{i, \alpha}}_{i \in [n], \alpha \in \F^*}$ with $\vardeg((y_{i, \alpha} \cdot \overline{y_{i, \beta}} - y_{i, \alpha - \beta}) f) \leq d$. The argument follows the same as in \cref{clm:validity}, just noting the representative vectors of $y_{i,\alpha}\cdot \overline{y_{i,\beta}} \cdot f$ and $y_{i, \alpha -\beta} \cdot f$ must be the same so the algorithm sets them equivalent.
\end{proof}

\begin{lemma}[Positivity]
    \label{lem:positivity}
    For any polynomial $f$ of variable degree at most $d/2$, $\pE[\abs{f}^2] \geq 0$.
\end{lemma}

\begin{proof}[Proof of \cref{lem:positivity}]
    Fix $f$ a complex polynomial in indeterminates $\cbra{y_{i, \alpha}}_{i \in [n], \alpha \in \F^*}$ with $\vardeg(f) \leq d/2$. We write the standard Fourier decomposition $f = \sum_{u \in \F^{\abs{\F^*}n}} \hat{f}(u) \cdot y_u$ for Fourier coefficients $\cbra{\hat{f}(u)}_{u \in \F^{\abs{\F^*}n}}$ in $\C$. We relate certain monomials in this decomposition by the following equivalence class.

    \begin{definition}
        We say two vectors $u, v \in \F^{\abs{\F^*}n}$ with $\vardeg(y_u), \vardeg(y_v) \leq d/2$ are related if $\pE[y_{u - v}] \neq 0$ in the algorithm above.
    \end{definition}

    \begin{claim}
        The above relation is an equivalence relation inducing equivalence classes $\{\mathcal{F}_i\}_{i \in [r]}$ for some $r \in \N$.
    \end{claim}

    \begin{proof}
        Reflexivity follows since $\pE[y_{u -u}] = \pE[y_0] = \pE[1] = 1$ for any $u \in \F^n$. For symmetry, note that if $u \sim v$ then $\pE[y_{u-v}] \neq 0$, so it must have been set in the algorithm. If it is set non-zero in Step 2 for $\beta \in \F^*$, then it will also be set for $-\beta \in \F^*$. If it is set in Step 3 then $\pE[y_{u-v}] = \pE[y_U] \overline{\pE[y_V]}$ for $U, V \in \F^n$ and $\abs{U - V} \leq d$. Then $\pE[y_{v- u}] = \pE[y_V]\overline{\pE\sbra{y_U}}$ since $\abs{V-U} = \abs{U-V}$. This is non-zero assuming $\pE[y_{u-v}]$ is.

        For the hard step assume $u \sim v$ and $v \sim w$ and let $U$, $V$, and $W \in \F^n$ be the representative vectors for $u, v$, and $w$ respectively. By the variable degree assumption we have $\abs{U}, \abs{V}, \abs{W} \leq d/2$, and by triangle inequality $\abs{U-V} \leq d$, $\abs{V-W} \leq d$, and $\abs{U-W} \leq d$. By symmetry $\pE[y_{W-V}] \neq 0$. Then $U-W = (U-V) - (W - V)$ and $\pE[y_{U-V}] \neq 0$ and $\pE[y_{W-V}] \neq 0$ and thus $\pE[y_{u-w}] = \pE[y_{U-W}] = \pE[y_{U-V}] \overline{\pE[y_{W-V}]} \neq 0$ by Step 3 of the algorithm. So $u \sim w$ as desired.
    \end{proof}

    The main point of defining these equivalence classes is now we can write $f = \sum_{i \in [r]} \sum_{u \in \mathcal{F}_i} \hat{f}(u) \cdot y_u := \sum_{i \in [r]}f_i$ and note that
    \begin{equation*}
        \pE\sbra{\abs{f}^2} = \pE\left[\sum_{i \in [r]} f_i \overline{\sum_{j \in [r]} f_j}\right] = \pE\left[\sum_{i, j \in [r]} f_i\bar{f}_j\right] =  \sum_{i \in [r]} \pE\sbra{\abs{f_i}^2}\mper
    \end{equation*}
    This last step follows from the definition of the $\mathcal{F}_i$, that for any $v_i \in \mathcal{F}_i$ and $v_j \in \mathcal{F}_j$ separate we must have $\pE\sbra{y_{v_i-v_j}} = 0$. It suffices then to prove positivity for each equivalence class. Fix $i \in [r]$. Since $f_i = \sum_{u \in \mathcal{F}_i} \hat{f}(u) \cdot y_u$, we write:
    \begin{align*}
        \pE\sbra{\abs{f_i}^2} &= \pE\sbra{\sum_{u, v \in \mathcal{F}_i} \hat{f}(u) \cdot \overline{\hat{f}(v)} \cdot y_u \cdot \overline{y_v}}\\
        &= \sum_{u, v \in \mathcal{F}_i} \hat{f}(u) \cdot \overline{\hat{f}(v)} \cdot \pE\sbra{y_u \cdot \overline{y_v}}\\
        &= \sum_{u, v \in \mathcal{F}_i} \hat{f}(u) \cdot \overline{\hat{f}(v)} \cdot \pE\sbra{y_{u-v}}\mper
    \end{align*}

    By the assumption $\vardeg(f) \leq d/2$, we have that the representative vectors of $u$ and $v$ have weight $\leq d/2$, so the weight of the representative vector for $u-v$ is $\leq d$ by triangle inequality. Now we choose an arbitrary $w \in \mathcal{F}_i$. By the same logic, representatives for $u-w$ and $w-v$ have weight $\leq d$, thus the algorithm attempts to set $\pE\sbra{y_{u-v}} = \pE\sbra{y_{u-w}} \pE\sbra{y_{w-v}}$ and must succeed since we terminate without error. We may then write
    \begin{align*}
        \pE\sbra{\abs{f_i}^2} &= \sum_{u, v \in \mathcal{F}_i} \hat{f}(u) \cdot \overline{\hat{f}(v)} \cdot \pE\sbra{\chi_{u-v}}\\ 
        &= \sum_{u, v \in \mathcal{F}_i} \hat{f}(u) \cdot \overline{\hat{f}(v)} \cdot \pE\sbra{\chi_{u-w}} \pE\sbra{\chi_{w-v}}\\
        &= \left(\sum_{u \in \mathcal{F}_i} \hat{f}(u) \cdot \pE\sbra{\chi_{u-w}}\right) \overline{\left(\sum_{u \in \mathcal{F}_i} \hat{f}(u) \cdot \pE\sbra{\chi_{u-w}}\right)}\\
        &\geq 0\mper
    \end{align*}
\end{proof}
\section{Simple Refutation Algorithms for \texorpdfstring{$k$-$\LIN$}{k-LIN}: Proof of \cref{thm:trivialalgorithm}}
\label{sec:trivialalgorithm}

In this section, we show a simple refutation algorithm for $k$-$\LIN(\F)$ that outperforms \cref{thm:refutation} in the ``high-degree'' regime $\ell \geq n/\abs{\F^*}^{1 - 2/k}$ (although the algorithm works for any $\ell \geq k$). The idea is to consider, for each set $S$ of size $\ell$ the set of constraints whose supports are contained within $S$, and compute by brute force the optimal local assignment. The full algorithm is easy to state so we specify it below.
\begin{tcolorbox}[
    width=\textwidth,   
    colframe=black,  
    colback=white,   
    title=Simple $k$-$\LIN(\F)$ Refutation Algorithm,
    colbacktitle=white, 
    coltitle=black,      
    fonttitle=\bfseries,
    center title,   
    enhanced,       
    frame hidden,           
    borderline={1pt}{0pt}{black},
    sharp corners,
    toptitle=2.5mm
]
\textbf{Input:} A $k$-$\LIN(\F)$ instance $\mcI = (\mcH, \{b_v\}_{v \in \mcH})$.\\

\textbf{Output:} $\algval(\mcI) \in [0,1]$ with guarantee $\algval(\mcI) \geq \max_{x \in \F^n} \val(\mcI, x)$.\\

\textbf{Algorithm:}
\begin{enumerate}
    \item For all $S \subseteq [n]$ with $\abs{S} = \ell$, find all $v \in \mcH$ with $\supp(v) \subseteq S$, let the induced subinstance on these constraints and the variables in $S$ be $\mcI_S = (\mcH_S, \{b_v\}_{v \in \mcH_S})$.
    \item Let $\algval(\mcI_S) = \max_{x_S \in \F^{S}} \val(\mcI_S, x_S)$.
    \item Output $\algval(\mcI) = \frac{1}{\abs{\mcH}{n \choose \ell-k}}\sum_{\substack{S \subseteq [n]\\ \abs{S} = \ell}} \algval(\mcI_S) \cdot \abs{\mcH_S}$.
\end{enumerate}

\end{tcolorbox}

\begin{remark}
    Notice this algorithm trivially generalizes to finite Abelian groups by replacing considerations of $\F$ with $G$. For simplicity, we focus on $\F$ here.
\end{remark}

The rest of this section is devoted to analyzing the above algorithm, in both the random and semirandom settings.

\begin{proof}[Proof of \cref{thm:trivialalgorithm}]
    
Note first that the algorithm runs in time $(\abs{\F}n)^{O(\ell)}$ since we consider all subsets of size $\ell$ and all $\abs{\F}^\ell$ assignments for each.

\begin{proof}[Proof of \cref{item:simplerefutation1}]
Suppose $x \in \F^n$ is an optimal assignment such that $\val(\mcI, x) = \val(\mcI)$. Write $v(x)$ to denote whether $x$ satisfies $v$ or not. Observe that we may write
\begin{equation*}
    \val(\mcI, x) 
    = \frac{1}{\abs{\mcH}}\sum_{v \in \mcH} v(x) 
    = \frac{1}{\abs{\mcH}{n \choose \ell - k}}\sum_{\substack{S \subseteq [n]\\ \abs{S} = \ell}} \sum_{\substack{v \in \mcH\\\supp(v) \subseteq S}} v(x) 
    = \frac{1}{\abs{\mcH}{n \choose \ell -k}}\sum_{\substack{S \subseteq [n]\\ \abs{S} = \ell}} \val(\mcI_S, x_S) \cdot \abs{\mcH_S}\mper
\end{equation*}
The second equality follows since for each $v \in \mcH$, $v(x)$ appears exactly ${n \choose \ell-k}$ times for each $S \subseteq [n]$, $\abs{S} = \ell$ with $\supp(v) \subseteq S$. From here we use $\val(\mcI_S, x_S) \leq \val(\mcI_S)$ for all $S$ to conclude $\val(\mcI) \leq \algval(\mcI)$.
\end{proof}

\begin{proof}[Proof of \cref{item:simplerefutation2}]

We fix $\mcI$ to be fully random. To show $\val(\mcI) \leq \frac{1}{\abs{\F}} + \varepsilon$ it suffices to show $\val(\mcI_S) \leq \frac{1}{\abs{\F}} + \varepsilon$ for all possible $S$. We do this in two steps: (1) we show given random equations $v \in \mcH$ at our equation threshold, the number of equations in each $\mcH_S$ is large enough and (2) each $\mcI_S$, given it is large enough, has value at most $\frac{1}{\abs{\F}} + \varepsilon$.

\begin{lemma}
    \label{lem:equationuniform}
    Let $\mcI = (\mcH, \{b_v\}_{v \in \mcH})$ be a random $k$-$\LIN(\F)$ instance in $n$ variables. For any $S \subseteq [n]$, $\abs{\mcH_S} \geq \frac{C}{2} \ell \log n \cdot \varepsilon^{-2}$ when $\abs{\mcH} \geq c^{-k} C \cdot n \left(\frac{n}{\ell}\right)^{k-1} \log n \cdot \varepsilon^{-2}$ with probability at least $1 - \frac{1}{n^{2\ell}}$ for some constants $c,C > 0$.
\end{lemma}

First, given the size lower bound of \cref{item:simplerefutation2}, we can apply \cref{lem:equationuniform} to every subset $S \subseteq [n]$ with $\abs{S} = \ell$ and union bound across ${[n] \choose \ell}$ so that all $\abs{\mcH_S} \geq \frac{C}{2} \ell \log n \cdot \varepsilon^{-2}$. This error probability is $n^{2\ell}$ per the lemma. Now for each we can bound the probability $\val(\mcI_S) > \frac{1}{\abs{\F}} + \varepsilon$.
    
\begin{lemma}
    \label{lem:unsatisfiability}
    Let $\mathcal{J} = (\mcH, \{b_v\}_{v \in \mcH})$ be a semirandom $k$-$\LIN(\F)$ instance in $\ell$ variables. Then $\val(\mathcal{J}) \leq \frac{1}{\abs{\F}} + \varepsilon$ when $2\ell\log (\abs{\F}n) \cdot \varepsilon^{-2}$ with probability at least $1 - \frac{1}{n^{2\ell}}$.
\end{lemma}

 By applying \cref{lem:unsatisfiability} to each $\mcI_S$, we again union bound across ${[n] \choose \ell}$ to get $\val(\mcI_S)\leq\frac{1}{\abs{\F}} + \varepsilon$ for all $S$, which concludes the proof.
\end{proof}

\begin{proof}[Proof of \cref{lem:equationuniform}]
    For random $v \in \mcH$, observe that $\supp(v) \subseteq S$ with probability $\frac{{n \choose \ell-k}}{{n \choose \ell}} \geq c^k\left(\frac{\ell}{n}\right)^k$ for some constant $c > 0$ by \cref{fact:binomest}. We can model the size of $\mcH_S$ as the sum of $\abs{\mcH}$ independent \textsf{Bernoulli}$\left(c^k\left(\frac{\ell}{n}\right)^k\right)$ random variables. The expected value is then at least $c^k\left(\frac{\ell}{n}\right)^k \abs{\mcH}$ and by standard Chernoff bounds
    \begin{equation*}
        \Pr\sbra{\abs{\mcH_S} \leq \frac{c^k}{2}\left(\frac{\ell}{n}\right)^k \abs{\mcH}} \leq e^{c^k\left(\frac{\ell}{n}\right)^k \abs{\mcH}/8} \mper
    \end{equation*}
    Taking $\abs{\mcH} \geq 16 c^{-k} \cdot n \left(\frac{n}{\ell}\right)^{k-1} \log n \cdot \varepsilon^{-2}$ finishes the proof.
\end{proof}

\begin{proof}[Proof of \cref{lem:unsatisfiability}]
Fix an assignment $x \in \F^\ell$ of which there are $\abs{\F}^\ell$. We show the probability of this assignment satisfying more than a $\frac{1}{\abs{\F}}+\varepsilon$ fraction of the constraints is small and union bound over all $\abs{\F}^\ell$ assignments. Note that each equation is satisfied with probability $\frac{1}{\abs{\F}}$. To see this, we can completely ignore the left-hand side of the equation and just know that it takes some value under the assignment, and that the right-hand sides match that value with probability $\frac{1}{\abs{\F}}$. It is easy to see this way whether a equation is satisfied is independent of other equations since the right-hand sides are drawn independently.

We model $\val(\mathcal{J},x) \cdot \abs{\mcH}$ as a sum of $\abs{\mcH}$ independent \textsf{Bernoulli}$\left(\frac{1}{\abs{\F}}\right)$ random variables. By Hoeffding's inequality
\begin{equation*}
    \Pr\left[\val(\mathcal{J},x) \cdot \abs{\mcH} \geq \frac{1}{\abs{\F}}\abs{\mcH} + \varepsilon \abs{\mcH}\right] \leq e^{-\varepsilon^2 \abs{\mcH}}\mper
\end{equation*}
Taking $\abs{\mcH} \geq 2\ell\log (\abs{\F}n) \cdot \varepsilon^{-2}$ implies the probability is at most $\frac{1}{(\abs{\F}n)^{2\ell}}$, and by union bound we conclude the $\val(\mathcal{J}) \leq \frac{1}{\abs{\F}} + \varepsilon$.
\end{proof}

We now explain how the same algorithm succeeds in the semirandom setting, albeit at a slightly higher equation density.

\begin{proof}[Proof of \cref{item:simplerefutation3}]
    Doing a post-mortem of the proof of \cref{item:simplerefutation2}, the apparent problem in the semirandom setting is that we do not have \cref{lem:equationuniform}, which relies on the randomness of the equation structure. We used this to claim each subset $S$ contains many equations, allowing us to apply \cref{lem:unsatisfiability}. Instead, we simply assume for any $S$ with $\abs{\mcH_S} < \frac{\varepsilon}{2C}\abs{\mcH}$ for some $C > 0$, we have $\algval(\mcI_S) = 1$. This is fine as we have the contribution from such a $\mcI_S$ is $\frac{1}{{n \choose \ell-k}} \frac{\varepsilon}{2C}$. Since there are at most ${n \choose \ell}$ choices for $S$, we would like the contribution to be $\frac{1}{{n \choose \ell}} \frac{\varepsilon}{2}$, which is true if $C = \frac{{n \choose \ell}}{{n \choose \ell-k}}$ which is $\Theta(1)$ by \cref{fact:binomest}.

    Now for any term with $\abs{\mcH_S} \geq \frac{\varepsilon}{2C}\abs{\mcH}$, we can apply \cref{lem:unsatisfiability} as before, which requires $\abs{\mcH}$ is a factor $\varepsilon^{-1}$ and constants larger.
\end{proof}
\end{proof}

\newpage

\bibliographystyle{alpha}
\bibliography{references.bib}

\newpage

\appendix

\section{Random Collections of $k$-Sparse Vectors Expand Well}
\label{append:expansion}

\begin{theorem}[Expansion in random $k$-sparse vector sets (\cref{thm:inversefeige} restated)]
    \label{thm:inversefeigerestated}
     Fix $k/2 \leq \ell \leq n/(\max(\abs{\F^*}, k))$. Let $\mcH$ be a set of $\abs{\mcH} \leq \Omega(\delta n) \cdot  \left(\frac{n\abs{\F^*}}{\ell}\right)^{k/2-1-\beta}$ uniformly random $k$-sparse vectors in $\F^n$. Then with probability $\geq 1-\delta$ there do not exist any $\mcV \subseteq \mathcal{H}$ with $\abs{\mcV} \leq \ell$ and coefficients $\cbra{\alpha_v}_{v \in \mcV}$ in $\F^*$ such that
     \begin{equation*}
         \left|\sum_{v \in \mcV} \alpha_v \cdot v\right| \leq \beta \abs{\mcV}\mper
     \end{equation*}
\end{theorem}

\begin{proof}[Proof of \cref{thm:inversefeigerestated}]

For a collection of $m$ vectors $\mcH$ we can define $\mathcal{R}_d = \{(\mcV, \{\alpha_v\}_{v \in \mcV}) \mid \mcV \subseteq \mcH,\abs{\mcV} = d, \alpha \in (\F^*)^d\}$ to be the set of all length exactly $d$ linear combinations in $\mcH$. Note that we have $\abs{\mathcal{R}}_d = {\abs{\mcH} \choose d}\abs{\F^*}^d \leq \left(\abs{\mcH} \cdot \frac{e\abs{\F^*}}{d}\right)^d$. We show that none have low weight by the following lemma.

\begin{lemma}
    \label{lem:largerefutations}
    For any $(\mcV, \{\alpha_v\}_{v \in \mcV}) \in \mathcal{R}_d$ and $d \leq \ell$ we have
    \begin{equation*}
    \Pr\left[\left|\sum_{v \in \mcV} \alpha_v \cdot v\right| \leq \beta d \right] \leq \left((ek)^{2k} \cdot \left(\frac{d}{n\abs{\F^*}}\right)^{k/2-\beta/2}\right)^d\mper
 \end{equation*}
\end{lemma}

Using \cref{lem:largerefutations} we can simply union bound over all of $\mathcal{V}$ in $\mathcal{R}_d$.
\begin{align*}
    \Pr\left[\exists (\mcV, \{\alpha_v\}_{v \in \mcV}) \in \mathcal{R}_d \text{ such that } \left|\sum_{v \in \mcV} \alpha_v \cdot v\right| \leq \beta d \right]  
    &\leq \left(\abs{\mcH} \cdot \frac{e\abs{\F^*}}{d}\right)^d \cdot \left((ek)^{2k} \cdot \left(\frac{d}{n\abs{\F^*}}\right)^{k/2-\beta/2}\right)^d\\
    &\leq \left(\abs{\mcH} \cdot (ek)^{2k+1} \cdot n \cdot \left(\frac{d}{n\abs{\F^*}}\right)^{k/2-1-\beta}\right)^d\\
    &\leq \delta^d\mper
\end{align*}
For the last line we assume $\abs{\mcH} \leq \frac{\delta}{(ek)^{2k+1}} \cdot n \cdot \left(\frac{n \abs{\F^*}}{\ell}\right)^{k/2-1-\beta}$. We complete the proof of \cref{thm:inversefeigerestated} by union bounding for all $d \leq \ell$, which is $O(\delta)$ by convergence of geometric series.
\end{proof}

\begin{proof}[Proof of \cref{lem:largerefutations}]

Without loss of generality we can assume $\alpha = \onevec$ and $\mcV$ is a set of random $k$-sparse vectors $v_i \in \F^n$ with $\abs{\mcV} = d$. Let $v = \sum_{j=1}^d v_j$. It is useful to identify the set $\{v_j\}_{j \in [d]}$ by a $[n] \times [d]$ grid where coordinate $\mcV[i, j]$ corresponds to $(v_j)_i$. Note such a grid always has $kd$ non-zero coordinates and often we refer to the $i$th coordinate to simply mean the $i$th non-zero coordinate under a ``column-first'' ordering (i.e. the relation induced by $(i_1, j_1) < (i_2, j_2)$ if $j_1 < j_2$ first and $i_1 < i_2$ if $j_1 = j_2$). This induces a unique ``function representation'' of $\mcV$ as $\phi_{\mcV} : [kd] \to [n] \times \F^*$, mapping the $i$th non-zero coordinate to its row number and value. We can partition the $kd$ non-zero coordinates in this grid into the following groups:
\begin{align*}
    &\mcV^{(-1)} = \cbra{(i, j) \in [n] \times [d] \mid \mcV[i, j] \neq 0, \sum_{j' = 0}^{j-1} \mcV[i, j'] \neq 0 \text{ and } \mcV[i,j] = -\sum_{j' = 0}^{j-1} \mcV[i, j']}\\
    &\mcV^{(0)} = \cbra{(i, j) \in [n] \times [d] \mid \mcV[i, j] \neq 0, \sum_{j' = 0}^{j-1} \mcV[i, j'] \neq 0 \text{ and } \mcV[i,j] \neq -\sum_{j' = 0}^{j-1} \mcV[i, j']}\\
    &\mcV^{(+1)} = \cbra{(i, j) \in [n] \times [d] \mid \mcV[i, j] \neq 0, \sum_{j' = 0}^{j-1} \mcV[i, j'] = 0}\mper
\end{align*}

These groups have a simple semantic characterization when viewed as follows. Suppose we sampled each non-zero coordinate one at a time, specifying what index in $[n]$ it sits in and what value in $\F^*$ it takes. At the end, we add up all values from the same index $i$, and either get a non-zero value, which contributes $+1$ to the Hamming weight, or $0$, which contributes nothing. This partition can be seen as a way to assign ``blame'' for a specific coordinate contributing $+1$. The coordinates in $\mcV^{(+1)}$ cause the partial sum at the point it is sampled to be $+1$. The coordinates in $\mcV^{(-1)}$ can be seen as those that reduce the contribution from $+1$ back to $0$ by cancelling what comes before. $\mcV^{(0)}$ leaves the contribution $+1$.

The motivation for this characterization and partition is the following observation.

\begin{observation}
    Given a set of $k$-sparse vectors $\mcV$ with $\abs{\mcV}= d$ we have that $\wt(v) = kd - 2\abs{\mcV^{(-1)}} - \abs{\mcV^{(0)}}$ for $v = \sum_{v_j \in \mcV} v_j$.
\end{observation}

To see this note that $\abs{\mcV^{(+1)}} = kd - \abs{\mcV^{(-1)}} - \abs{\mcV^{(0)}}$. As remarked above, these coordinates cause the corresponding contribution to the Hamming weight to be $+1$, and this is only changed if a later coordinate cancels. These coordinates then fall into $\mcV^{(-1)}$, meaning the total Hamming weight is $\abs{\mcV^{(+1)}}-\abs{\mcV^{(-1)}}$ or $kd - 2\abs{\mcV^{(-1)}} - \abs{\mcV^{(0)}}$.

From this we derive that $2\abs{\mcV^{(-1)}} + \abs{\mcV^{(0)}} < kd - \beta d$ is a sufficient condition for $\wt(v) > \beta d$, so we show the probability of this event is high.

\begin{definition}[Templates]
    For a set of coordinates $[kd]$ we define a \textit{template} to be a labelling of one of $\{-1, 0, +1\}$ to each coordinate. We parameterize a $(t_1, t_2)$-template to be a template which assigns the label ``$-1$'' $t_1$ times and ``$0$'' $t_2$ times. We further call a template a $t$-template if $2t_1 + t_2 = t$.
\end{definition}

\begin{definition}[Template-satisfying assignments]
    We say a template is satisfied by an assignment $\phi : [kd] \to [n] \times \F^*$ if all the labels are consistent with the above definitions under the assignment. Formally, this means that if labelled ``$b$'', then the coordinate falls in $\mcV^{(b)}$. Recall the representative assignment $\phi_{\mcV}$ for a set of $k$-weight vectors $\mcV$ given by taking the indices and values in all non-zero coordinates.
\end{definition}

These definitions are then useful in light of the following fact.

\begin{observation}
    \label{obs:templates}
    For $\mcV = \{v_j\}_{j \in [d]}$, there is a $(\abs{\mcV^{(-1)}}, \abs{\mcV^{(0)}})$-template that is satisfied by the representative assignment $\phi_\mcV$.
\end{observation}

To see this, just let the template be the actual behavior of each coordinate. The labelling corresponds to its partition $\mcV^{(b)}$ and its associated coordinate indeed corresponds to the most recent coordinate that makes the Hamming weight contribution $+1$. 

\begin{remark}
    The above definitions may seem a bit redundant, but the correspondence given in \cref{obs:templates} can concretely be seen as a generalization of the patterns that can appear when looking at a low weight sum of vectors. This will allow us to easily count ways this can happen and show that the probability any occurs is very low.
\end{remark}

Now we are set to prove $2\abs{\mcV^{(-1)}} + \abs{\mcV^{(0)}} < kd - \beta d$ by showing that $\phi_\mcV$ does not satisfy any $\geq kd-\beta d$-template with high enough probability, which can be accomplished through the following lemma.
\begin{lemma}
    \label{lem:templatebound}
    Let $\tau$ be a $t$-template on $[kd]$ coordinates assigned $[n] \times \F^*$, then
    \begin{equation*}
        \Pr\left[\tau \text{ satisfied  by } \phi_\mcV\right] \leq k^{2kd} \cdot \left(\frac{d}{n\abs{\F^*}}\right)^{t/2}\mper
    \end{equation*}
\end{lemma}

We also have the following bound on the total number of templates on $[kd]$ coordinates, which we use as a crude approximation to the number of $t$-templates for $t > kd-\beta d$.
\begin{observation}
    \label{fact:templatenum}
    There are less than $3^{kd}$ \textit{total} templates on $[kd]$ coordinates.
\end{observation}

Putting the two together, we can simply union bound across large templates and get
\begin{equation*}
    \Pr\left[\text{any $\geq (kd-\beta d)$-template satisfied}\right] \leq \left((ek)^{2k} \cdot \left(\frac{d}{n\abs{\F^*}}\right)^{k/2-\beta/2}\right)^d\mper
\end{equation*}
We finish by the proof of \cref{lem:templatebound}.
 \end{proof}
 
\begin{proof}[Proof of \cref{lem:templatebound}]
    Fix a $t$-template $\tau$ on $[kd]$ coordinates. First we consider the probability $\tau$ is satisfied by a fully random assignment $\phi : [kd] \to [n] \times \F^*$.
    \begin{equation*}
        \Pr[\tau \text{ satisfied by } \phi] = \prod_{i=1}^{kd} \Pr[\text{coord. } i \text{ consistent with } \phi(i) \mid \text{coords. } <i \text{ consistent}]\mper
    \end{equation*}
    We can then case on $\phi(i)$. For a coordinate to be consistent with the ``$-1$'' label, it must necessarily (1) match its associated element's index $i \in [n]$ and (2) negate the partial sum of all previous elements in some index (assuming the partial sum is non-zero). The probability of the first is $\frac{1}{n}$ by uniformity. Given this succeeds, the probability of the second is $\frac{kd}{\abs{\F^*}}$, again by uniformity and the fact that there are at most $kd$ indices with non-zero partial sums. Similarly for those labelled $0$, we need condition (1) and a modification of (2) that we do \textit{not} negate the proposed sum. These probabilities are $\frac{kd}{n}$ and $\frac{\abs{\F^*}-1}{\abs{\F^*}}$ respectively, the latter of which we just bound by $\leq 1$. Given the number of $-1$-labels is $t_1$ and the number of $0$-labels is $t_2$, we can bound the total probability as
    \begin{equation*}
        \Pr[\tau \text{ satisfied by } \phi] 
        \leq \left(\frac{kd}{n\abs{\F^*}}\right)^{t_1} \cdot \left(\frac{kd}{n}\right)^{t_2} 
        \leq \left(\frac{kd}{n\abs{\F^*}}\right)^{t_1} \cdot \left(\frac{k^2d^2}{n^2}\right)^{t_2/2} \leq \left(\frac{k^2 d}{n\abs{\F^*}}\right)^{t/2}\mper
    \end{equation*}

    In this last step, we are crucially using two facts. The first is that $\frac{d}{n^2} \leq \frac{d^2}{n\abs{\F^*}}$ which is implied by our assumption $\ell \leq \frac{n}{\abs{\F^*}}$. We then use that $t_1 + t_2/2 = t/2$ by our definition of a $t$-template.
    
    While a random set $\mcV = \{v_j\}_{j \in [d]}$ does not correspond to a random assignment $\phi$, we can interpret the above estimate as the fraction of all assignments $\phi$ satisfying $\tau$, and crudely restrict the domain to those coming from a random vector set.
    \begin{equation*}
        \Pr\left[\tau \text{ satisfied by } \{v_j\}_{j \in [d]}\right] \leq \Pr\left[\tau \text{ satisfied by } \phi\right] \cdot \frac{\abs{\{\phi \mid \phi : [kd] \to [n] \times \F^*\}}}{\abs{\{\{v_j\}_{j \in [d]} \mid v_j \in \F^n, \abs{v_j} = k\}}}\mper
    \end{equation*}
    The proof is concluded by noting
    \begin{equation*}
        \frac{\abs{\{\phi \mid \phi : [kd] \to [n] \times \F^*\}}}{\abs{\{\{v_j\}_{j \in [d]} \mid v_j \in \F^n, \abs{v_j} = k\}}} = \frac{(n\abs{\F^*})^{kd}}{{n \choose k}^d \abs{\F^*}^{kd}} \leq k^{kd}\mcom
    \end{equation*}
    using standard binomial estimates.
\end{proof}

\end{document}